\documentclass[11pt]{article}
\usepackage[font={footnotesize}]{caption}
\usepackage[margin=1in]{geometry}
\usepackage[utf8]{inputenc}
\usepackage{amssymb}
\usepackage{amsmath}
\usepackage{amsthm}
\usepackage{amsfonts}
\usepackage{algorithm}
\usepackage{algpseudocode}
\usepackage{xcolor}
\usepackage{hyperref}
\usepackage{graphicx}
\usepackage{caption}
\usepackage{subcaption}

\renewcommand{\Pr}[1]{\mathbf{Pr}[#1]}
\newcommand{\PrBig}[1]{\mathbf{Pr}\Big[#1\Big]}

\newcommand{\Var}[1]{\mathbf{Var}\left(#1\right)}
\DeclareMathOperator{\Tr}{Tr}
\DeclareMathOperator{\argmin}{argmin}
\newcommand{\period}{\,.}

\newcommand{\eps}{\varepsilon}
\newcommand{\R}{\mathbb{R}}
\newcommand{\N}{\mathbb{N}}
\newcommand{\poly}{\text{poly}}
\newcommand{\nnz}{\text{nnz}}
\newcommand{\calL}{\mathcal{L}}
\newcommand{\tensorvec}{\text{vec}}
\newcommand{\calD}{\mathcal{D}}
\newcommand{\calF}{\mathcal{F}}
\newcommand{\calE}{\mathcal{E}}
\newcommand{\calS}{\mathcal{S}}
\newcommand{\calG}{\mathcal{G}}
\newcommand{\calQ}{\mathcal{Q}}
\newcommand{\calT}{\mathcal{T}}
\newcommand{\calI}{\mathcal{I}}
\newcommand{\tw}{\text{tw}}

\newtheorem{theorem}{Theorem}
\newtheorem{definition}{Definition}
\newtheorem{lemma}{Lemma}
\newtheorem{problem}{Problem}
\newtheorem{remark}{Remark}
\newtheorem{conjecture}{Conjecture}

\newcommand{\AM}[1]{\textcolor{blue}{#1}}

\begin{document}

\begin{titlepage}

\title{Near-Linear Time and Fixed-Parameter Tractable Algorithms for Tensor Decompositions}
\author{Arvind V. Mahankali\\Stanford University\\amahanka@stanford.edu \and David P. Woodruff\\Carnegie Mellon University\\dwoodruf@cs.cmu.edu \and Ziyu Zhang\\Massachusetts Institute of Technology\\sylziyuz@csail.mit.edu}
\date{}

\maketitle

\begin{abstract}%
We study low rank approximation of tensors, focusing on the Tensor Train and Tucker decompositions, as well as approximations with tree tensor networks and general tensor networks. As suggested by hardness results also shown in this work, obtaining $(1+\varepsilon)-$approximation algorithms for rank $k$ tensor train and Tucker decompositions efficiently may be computationally hard for these problems. Therefore, we propose different algorithms that respectively satisfy some of the objectives above while violating some others within a bound, known as bicriteria algorithms. On the one hand, for rank-$k$ tensor train decomposition for tensors with $q$ modes, we give a $(1 + \varepsilon)$-approximation algorithm with a small bicriteria rank ($O(qk/\varepsilon)$ up to logarithmic factors) and $O(q \cdot \text{nnz}(A))$ running time, up to lower order terms. Here $\text{nnz}(A)$ denotes the number of non-zero entries in the input tensor $A$. We also show how to convert the algorithm of \cite{huber2017randomized} into a relative error approximation algorithm, but their algorithm necessarily has a running time of $O(qr^2 \cdot \text{nnz}(A)) + n \cdot \text{poly}(qk/\varepsilon)$ when converted to a $(1 + \varepsilon)$-approximation algorithm with bicriteria rank $r$. Thus, the running time of our algorithm is better by at least a $k^2$ factor. To the best of our knowledge, our work is the first to achieve a near-input-sparsity time relative error approximation algorithm for tensor train decomposition. Our key technique is a method for efficiently obtaining subspace embeddings for a matrix which is the flattening of a Tensor Train of $q$ tensors - the number of rows in the subspace embeddings is polynomial in $q$, thus avoiding the curse of dimensionality. We extend our algorithm to tree tensor networks and tensor networks on arbitrary graphs.
Another way of coping with intractability is by looking at fixed-parameter tractable (FPT) algorithms. We give FPT algorithms for the tensor train, Tucker, and Canonical Polyadic (CP) decompositions, which are simpler than the FPT algorithms of \cite{swz19_tensor_low_rank}, since our algorithms do not make use of polynomial system solvers. Our technique of using an exponential number of Gaussian subspace embeddings with exactly $k$ rows (and thus exponentially small success probability) may be of independent interest.
\end{abstract}

\setcounter{page}{0}
\thispagestyle{empty}
\end{titlepage}

\section{Introduction}

Data dimensionality reduction has played an important role in numerical linear algebra tasks, such as regression and low rank approximation. Typically in such problems, one has as input a matrix $A \in \mathbb{R}^{n \times d}$, and applies a random linear map $S \in \mathbb{R}^{k \times n}$ for $A$, obtaining $S \cdot A$, where the dimension $k$ is much smaller than $n$. This provides a compression of $A$, and often still retains many of its useful properties, e.g., the row span of $SA$ includes a good low rank approximation to $A$. We refer the reader to \cite{woodruff14_sketching_as_a_tool} for an extensive overview of the applications of sketching for dimensionality reduction.

We investigate the role of randomized dimensionality reduction for tasks involving \textit{tensors} --- in many settings, it is more effective to represent data as a multidimensional tensor $T \in \R^{n \times n \times \ldots \times n}$, rather than as a matrix. This work proposes new algorithms for \textit{tensor low-rank approximation}. In cases where data is represented as a very large matrix $A \in \R^{n \times d}$, significant space savings can be achieved using a low-rank approximation of $A$ --- concretely, if we find matrices $U \in \R^{n \times k}$ and $V \in \R^{k \times d}$ such that $UV \approx A$ and $k \ll n, d$, the space required to represent the information in $A$ is $nk + kd$, which is much less than the space required for $A$ itself. Similarly, significant space savings can be achieved with low rank tensor decompositions of a $q$-mode tensor $A \in \R^{n \times \ldots \times n}$ with $n$ and $q$ large. For instance, a Tensor Train decomposition of $A$ with rank $k$ requires $O(qnk^2)$ parameters. If $k$ is small, this can be much fewer than the $n^q$ parameters required to store $A$.

There are multiple generalizations of the concept of rank to tensors with more modes. Here, we use \textit{mode} to refer to the number of indices required to access an entry of a tensor --- thus, a matrix $A \in \R^{n \times n}$ has 2 modes, and a tensor $A \in \R^{n \times n \times n}$ has 3 modes. We frequently use the letter $q$ to denote the number of modes that an input tensor $A$ has. We investigate Tensor Train, Tucker, and CP decompositions, as well as general tensor networks in this work. These decompositions are introduced shortly below.

Tensor low rank approximations have the potential to allow us to handle high-dimensional data in many applications. Tucker and Tensor Train decompositions have been applied in various fields such as simulation data and robotics (\cite{tt_robotics}), machine learning model compression (\cite{tt_cnn_compress}), scientific computing (\cite{tucker_acs_dfc_app}), and machine learning theory (\cite{tt_bayesian_filter}). Tensor networks are especially suited for solving high-dimensional scientific computing problems~\cite{Khoromskij2012, Dolgov2012, Kazeev2012}. In particular, they have found extensive use in quantum computing (\cite{tt_quantum}). They have also been used recently to solve parametric PDEs~(\cite{Dolgov2015, Dektor2021}), Hamilton-Jacobi-Bellman PDE~(\cite{Horowitz2014, Gorodetsky2018}), and others~(\cite{Khoromskij2012, Kormann2015, Lubasch2018}).

\medskip

The CP decomposition is one notion of low-rank approximation of tensors that we study in this paper. To define the CP decomposition, we first define the outer product of two tensors: \footnotemark \footnotetext{We refer the interested reader to \cite{sdfhpf17_tensor_decomposition_survey} for an overview of tensor decompositions.}

\begin{definition}[Outer Product]
Let $v_1, v_2, \ldots, v_q \in \R^n$. Then, the \textbf{outer product} of $v_1, \ldots, v_q$, denoted $v_1 \otimes \ldots \otimes v_q \in \R^{n \times n \ldots \times n}$, is the $q$-mode tensor whose entry in the index $(i_1, \ldots, i_q)$ is $v_1(i_1) \ldots v_q(i_q)$. More generally, given two tensors $A\in \R^{n_1\times \cdots \times n_p}, B\in \R^{m_1\times \cdots m_q}$, the \textbf{tensor outer product} of $A$ and $B$, denoted $A\otimes B$, is a tensor with dimensions $(n_1, \cdots, n_p, m_1, \cdots, m_q)$. The $(i_1, \cdots, i_p, j_1, \cdots, j_q)^{th}$ entry of $A\otimes B$ is $A_{i_1, \cdots, i_p}B_{j_1, \cdots, j_q}$.
\end{definition}

We can now define the rank-$k$ CP decomposition:

\begin{definition}[Tensor Rank and CP Decomposition] \label{def:tensor_rank_cp_decomposition}
Let $A \in \R^{n \times n \times \ldots \times n}$ be a $q$-mode tensor. We say $A$ has \textbf{CP rank} $k$ if there matrices $U^1, \cdots, U^q \in \R^{n \times k}$ such that $A = \sum_{i = 1}^k U_i^1 \otimes \cdots \otimes U_i^q$. 

For a general tensor $A \in \R^{n \times \ldots \times n}$, rank-$k$ CP decomposition is the problem of finding a tensor $B \in \R^{n \times \ldots \times n}$ with CP rank at most $k$ such that $\|A - B\|_F$ is minimized. In other words, rank-$k$ CP decomposition is the problem of finding
\begin{equation}
\begin{split}
U^1_*, U^2_*, \ldots, U^q_*
= \argmin_{U^1, U^2, \ldots, U^q \in \R^{n \times k}} \Big\| \sum_{i = 1}^k U^1_i \otimes U^2_i \otimes \cdots \otimes U^q_i - A\Big \|_F \period
\end{split}
\end{equation}
Here, given a tensor $A \in \R^{n \times \ldots \times n}$ with $q$ modes, its Frobenius norm is defined as $\|A\|_F = \sqrt{\sum_{i_1, i_2, \ldots, i_q = 1}^n A_{i_1, i_2, \ldots, i_q}^2}$.
\end{definition}

The work of \cite{swz19_tensor_low_rank} previously made a significant advance on sketching for low CP rank decomposition --- this work gave bicriteria and fixed parameter tractable algorithms for low CP rank factorizations under a variety of loss functions, such as the squared Frobenius norm, sum of Frobenius norms of faces, and so on. Here, by \textit{bicriteria}, we mean that the rank of the output can be slightly larger than $k$, though the approximation quality is compared with the best CP rank-$k$ approximation. When considering other notions of tensor rank, e.g. Tensor Train rank, we use the term bicriteria to mean that the Tensor Train rank of the output can be larger than $k$, while we compare to the best approximation of Tensor Train rank at most $k$.

Unfortunately, the best CP low rank approximation is difficult to define, given border rank issues \footnotemark \footnotetext{The border rank of a tensor $A$ is defined as the minimum $k\in \N$ such that $\forall \eps > 0$, there exists a tensor $A' = \sum_{i=1}^k\otimes_{j=1}^q U^{j}_i$ such that $\lVert A - A'\rVert < \eps$. Border rank issues refer to that the CP rank of $A$ is not necessarily equal to the border rank of $A$.}, and even computing the CP rank is NP-hard, which rules out any relative error low rank approximation in polynomial time. Due to these issues, practical work on tensor decomposition has often studied other notions of rank, such as the Tensor Train rank \cite{huber2017randomized, oseledets2011tensor} (Matrix Product State rank) and the Tucker rank \cite{LMV00_multilinear_svd_tucker,kolda2009tensor}. We first define the Tucker rank, and the corresponding problem of low Tucker rank decomposition, using the concepts of Kronecker product and matricization of tensors:

\begin{definition}[Kronecker Product]
Let $A \in \R^{a \times b}$ and $B \in \R^{c \times d}$. Then, their \textbf{Kronecker product} is the matrix $A \otimes B \in \R^{ac \times bd}$ whose entry in row $(i, j)$ and column $(k, l)$ is $A_{ik} B_{jl}$. We also occasionally denote the Kronecker product by $A \times B$.
\end{definition}

\begin{definition}[Vectorization and Matricization of Tensors]
Let $A \in \R^{n \times n \times \ldots \times n}$ be a $q$-mode tensor. Then the $\textbf{vectorization}$ of $A$, denoted $\tensorvec(A) \in \R^{n^q}$, is the vector whose $(1 + \sum_{j = 1}^q (i_j - 1)n^{q - j})^{th}$ entry is $A(i_1, \ldots, i_q)$. The \textbf{mode-$t$ matricization} of $A$, which we denote by $M_t(A) \in \R^{n \times n^{q - 1}}$, is the matrix whose $j^{th}$ row is $\tensorvec(A(:, \ldots, :, j, :, \ldots, :))$, where $A(:, \ldots, :, j, :, \ldots, :)$ denotes the slice of $A$ whose index in the $t^{th}$ mode is $j$.

More generally, for $m < q$, the matricization $M_{i_1, \cdots, i_m}(A)$ of a $q$-mode tensor $A\in \R^{n \times n \times \ldots \times n}$ is the $n^m \times n^{q-m}$ matrix whose $\left(\sum_{t=1}^m j_t n^{m-t}\right)^{th}$ row is the vectorization of the slice of $A$ whose index in the $i_t^{th}$ mode is $j_t$. For instance, given a $q$-mode tensor $A \in \R^{n \times \ldots \times n}$ for $q \geq 3$, the matricization $M_{1, 2}(A)$ is the $n^2 \times n^{q - 2}$ matrix whose $(j_1 n + j_2)^{th}$ row is $\tensorvec(A(j_1, j_2, \ldots))$. Recall that $A(j_1, j_2, \ldots)$ denotes the slice of $A$ containing the entries whose entries in the first and second modes are $j_1$ and $j_2$ respectively.
\end{definition}

Thus, we can define the Tucker rank of a tensor:

\begin{definition}[Tucker Rank]
\label{def:tucker_decomposition}
Let $A \in \R^{n \times n \times \ldots \times n}$ be a $q$-mode tensor. We say $A$ has Tucker rank at most $(k_1, \ldots, k_q)$ if there exist $U^{1} \in \R^{n \times k_1}, U^{2} \in \R^{n \times k_2}, \ldots, U^{q} \in \R^{n \times k_q}$, and a $q$-mode tensor $G \in \R^{k_1 \times k_2 \times \ldots \times k_q}$, such that $A = \sum_{i_1 \in [k_1], i_2 \in [k_2], \ldots, i_q \in [k_q]} G(i_1, \ldots, i_q) U^{1}(:, i_1) \otimes \ldots \otimes U^{q}(:, i_q) = G\times_1 U^{1}\times_2\cdots\times_q U^{q}$, with $\times_j$ being the $j^{th}$-mode product. As noted in \cite{sdfhpf17_tensor_decomposition_survey}, this can be rewritten as $\tensorvec(A) = (U^{1} \times \ldots \times U^{q})\tensorvec(G)$, or as $M_t(A) = U^{t} M_t(G) (U^{1} \times \ldots U^{t - 1} \times U^{t + 1} \times \ldots U^{q})^T$. This notion of rank is also referred to in the literature as a \textbf{multilinear rank}.
\end{definition}

One of the main problems that we study in this paper is that of approximating an arbitrary tensor $A$ by another tensor $B$ which has a low Tucker rank. For simplicity, we consider a special case for the tuple $(k_1, \ldots, k_q)$ in the definition above:

\begin{problem}[Tucker-$(p, q)$ Decomposition]
Let $A \in \R^{n \times n \ldots \times n}$ be a $q$-mode tensor, and $k \in \N$. Then, we wish to find a $q$-mode tensor $B \in \R^{n \times n \times \ldots \times n}$ for which $\|B - A\|_F^2$ is as small as possible, subject to the constraint that $B$ has multilinear rank at most $(k, \ldots, k, n, \ldots, n)$ (where this tuple has $k$ in the first $p$ coordinates and $n$ in the last $q - p$ coordinates). In other words, the first $p$ factors in the Tucker decomposition of $B$ have $k$ columns, and the last $q - p$ factors have $n$ columns. Note that without loss of generality, this means that the last $q - p$ factors can be taken to be the $n \times n$ identity matrix. We refer to the Tucker-$(q, q)$ decomposition problem (i.e., the case where the tensor $B$ has multilinear rank at most $(k, k, \ldots, k)$) as the Tucker-$q$ decomposition problem.
\end{problem}

We next define the Tensor Train rank, based on the concept of tensor contractions.

\begin{definition}[Tensor Contraction $\circ_{i, j}$]
Let $A\in \R^{n_1\times \cdots \times n_p}$ be a $p-$mode tensor and $B\in \R^{m_1\times \cdots \times m_q}$ a $q-$mode tensor. Assuming that $n_i = m_j$, their \textbf{tensor inner product} or \textbf{tensor contraction} $A\circ_{i, j}B$ is defined as $\sum_{u=1}^{u=n_i=m_j}A_{:, \cdots, u, \cdots, :}\otimes B_{:, \cdots, u, \cdots, :}$, where $u$ is an index for the $i^{th}$ mode of $A$ and $j^{th}$ mode of $B$. The operation $\circ_i$ is short for $A\circ_{p, i}B$ and the operation $\circ$ is short for $A\circ_{p, 1}B$. For example, when $p=q=2$, $A\circ B$ corresponds to matrix multiplication.

Additionally, we can take the tensor contraction for several modes at once, as follows. Let $A \in \R^{n_1 \times \ldots \times n_p}$ be a $p$-mode tensor and $B \in \R^{m_1 \times \ldots \times m_q}$ a $q$-mode tensor. Then, given two tuples $(i_1, \ldots, i_d)$ and $(j_1, \ldots, j_d)$ of modes (and assuming $n_{i_r} = m_{j_r}$ for all $r \in [d]$), we can define the contraction $A \circ_{(i_1, \ldots, i_d), (j_1, \ldots, j_d)} B$ as $\sum_{k_1 = 1}^{n_{i_1}} \sum_{k_2 = 1}^{n_{i_2}} \cdots \sum_{k_d = 1}^{n_{i_d}} A(i_1 = k_1, \ldots, i_d = k_d) \otimes B(j_1 = k_1, \ldots, j_d = k_d)$. Here we have used $A(i_1 = k_1, \ldots, i_d = k_d)$ to denote the slice of $A$ in which the index for mode $i_1$ is $k_1$, and so on.
\end{definition}

\begin{remark}
This operation is referred to in the literature as tensor contraction because it is equivalent to the self-contraction at modes $(i, p+j)$ of $A\otimes B$.
\end{remark}

We now define Tensor Train rank, and the corresponding problem of rank-$k$ Tensor Train decomposition, which is another of the main problems that we study in this paper.

\begin{definition}[Tensor Train Rank and Tensor Train Decomposition (introduced by \cite{oseledets2011tensor})]
\label{def:tensor_train_decomposition_problem}
Let $A \in \R^{n \times \ldots \times n}$ be a $q$-mode tensor for some $q \in \N$. We say $A$ has \textbf{Tensor Train rank} $(k_1, \ldots, k_{q - 1})$ for some $k_1, \ldots, k_{q - 1} \in \N$ if there exist $U_1 \in \R^{n \times k_1}$, $U_q \in \R^{k_{q - 1} \times n}$, and $U_i \in \R^{k_{i - 1} \times n \times k_i}$ for $i \in [q] \setminus \{1, q\}$, such that $A = U_1 \circ \cdots \circ U_q$. We define the problem of \textbf{Tensor Train decomposition} as follows. Let $A \in \R^{n \times \ldots \times n}$ be a $q$-mode tensor for some $q \in \N$, and let $k \in \N$. Then, we wish to solve
$$\min_{U_1 \in \R^{n \times k}, U_q \in \R^{k \times n}, U_i \in \R^{k \times n \times k}} \|U_1 \circ \cdots \circ U_q - A\|_F$$
In other words, we wish to find a tensor $B$ of Tensor Train rank $(k, \ldots, k)$ such that $\|A - B\|_F$ is minimized.
\end{definition}

Although the formulation of the Tensor Train Decomposition approximation problem in this work uses the original $q-$mode tensor $A$ as the input, our subspace embedding and algorithm could be adapted to perform dimensionality reduction when the input is given in the Tensor Train format with a greater rank.

\subsection{Our Contributions} \label{subsec:contributions}

As mentioned in the last section, even computing the CP rank is NP-hard. However, for both the Tensor Train and Tucker decompositions, an $O(\sqrt{q})$-relative error low rank approximation in Frobenius norm is computable in polynomial time (see, e.g., the discussion and references in \cite{huber2017randomized}), thus suggesting that there are more efficient low rank approximation algorithms for these notions of rank than for the CP rank. However, as suggested by hardness results in this work, it is difficult to achieve both $(1+\eps)$ relative error and rank exactly $k$ in polynomial time. Thus, we present bicriteria algorithms in both directions.

We obtain the first polynomial-time $(1 + \eps)$-approximation algorithm for the Tensor Train decomposition with any non-trivial bicriteria rank. In particular, our bicriteria rank is  $O(\frac{qk}{\epsilon} \log(\frac{q}{\delta}))$.

\begin{theorem}[Special Case of Theorem \ref{thm:analysis_of_tree_network_bicriteria} for Tensor Train Decomposition]
Let $A \in \R^{n \times \ldots \times n}$ be a $q$-mode tensor. Then, Algorithm \ref{algorithm:tensor_train_bicriteria} outputs the factors $U^1, \ldots, U^q$ of a tensor $M$ with Tensor Train rank at most $O(\frac{qk}{\epsilon} \log (\frac{q}{\delta}))$ such that
$$\|M - A\|_F \leq (1 + \epsilon) \min_{T} \|T - A\|_F$$
where the minimum on the right-hand side is taken over all tensors $T$ with Tensor Train rank at most $k$. The running time of our algorithm is $O(q \cdot \nnz(A))$ up to lower-order terms.
\end{theorem}

Our analysis is novel compared to \cite{huber2017randomized} (which proposed an $O(\sqrt{q})$-approximation algorithm for Tensor Train decomposition) as we introduce new techniques for efficiently computing subspace embeddings for matrices which are the sequential contraction of several tensors. We further discuss our analysis in Subsection \ref{subsec:contributions}. 
Compared to a tensor with Tensor Train rank $k$, which requires $O(qnk^2)$ parameters to express, the tensor output by our algorithm requires only $\widetilde{O}(\frac{q^3 nk^2}{\epsilon^2})$ parameters. Thus, the number of parameters we obtain has an optimal dependence on $n$ and $k$.

The other direction where we make contributions for these same decompositions is simpler fixed-parameter tractable algorithms for exact rank $k$ approximations under the Tensor Train and Tucker decompositions. The work of \cite{swz19_tensor_low_rank} gives $(1 + \varepsilon)$-approximation algorithms for the Tensor Train and Tucker decompositions, with output rank exactly $k$ and 
running time $2^{\poly(k/\epsilon)}$ (ignoring the dependence on $q$ and $\nnz(A)$). However, the algorithms of \cite{swz19_tensor_low_rank} have the following drawbacks: (1) the $\poly(k/\varepsilon)$ factors in the $2^{\poly(k/\varepsilon)}$ time are large and perhaps suboptimal, and (2) their algorithms run polynomial system solvers for deciding the existential theory of the reals (see, e.g., \cite{polyDecider}), which was shown to be equivalent to deciding CP-rank \cite{ss16_tensor_rank_polynomial_systems}. Note that generic polynomial system solvers are highly complex, raising the question of whether simpler fixed-parameter tractable algorithms for the Tucker and Tensor Train decompositions can be obtained. For (1), our fixed-parameter tractable algorithm for Tensor Train decomposition improves the dependence on the $\poly(k/\epsilon)$ factor in the exponent by an $O(k^2)$ factor. For (2), our algorithms for both the Tensor Train and Tucker decompositions only make use of dense Gaussian sketching matrices rather than polynomial system solvers. We note that in addition to \cite{swz19_tensor_low_rank}, there are also  works, see, e.g., \cite{fahrbach2021fast_ridge_regression,faster_tt_sparse_2019,LMV00_multilinear_svd_tucker,lmv00_best_multilinear_rank_approximation, TensorSketch,  MinsterRachel2020RAfL, oseledets2011tensor, huber2017randomized} that compute low rank approximations in polynomial time, but only provide additive error guarantees or take a prohibitive amount of time. We state the guarantees that our fixed-parameter tractable algorithms achieve:

\begin{theorem} \label{thm:strongly_fpt_time_algorithm_tucker}
There is an algorithm for Tucker-$(p, q)$ decomposition (Algorithm \ref{algorithm:fpt_rank_exactly_k_1_plus_eps_approximation}) which, given a $q$-mode tensor $A \in \R^{n \times n \times \ldots \times n}$, and $k \in \N$, outputs $U^1, \ldots, U^p \in \R^{n \times k}$ for which, with probability at least $\frac{4}{5}$,
\begin{equation}
\begin{split}
\min_{G} \|(U^1 \times \ldots \times U^p \times I_n & \times \ldots \times I_n)  \tensorvec(G) - \tensorvec(A) \|_F \\
& \leq (1 + \eps) \min_{U^1, \ldots, U^p, G} \|(U^1 \times \ldots \times U^p \times I_n \times \ldots \times I_n) \tensorvec(G) - \tensorvec(A)\|_F
\end{split}
\end{equation}
with running time $O(p \cdot \nnz(A)) + n \cdot (\frac{pk}{\eps})^{O(\frac{p^2 k^2 \log p}{\eps})}$.
\end{theorem}

\begin{theorem} \label{thm:fpt_tensor_train_rank_exactly_k}
There is an algorithm for Tensor Train decomposition (Algorithm \ref{algorithm:fpt_rank_exactly_k_tensor_train}) which, given a $q$-mode tensor $A \in \R^{n \times \ldots \times n}$, and $k \in \N$, outputs $U^1 \in \R^{n \times k}$, $U^2, \ldots, U^{q - 1} \in \R^{k \times n \times k}$, and $U^q \in \R^{k \times n}$, such that with probability at least $\frac{2}{3}$,
$$\|U^1 \circ \cdots \circ U^q - A\|_F \leq (1 + \eps) \min_{U^1, \ldots, U^q} \|U^1 \circ \cdots \circ U^q - A\|_F$$
with running time $O(q \cdot \nnz(A)) + n \cdot \poly(\frac{qk}{\eps \delta}) \cdot e^{\Theta(\frac{q^2 k^2}{\eps} \log(\frac{qk}{\eps}))}$ and polynomial space.
\end{theorem}

\begin{algorithm}
\caption{$(1 + \eps)$-approximation algorithm for Tensor Train decomposition with output rank $t = O(\frac{qk}{\eps} \log(\frac{q}{\delta}))$.}
\label{algorithm:tensor_train_bicriteria}
\begin{algorithmic}
\Require $A \in \R^{n \times \ldots \times n}$ with $q$ modes.
\Ensure $U^1 \in \R^{n \times k}$, $U^2, \ldots, U^{q - 1} \in \R^{k \times n \times k}$, $U^q \in \R^{k \times n}$

\item []
\State // Compute $U^1$, processing the first mode, first
\State $T_1 \gets $ An $O(\frac{q^3 k^2}{\eps^2 \delta}) \times n^{q - 1}$ Countsketch matrix
\State $R_1 \gets $ A $t \times O(\frac{q^3 k^2}{\eps^2 \delta})$ matrix whose entries are drawn i.i.d. from $\{-\frac{1}{\sqrt{t}}, \frac{1}{\sqrt{t}}\}$
\State $U^1 \gets M_1(A) T_1^T R_1^T \in \R^{n \times t}$

\item []
\State // Store a sketch of $U^1$, which we denote $M_1$, for future use.
\State $S_1 \gets $ An $s \times n$ Countsketch matrix, where $s = O(\frac{q^4 t^2 d^3}{\eps^2 \delta}) = O(\frac{q^6 k^2 d^3}{\eps^4 \delta} \log^2(\frac{q}{\delta}))$.
\State $M_1 \gets S_1 U^1 \in \R^{s \times t}$
\State $A \gets S_1 \circ_1 A \in \R^{s \times n \times \ldots \times n}$, i.e. a $q$-mode tensor where the last $q - 1$ modes have dimension $n$.

\item []
\State // Now process modes $2$ through $q$. At the beginning of the $i^{th}$ iteration, $A$ has $q - i + 2$ modes, 
\State // where the first has dimension $s$ and the rest have dimension $n$.
\For{$i = 2, \ldots, q$}
    \State // First compute $U^i$ using $M_{i - 1}$ which is a sketch of $U^1 \circ \cdots \circ U^{i - 1}$.
    \State $U^i \gets M_{i - 1}^\dagger \circ_1 A \in \R^{t \times n \times \ldots \times n}$, with $q - i + 2$ modes.
    \If{$i$ is not $q$}
     \State $T_i \gets $ An $O(\frac{q^3 k^2}{\epsilon^2 \delta}) \times n^{q - i}$ Countsketch matrix (i.e. an $O(\frac{q^3 k^2}{\epsilon^2 \delta}) \times n \times \ldots \times n$ tensor)
     \State $R_i \gets $ A $t \times O(\frac{q^3 k^2}{\epsilon^2 \delta})$ matrix whose entries are drawn i.i.d. from $\{-\frac{1}{\sqrt{t}}, \frac{1}{\sqrt{t}}\}$
     \State // Note that in the following we represent $R_i T_i$ as a $t \times n \times \ldots \times n$ tensor.
     \State $U^i \gets U^i \circ_{(3, \ldots, q - i + 2), (2, \ldots, q - i + 1)} (R_i T_i) \in \R^{t \times n \times t}$
     \item []
     \State // Next compute $M_i$
     \State $M_i \gets M_{i - 1} \circ U_i \in \R^{s \times n \times t}$
     \State $S_i \gets $ An $s \times sn$ Countsketch matrix, represented as $s \times s \times n$ tensor
     \State $M_i \gets S_i \circ_{(2, 3), (1, 2)} M_i \in \R^{s \times t}$
     \State $A \gets S_i \circ_{(2, 3), (1, 2)} A \in \R^{s \times n \times \ldots \times n}$
    \EndIf
\EndFor \\
\\

\Return $U^1, \ldots, U^q$
\end{algorithmic}
\end{algorithm}

\begin{remark}
Throughout we assume our input tensor has all modes of the same dimension $n$. This is for presentation purposes only and our techniques can straightforwardly handle tensors where modes have differing dimensions with minor modifications.
\end{remark}

To summarize, we propose a polynomial-time bicriteria algorithm, using a novel sketch for tensor contractions, for $(1 + \eps)$-approximate Tensor Train decomposition, obtaining a bicriteria rank of $O(\frac{qk}{\epsilon} \log(\frac{q}{\delta}))$. In addition, we give fixed-parameter tractable $(1 + \eps)$-approximation algorithms for the Tucker, Tensor Train and CP decompositions which do not depend on polynomial system solvers.

We also include fine-grained hardness results for rank-$1$ Tucker-$(2, 3)$ decomposition, rank-$1$ CP decomposition, and Tensor Train decomposition with $q = 3$, which lower-bound the optimal dependency in terms of $1/\eps$ which any algorithm with output rank $k$ can achieve. Additionally, we consider further generalizations of our bicriteria algorithm for Tensor Train decomposition, to other notions of tensor rank based on \textit{tensor networks}. Lastly, we obtain bicriteria algorithms for Tucker-$(p, q)$ decomposition with a robust loss function.

\subsubsection{Bicriteria Algorithm for Tensor Train Decomposition} \label{subsubsec:tensor_train_bicriteria}

We state known theoretical guarantees for polynomial time Tensor Train decomposition in Table \ref{tab:tensor_train_polynomial_time_results}, along with our result. We stress that no polynomial time relative error $(1 + \eps)$-approximations were known --- even for the case of Tensor Train decomposition, previous work either obtained additive error or $O(\sqrt{q})$-approximation.

\begin{table}
    \centering
    \resizebox{\textwidth}{!}{
    \begin{tabular}{|c|c|c|c|}
    \hline
    Work & Running Time & Approximation Factor & Rank \\ \hline
    \cite{oseledets2011tensor} (TT-SVD) & $n^{O(q)}$ & $\sqrt{q - 1}$ & $k$ \\ \hline
    \cite{huber2017randomized} & $O(q((k+p)^2\nnz(A)+(k+p)^3n))$  & $\sqrt{q-1}\left(1+O(\sqrt{\frac{12k}{p}})+O(\frac{e\sqrt{k+p}}{p+1})\right)$ & $r=k+p$    \\ \hline
    \textbf{This work} & $O(q \cdot \nnz(A)) + n \cdot \poly(qk/\eps)$ & $(1 + \eps)$ &  $O(qk/\eps \log(q/\delta))$    \\ \hline
    \end{tabular}}
    \caption{Known algorithms for Tensor Train decomposition. In the algorithm of \cite{huber2017randomized}, $p$ is a parameter for oversampling on the output bicriteria rank. We can extend their algorithm to a $(1 + \eps)$-approximation algorithm with running time $qr^2 \cdot \nnz(A) + n \cdot \poly(qk/\eps)$, where $r$ is the desired bicriteria rank. The leading order term in the running time is significantly slower than that of our algorithm.}
    \label{tab:tensor_train_polynomial_time_results}
\end{table}

\paragraph{Subspace Embeddings for Tensor Contractions}

The key component in our bicriteria algorithms is a \textit{subspace embedding for matrices which are implicitly defined in terms of tensor contractions} --- in other words, we give a technique for obtaining a subspace embedding of matrices of the form $M_{\{1, \ldots, i - 1\}}(U^1 \circ \cdots \circ U^{i - 1}) \in \R^{n^{i - 1} \times k}$, with $n \cdot \poly(ik/\eps)$ running time, without computing all $k \cdot n^{i - 1}$ entries of this matrix. Recall that, as in \cite{woodruff14_sketching_as_a_tool}, given a matrix $A \in \R^{n \times d}$, a $(1 \pm \eps)$ subspace embedding of $A$ is a matrix $S \in \R^{s \times n}$ such that $\|SAx\|_2 = (1 \pm \eps) \|Ax\|_2$ for all $x \in \R^d$.

We can obtain such a subspace embedding for $M_{\{1, \ldots, i - 1\}}(U^1 \circ \cdots \circ U^{i - 1})$ as follows. We wish to construct a linear map $\mathcal{L} : \R^{n^{i - 1}} \to \R^s$, for $s = \poly(ik/\eps)$, such that for all $x \in \R^k$,
$$\|\mathcal{L} M_{\{1, \ldots, i - 1\}}(U^1 \circ \cdots \circ U^{i - 1}) x\|_2 = (1 \pm \eps) \|M_{\{1, \ldots, i - 1\}}(U^1 \circ \cdots \circ U^{i - 1}) x\|_2$$
However, for $x \in \R^k$, the entries of $M_{\{1, \ldots, i - 1\}}(U^1 \circ \cdots \circ U^{i - 1})x$ are in one-to-one correspondence with those of the $(i - 1)$-mode tensor $U^1 \circ \cdots \circ U^{i - 1} \circ x \in \R^{n \times \ldots \times n}$. We can thus significantly reduce the dimension of $U^1 \circ \cdots \circ U^{i - 1} \circ x$, while preserving its norm, for all $x \in \R^k$, as follows. First let $S_1 \in \R^{s_1 \times n}$ be a $(1 \pm \eps/q)$ subspace embedding for $U^1$ --- then, for all $x \in \R^k$, we have the equality
\begin{equation}
\begin{split}
\|S_1 \circ U^1 \circ U^2 \circ \cdots \circ U^{i - 1} \circ x\|_2
& = \|S_1 U^1 M_1(U^2 \circ \cdots \circ U^{i - 1} \circ x) \|_2 \\
& = (1 \pm \eps/q) \|U^1 M_1(U^2 \circ \cdots \circ U^{i - 1} \circ x)\|_2 \\
& = (1 \pm \eps/q) \|U^1 \circ U^2 \circ \cdots \circ U^{i - 1} \circ x\|_2
\end{split}
\end{equation}
\begin{figure}[h]
\caption{Illustration for Subspace Embeddings for Tensor Contractions}
\centering
\includegraphics[width=0.5\textwidth]{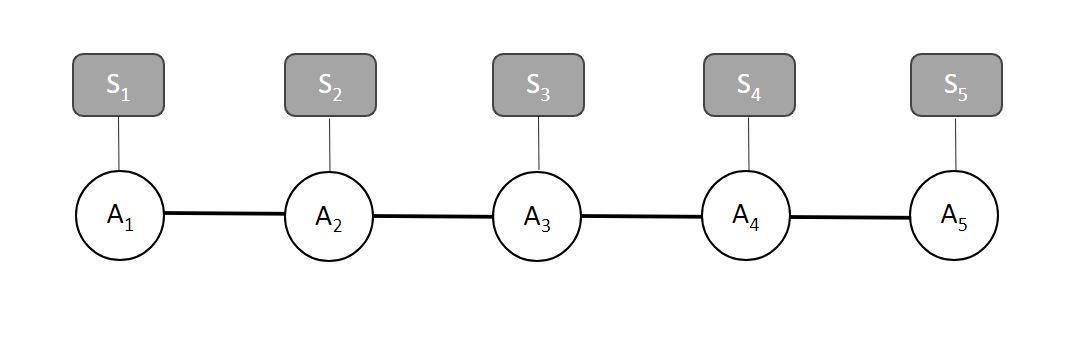}
\end{figure}
In addition, note that
$$S_1 \circ U^1 \circ \cdots \circ U^{i - 1} \circ x = (S_1 U^1) \circ U^2 \circ U^3 \circ \cdots \circ U^{i - 1} \circ x$$
and $S_1 U^1 \in \R^{s_1 \times k}$ can be computed in $O(s_1 nk)$ time, while $(S_1 U^1) \circ U^2 \in \R^{s_1 \times n \times k}$ can be computed in $O(s_1 n k^2)$ time. Now, define $V^2 \in \R^{s_1 n \times k}$ by $V^2 = M_{\{1, 2\}}((S_1 U^1) \circ U^2)$ --- then,
$$\|V^2 \circ U^3 \circ \cdots \circ U^{i - 1} \circ x\|_2 = \|(S_1 U^1) \circ U^2 \circ U^3 \circ \cdots \circ U^{i - 1} \circ x\|_2 = (1 \pm \eps/q) \|U^1 \circ \cdots \circ U^{i - 1} \circ x\|_2$$
and thus the next step is to apply a subspace embedding $S_2 \in \R^{s_2 \times n}$ to $V^2$. In general, we proceed iteratively --- if we have computed $V^j \in \R^{s_{j - 1} n \times k}$, such that for all $x \in \R^k$,
$$\|V^j \circ U^{j + 1} \circ \cdots \circ U^{i - 1} \circ x\|_2 = (1 \pm \eps/q)^{j - 1} \|U^1 \circ \cdots \circ U^{i - 1} \circ x\|_2$$
then we can first compute $S_j V^j$ where $S_j \in \R^{s_j \times s_{j - 1} n}$ is a subspace embedding for $V^j$, followed by computing $V^{j + 1} := M_{\{1, 2\}}((S_j V^j) \circ U^{j + 1}) \in \R^{s_j \times k}$ --- we then find that
$$\|V^{j + 1} \circ U^{j + 2} \circ \cdots \circ U^{i - 1} \circ x\|_2 = (1 \pm \eps/q)^j \|U^1 \circ \cdots \circ U^{i - 1} \circ x\|_2$$
At the end of this procedure, we will have computed $V^{i - 1} \in \R^{s_{i - 2}n \times k}$ such that for all $x \in \R^k$,
$$\|V^{i - 1} x\|_2 = (1 \pm \eps/q)^{i - 1} \|U^1 \circ \cdots \circ U^{i - 1} \circ x\|_2$$
and if we let $W_{i - 1} = S_{i - 1} V^{i - 1} \in \R^{s_{i - 1} \times k}$ where $S_{i - 1} \in \R^{s_{i - 1} \times s_{i - 2} n}$ is a subspace embedding matrix, we find that
$$\|W_{i - 1} x\|_2 = (1 \pm O(\eps)) \|U^1 \circ \cdots \circ U^{i - 1} \circ x\|_2$$
and the overall computation is $O(q s nk \cdot (s + k))$. We can choose $s_j = \poly(qk/\eps)$, meaning the final matrix $W_{i - 1}$ has dimensions $\poly(qk/\eps) \times k$.  We note that our subspace embedding for tensor contractions may be of independent interest, and is part of a growing body of work on obtaining subspace embeddings for matrices that are only represented implicitly \cite{ASW13,SW19,JSWY21,MMMW22}.

In Section \ref{subsection:improving_hsw_17}, we discuss how to apply this technique to improve on previous work of \cite{huber2017randomized} and obtain a bicriteria Tensor Train decomposition algorithm. We show that the algorithm of \cite{huber2017randomized} can be recast as a $(1 + \eps)$-approximation algorithm for Tensor Train decomposition, and show how our subspace embedding can be used to reduce their running time to $O(q \cdot \nnz(A))$, together with lower-order terms. We note that our analysis is significantly different from \cite{huber2017randomized}, since we give a technique for obtaining a subspace embedding of a matricization of a Tensor Train. In addition, our subspace embedding does not follow directly from the subspace embedding of \cite{akkpvwz20_tensor_sketch_outer_product} --- the subspace embedding of \cite{akkpvwz20_tensor_sketch_outer_product} applies only to tensors with low CP rank, and would require a sketch size of at least $k^q$ if applied to Tensor Train decomposition.

The way we have described the above subspace embedding for a tensor contraction is sequential. One could instead matricize each mode in parallel, obtaining a matrix with $n$ rows and rank $k^2$ for each internal mode matricization. One could then build a binary tree of sketches, fusing and sketching two modes at a time for the internal nodes of the tree. While this could help with multiple processors, the concrete $\varepsilon, k,$ and $q$ factors in the bicriteria rank in our Tensor Train application in Section \ref{subsection:improving_hsw_17} are a bit worse. 

We extend our techniques to obtain a $(1 + \epsilon)$-approximation algorithms for decomposing a tensor according to general tree networks as well. (See \cite{clopzm16_tensor_networks} for a survey of more general tensor networks, which generalize the Tucker and Tensor Train ranks.) We do this using a dynamic programming approach, by processing the tensors from the leaves to the root (see \cite{grasedyck10_hierarchical_tucker} for a prior application of this approach to tree tensor networks).

\subsubsection{Fixed Parameter Tractable Algorithms}

We give fixed parameter tractable $(1 + \varepsilon)$-approximation algorithms for CP decomposition, Tucker decomposition, and Tensor Train decompositions. 

\begin{table}
    \centering
    \begin{tabular}{|c|c|c|c|}
    \hline
    Work & Running Time & Decomposition & Bicriteria Rank \\ \hline
    \cite{swz19_tensor_low_rank} & $q \cdot \nnz(A) + n\cdot \poly(qk/\eps)$ & CP & $O((k/\eps)^{q-1})$ \\ \hline
    \cite{swz19_tensor_low_rank} & $(q \cdot \nnz(A) + n \poly(qk/\eps) +\exp{(qk^2/\eps)})\cdot n^{\delta}$ & CP & $k$ \\ \hline
    \cite{swz19_tensor_low_rank} & $\nnz(A) + n \cdot \poly(k, 1/\varepsilon) + 2^{O(k^2/\eps + k^3)}$ & Tucker-$(3, 3)$ & $k$ \\ \hline
    \cite{swz19_tensor_low_rank} & $\nnz(A) + n \cdot \poly(k, 1/\eps) + 2^{O(k^4/\eps)}$ & TT ($q = 3$) & $k$ \\ \hline
    \textbf{This work} & $O(p\cdot \nnz(A))+n\cdot\left(\frac{pk}{\eps}\right)^{O(p^2k^2\log p/\eps)}$ & Tucker-($p, q$) &  $k$    \\ \hline
    \textbf{This work} & $O(q \cdot \nnz(A)) + n \cdot \poly(\frac{qk}{\eps \delta}) \cdot e^{\Theta(\frac{q^2 k^2}{\eps} \log(\frac{qk}{\eps}))}$ & TT &  $k$    \\ \hline
    \textbf{This work} & $\left(q\cdot \nnz(A) + n\poly(k, q/\eps)\right)^{\poly(kq/\eps)}$ & CP &  $k^{q-1}$    \\ \hline
    \end{tabular}
    \caption{Known algorithms for fixed parameter tractable $(1+\eps)-$approximation for different decompositions --- here TT refers to Tensor Train. For each of these decompositions, we remove the use of generic polynomial system solvers that were used in \cite{swz19_tensor_low_rank}. Here $s=\poly(\frac{qk}{\eps}\exp(\poly(qk/\eps)))$ in our Tensor Train algorithm.}
    \label{tab:fpt_results}
\end{table}

The idea behind our algorithms is to generate an exponential (in $k, q,$ and $1/\varepsilon$) number of guesses for the factor matrices and take the best solution found by sketching. Analyzing such algorithms requires understanding common primitives, such as $\ell_2$-regression, subspace embeddings, and approximate matrix product, but in a regime that has not been studied before, namely, where the success probability is exponentially small. The reason the success probability is so small is that we need to choose sketching matrices with exactly $k$ rows (as opposed to say, $\poly(k/\varepsilon)$) in order to ensure that our output has rank exactly $k$. Since the success probability of our primitives is $2^{-\poly(kq/\varepsilon)}$, by repeating $2^{\poly(kq/\varepsilon)}$ times independently, one of the solutions found by sketching will provide a $(1+\varepsilon)$-approximation.

Since we just enumerate over our guesses, we avoid the polynomial system solvers used in \cite{swz19_tensor_low_rank}. Our algorithms are conceptually simpler and easier to implement (e.g., the algorithm proposed in \cite{polyDecider} relies on real algebraic geometry). Our techniques may be useful for other linear algebra problems where known algorithms use polynomial system solvers, such as weighted low rank approximation \cite{bwz19_weighted_lra}, and non-negative matrix factorization \cite{agkm_16_nmf, moitra14_almost_optimal_nonnegative_rank}.

Perhaps of interest independent of our FPT algorithms is Theorem \ref{thm:approximate_multiple_l2_regression}, which shows
that for the same optimal sketching dimension \cite{cw09_numerical_linear_algebra_streaming} considered in
previous work for approximate regression, one can solve
for an {\it approximate} minimizer in the sketch space,
and argue this is a $(1+\varepsilon)$-approximate solution
in the original space. Previous work surprisingly could
only show this for solving for the exact solution in the sketch space. We use this in our Tensor Train FPT algorithm for finding candidate guesses. More details on the techniques and results for each decomposition are summarized below.

\paragraph{Tucker FPT Algorithm.} 
A key fact that guarantees the success probability of our FPT Algorithm (shown in Algorithm \ref{algorithm:fpt_rank_exactly_k_1_plus_eps_approximation}
) for $(1\pm\eps)$ rank-$k$ approximation under the Tucker-$(p, q)$ decomposition is Lemma \ref{lemma:gaussian_with_exactly_k_rows_multiple_response_regression}, reproduced below:

\begin{lemma}[Lemma \ref{lemma:gaussian_with_exactly_k_rows_multiple_response_regression}]
Let $n, k, s, t\in \N, \eps > 0,$ and $k, t < n$. Suppose $A\in \R^{n\times k}$ has rank $k$ and $B\in \R^{n\times s}$ has rank $t$. Let $R\in \R^{k\times n}$ have i.i.d. $\mathcal{N}(0, 1/k)$ entries. If $X^* = \argmin_X \| AX - B\|_F^2$ and $\widehat{X} = \argmin_X \|RAX - RB\|_F^2$, then $\|A\widehat{X} - B\|_F^2\leq (1+\eps)\|AX^* - B\|_F^2$ with probability at least $e^{-\Theta(k^2\log k)}(\frac \eps k)^{O(kt)}$.
\end{lemma}

As we range over the guesses of $U$ produced by our algorithm, at least one guess is a \textit{good} guess, meaning it provides a $(1\pm \frac{\eps}{p^2})$ approximation. We then use an $\ell_2$ Kronecker Product Regression result of \cite{djssw19_kronecker_product_regression} together
with projection cost-preserving sketches
\cite{cemmp15_pcps}, to 
efficiently evaluate the cost of each guess. 

\begin{algorithm}
\caption{Fixed-parameter tractable algorithm for obtaining a $(1 + \eps)$-approximate rank $k$ solution for Tucker-$(p, q)$ decomposition. Here, $\textsc{ErrorEstimate}$ denotes the algorithm referred to in Theorem \ref{thm:kronecker_product_regression} (which is the algorithm of \cite{djssw19_kronecker_product_regression} with certain parameters modified) --- $\textsc{ErrorEstimate}(U^1, \ldots, U^p, M_j, \delta)$ returns an estimate of the error $\min_x \|(U^1 \times \ldots \times U^p)x - M_j\|_2$, with failure probability $\delta$. The variable $\textsc{NewError}$ is referred to as $\widehat{e}$ in the analysis. Note that in order to output the core tensor, we would incur a $n^{q - p}$ term in the running time. In the case $p = q$, we can also output a near-optimal core tensor, by using the algorithm of \cite{djssw19_kronecker_product_regression}, while still achieving the desired running time. In addition, in the case $p = q$, it is not necessary to multiply $A_p$ by a PCP before performing Kronecker product regression.}
\label{algorithm:fpt_rank_exactly_k_1_plus_eps_approximation}
\begin{algorithmic}
\Require A $q$-mode tensor $A \in \R^{n \times \ldots \times n}$, $p \leq q$, $k \in \N$, $\eps \in (0, 1)$
\Ensure $U^1, U^2, \ldots, U^p \in \R^{n \times k}$

\item []
\State // Generating guesses for $U^m$ for $m \in [p]$
\State $T \gets (\frac{pk}{\eps})^{O(\frac{pk^2 \log p}{\eps})}$ \\
\For{$m = 1 \to p$}
    \State $S_m \gets $ An $O(\frac{p^3k^2}{\eps^2}) \times n^{q - 1}$ Countsketch matrix
    \State $T_m \gets $ An $s \times O(\frac{p^3k^2}{\eps^2})$ matrix with i.i.d. $\pm \frac{1}{\sqrt{s}}$ entries, where $s = O(\frac{pk \log p}{\eps})$. 
    \State $\calS_m \gets \varnothing$
    \State $\widehat{A_m} \gets M_m(A) S_m^T T_m^T$ \\
    \For{$t = 1 \to T$}
        \State $R_m \gets $ A $k \times s$ matrix with i.i.d. $N(0, 1/k)$ entries.
        \State $U^m_t \gets \widehat{A_m} R_m^T$
        \State $\calS_m \gets \calS_m \cup \{U^m_t\}$
    \EndFor
\EndFor

\item []
\State // Evaluating error of each tuple in $\calS_1 \times \ldots \times \calS_p$
\State $A_p \gets $ An $n^p \times n^{q - p}$ matrix whose $((i_1, \ldots, i_p), (i_{p + 1}, \ldots, i_q))^{th}$ entry is $A(i_1, \ldots, i_q)$
\State $S_{\textsc{PROJ}} \gets $ An $O(k^{2p}/\eps^2) \times n^{q - p}$ Countsketch matrix.
\State $M \gets A_p S_{\textsc{PROJ}}^T$ \Comment{Used as PCP for $A_p$ in Kronecker product regression}
\State $\widehat{U^1}, \widehat{U^2}, \ldots, \widehat{U^p} \gets 0$
\State $\textsc{MinError} \gets \|A\|_F^2$ \\
\For{$U^1 \in \calS_1, \ldots, U^p \in \calS_p$}
    {\State $\textsc{NewError} \gets 0$ \\
    \For{$j = 1 \to O(k^{2p}/\eps^2)$}
        {\State $r \gets O(k^{2p}/\eps^2)$, the number of columns in $M$
        \State $\textsc{NewError} \gets \textsc{NewError} + \textsc{ErrorEstimate}(U^1, \ldots, U^p, M_j, \delta = \frac{1}{1000r T^p})$
        }
    \EndFor
    \State $\textsc{NewError} \gets \textsc{NewError} + \|M\|_F^2$ \\
    \If{$\textsc{NewError} \leq \textsc{MinError}$}
        {\State $\textsc{MinError} \gets \textsc{NewError}$
        \State $\widehat{U^1}, \widehat{U^2}, \ldots, \widehat{U^p} \gets U^1, U^2, \ldots, U^p$}
    \EndIf
    }
\EndFor \\
\Return $\widehat{U^1}, \widehat{U^2}, \ldots, \widehat{U^p}$
\end{algorithmic}
\end{algorithm}

%
%
%

\paragraph{Tensor Train FPT Algorithm.}
Our FPT algorithm for Tensor Train decompositions uses
similar tools as that for Tucker decompositions, but notably it uses our new subspace embedding for tensor contractions that we introduced above in the context of obtaining efficient bicriteria algorithms.

For Tensor Train and Tucker-$(p, q)$ decompositions, our algorithm achieves output rank exactly $k$ and $(1+\eps)$ relative approximation error, in fixed-parameter tractable time. Previous work, such as \cite{sun_lowranktucker} which uses High Order SVD, generally produces either higher than $(1+\eps)$ approximation error, or higher output rank, or both. A summary of our results and previous fixed parameter tractable algorithms that achieve $(1+\eps)$-approximation is presented in Table \ref{tab:fpt_results}.

\begin{remark}
Note that for CP rank, the best rank $k$ approximation may not exist in general. To deal with the case where the best rank $k$ approximation does not exist, we can add an arbitrarily small additive error term $\gamma$ to the approximation guarantees, as done in \cite{swz19_tensor_low_rank}.
\end{remark}

\paragraph{CP-Rank FPT Algorithm.}
Finally, our toolbox allows us to obtain
new FPT algorithms for CP decompositions. Our approach here is inspired by that of \cite{swz19_tensor_low_rank}, but instead of using polynomial system solvers, uses our  low probability sketching primitives. For CP decomposition, our algorithm gives a tradeoff compared to \cite{swz19_tensor_low_rank}, as we obtain bicriteria rank $k^{q - 1}$, which is better than their bicriteria algorithm but worse than their fixed-parameter tractable algorithm (which obtains output rank $k$), but our running time is fixed-parameter tractable while their bicriteria algorithm is polynomial time. As mentioned above, our algorithm does not use polynomial system solvers unlike the fixed-parameter tractable algorithm of \cite{swz19_tensor_low_rank}. Note that by Theorem 1.1 of \cite{ss16_tensor_rank_polynomial_systems}, computing the CP rank of a tensor is equivalent to solving polynomial systems, meaning obtaining output rank $k$ is at least as hard as solving polynomial systems, which may not be true for the Tucker and Tensor Train decompositions.


\subsubsection{Hardness}

We show new fine-grained hardness results for rank-$1$ Tucker-$(2, 3)$ decomposition, rank-$1$ CP decomposition, and Tensor Train decomposition with $q = 3$. These three are in fact equivalent decompositions. To see why, first note that rank-$1$ CP decomposition is equivalent to rank-$1$ Tensor Train decomposition since a tensor $A \in \R^{n \times n \times n}$ with Tensor Train rank $1$ can be written as $u_1 \circ u_2 \circ u_3$ where $u_1 \in \R^{n \times 1}$, $u_2 \in \R^{1 \times n \times 1}$, and $u_3 \in \R^{1 \times n}$, and the $(i, j, k)^{th}$ entry of $u_1 \circ u_2 \circ u_3$ can be written as $u_{1, i} u_{2, j} u_{3, k}$. Thus, every tensor with CP rank $1$ has Tensor Train rank $1$, and vice versa. Additionally, suppose a tensor $A \in \R^{n \times n \times n}$ can be written as a tensor with multilinear rank $(1, 1, n)$ (corresponding to Tucker-$(2, 3)$ decomposition). Then, we can write $A = (u_1 \times u_2 \times I)G$ for some tensor $G \in \R^{1 \times 1 \times n}$ and $u_1, u_2 \in \R^{n \times 1}$. Thus, $A(i, j, k) = (e_i^\top u_1 \times e_j^\top u_2 \times e_k^\top I) G = u_{1, i} u_{2, j} G(1, 1, k)$ --- in other words, Tucker-$(2, 3)$ decomposition is equivalent to rank-$1$ CP decomposition. It can also be seen that the rank-$1$ Tucker-$3$ decomposition is equivalent to rank-$1$ CP decomposition using the same argument. Thus, using for instance that rank-$1$ CP decomposition is NP-hard \cite{HASTAD1990644,mostNPHard}, all of these problems are NP-hard for $\varepsilon = 0$.

Moreover, as shown in \cite{swz19_tensor_low_rank}, under the Exponential Time Hypothesis, there is a $2^{\Omega(\varepsilon^{1/4})}$ time lower bound for rank-$1$ CP decomposition and thus for all of these problems. We make a stronger assumption based on the $2$-to-$4$ norm defined as follows. 
\begin{definition}
Let $B \in \R^{n \times n}$. Then, we define $\|B\|_{2, 4} = \Big( \sum_{i = 1}^n \|B_{i, :}\|_2^4 \Big)^{1/4}$, where $B_{i, :}$ is the $i^{th}$ row of $B$, and $\|B\|_{2 \to 4} = \sup_{\|x\|_2 = 1} \|Bx\|_4$.
\end{definition}
\begin{conjecture}
\label{conjecture:2_to_4_norm_running_time}
Any algorithm which, given a matrix $A \in \R^{n \times n}$, approximates $\|A\|_{2 \to 4}$ to a multiplicative $O(1)$ factor, requires $2^{\Omega(n)}$ time.
\end{conjecture}

\paragraph{Running Time Lower Bound for Tucker-$(2, 3)$ Decomposition}

We show that under Conjecture \ref{conjecture:2_to_4_norm_running_time}, any \textbf{$(1 + \varepsilon)$-approximation} algorithm for rank-$1$ CP, Tucker and Tensor Train decomposition requires $2^{\Omega(1/\varepsilon)}$ time.

\begin{theorem} \label{thm:hardness_rank1}
For convenience, given a tensor $A \in \R^{n \times n \times n}$, let $A^k := A(:, :, k)$. If Conjecture \ref{conjecture:2_to_4_norm_running_time} is true, then any algorithm which, given $A \in \R^{n \times n \times n}$, finds unit vectors $u, v \in \R^n$, such that
$$\sum_{k = 1}^n \|uu^T A^k vv^T - A^k \|_F^2 \leq \Big(1 + \Theta\Big(\frac{1}{n}\Big)\Big) \min_{\|u_*\|_2 = \|v_*\|_2 = 1} \sum_{k = 1}^n \|u_* u_*^T A^k v_* v_*^T - A^k \|_F^2$$
requires at least $2^{\Omega(n)}$ time. Thus, any algorithm which, given $A \in \R^{n \times n \times n}$ and $\eps \in (0, 1)$, finds unit vectors $u, v \in \R^n$, such that
$$\sum_{k = 1}^n \|uu^T A^k vv^T - A^k \|_F^2 \leq (1 + \epsilon) \min_{\|u_*\|_2 = \|v_*\|_2 = 1} \sum_{k = 1}^n \|u_* u_*^T A^k v_* v_*^T - A^k \|_F^2$$
requires at least $2^{\Omega(1/\eps)}$ time.
\end{theorem}

This theorem implies that the rank-$1$ Tucker-$(2, 3)$ decomposition problem requires $2^{\Omega(1/\epsilon)}$ time if $1 + \epsilon$ relative error is desired. To see why, recall that (see Definition \ref{def:tucker_decomposition}) the rank-$1$ Tucker-$(2, 3)$ decomposition problem corresponds to finding $u_1 \in \R^n, u_2 \in \R^n$ and $G \in \R^{1 \times 1 \times n}$ such that $\|M_3(G) (u_1^\top \times u_2^\top) - M_3(A)\|_F^2$ is minimized. We can thus rewrite the objective as
\begin{equation}
\begin{split}
\min_{g, u_1, u_2} \sum_{k = 1}^n \|g_k (u_1^\top \times u_2^\top) - \tensorvec(A^k) \|_F^2 
& = \min_{g, u_1, u_2} \sum_{k = 1}^n \|g_k u_1 u_2^\top - A^k\|_F^2
\end{split}
\end{equation}
using the notation $g_k = G(1, 1, k)$ and $A^k = A(:, :, k)$. Therefore, given a fixed $u_1$ and $u_2$, we have $g_k := \argmin_t \|t u_1 u_2^\top - A^k \|_F^2$. Additionally, by the Pythagorean theorem, assuming that $u_1$ and $u_2$ have unit norm (as otherwise, we can scale all the $g_k$'s appropriately), we have
\begin{equation}
\begin{split}
\|t u_1 u_2^\top - A^k\|_F^2
& = \|t u_1 u_2^\top - u_1 u_1^\top A^k \|_F^2 + \|u_1 u_1^\top A^k - A^k\|_F^2 \\
& = \|t u_1 u_2^\top - u_1 u_1^\top A^k u_2 u_2^\top \|_F^2 + \|u_1 u_1^\top A^k u_2 u_2^\top - u_1 u_1^\top A^k\|_F^2 + \|u_1 u_1^\top A^k - A^k\|_F^2
\end{split}
\end{equation}
where we have applied the Pythagorean theorem twice. Thus, it is optimal to have $g_k = u_1^\top A^k u_2$, and therefore, for unit vectors $u, v \in \R^n$, the objective
$$\sum_{k = 1}^n \|uu^T A^k vv^T - A^k \|_F^2$$
is equivalent to the Tucker-$(2, 3)$ decomposition objective. Thus, under this conjecture, our fixed-parameter tractable algorithms for the Tucker and Tensor Train decompositions are optimal in terms of their dependency on $1/\varepsilon$ in the running time.

\paragraph{Techniques} Our reduction from $2 \to 4$ norm uses the specific form of a Tucker-$(2,3)$ decomposition, expanding the objective function, together with the fact that for symmetric matrices $A_k$, it holds that $\max_{\|u\|_2 = 1} \sum_{k = 1}^m \langle A_k u, u \rangle^2 = \max_{\|x\|_2 = \|y\|_2 = 1} \sum_{k = 1}^m \langle A_k x, y \rangle^2$. We use an inequality relating $\|A\|_{2 \to 4}$ to the sum of $4$-th powers of Euclidean norms of slices of $A$, which in turn is related to the Tucker-$(2,3)$ objective after applying this fact. 

We note that known upper bounds for approximating $\|A\|_{2 \to 4}$ up to a multiplicative $O(1)$ factor do not contradict this. Indeed, the work of \cite{b12} shows how in $2^{O(n^{1/2})}$ time one can distinguish between the cases $\|A\|_{2 \to 4} \leq c \sigma_{min}(A)$ and $\|A\|_{2 \to 4} \geq C \sigma_{min}(A)$ if $1 < c < C$ are fixed constants, where $\sigma_{min}(A)$ is the smallest singular value of $A$. However, for super-constant $c$ and $C$, one can show their algorithm takes $2^{\Omega(n)}$ time, and thus is no faster than enumerating over a net. \footnotemark \footnotetext{On page 64 of \url{https://arxiv.org/pdf/1205.4484.pdf} it is mentioned that if $\dim(V) \leq C^2 n^{2/q}$ which is $C^2 n^{1/2}$ for $q = 4$, then brute force enumeration can be used and this would take time exponential in $\dim(V)$. Note that the goal is to determine whether $\|A\|_{2 \to 4} \geq C \sigma$ where $\sigma$ is the least singular value of $A$. If no restrictions are placed on $C$, then it could be larger than $n^{1/4}$, since for a fixed $v$, $\frac{\|Av\|_4}{\|Av\|_2}$ could be as large as $n^{1/4}$, and moreover, $\frac{\|Av\|_2}{\|v\|_2}$ could be arbitrarily larger than $\sigma$, meaning $\frac{\|Av\|_4}{\|v\|_2}$ could be larger than $n^{1/4}\sigma$. Thus, brute-force enumeration would take $2^{O(n)}$ time, in this case. If brute-force enumeration is not used, then Corollary 10.2 of that work can be used to obtain the bound $\|V\|_{2->4} \geq \sqrt{\dim(V)}/n^{1/4}$. Since $\dim(V) \leq n$, this bound is at best $n^{1/4}$, while $C$ could be larger. Thus, the algorithm of \cite{b12} can take exponential time if no restrictions are placed on $C$ and $c$.} Thus, either our lower bounds hold, or a major breakthrough will be needed for approximating the $\|A\|_{2 \to 4}$ norm. We note that such reductions are common in fine-grained complexity theory (see \cite{bringmann_fine_grained_complexity_tutorial}).

\subsubsection{Robust Loss for Tucker Decomposition}

As studied in  \cite{swz19_tensor_low_rank}, another question is whether it is possible to obtain low rank decompositions with good relative error for more robust loss functions than the Frobenius norm. A number of works attempt to address this question, see, e.g., \cite{swz19_tensor_low_rank,l1Tucker} for work which looks at the $\ell_1$-norm and sum of Euclidean norm losses. We propose a new algorithm for Tucker decomposition where the loss function is the sum of the Frobenius norms of the faces, where a face is a subtensor obtained by fixing the index along the last mode. We note that \cite{swz19_tensor_low_rank} only consider robust loss functions for the CP decomposition.

We consider the following norm $\|\cdot\|_R$, which is less sensitive to outliers than the Frobenius norm:

\begin{definition}[Sum of Frobenius Norms of Mode-$q$ Faces]
\label{def:robust_loss_function}
Let $A \in \R^{n \times \ldots \times n}$ be a $q$-mode tensor. Then,
$\|A\|_R = \sum_{i \in [n]} \|A(:, \ldots, :, i)\|_F$, 
where $A(:, \ldots, :, i)$ is the slice of $A$ whose index in the last mode is $i$. In other words, $\|A\|_R$ is the sum of the Frobenius norms of the faces along the last mode.
\end{definition}

$\| \cdot \|_R$ is a high dimensional generalization of the $\ell_{2, 1}$ norm for matrices. $\ell_{2, 1}$ norm has been frequently adopted over the Frobenius norm for many optimization and maching learning tasks. For example, \cite{matrix_L12norm_application} uses $\ell_{2, 1}$ norm for discrimitive feature selection.

\begin{problem}[Robust Tucker-$(p, q)$ Decomposition]
\label{problem:robust_tucker_decomposition}
Let $A \in \R^{n \times \ldots \times n}$ be a $q$-mode tensor, with $p < q$, and let $k \in \N$. Then, we wish to find a $q$-mode tensor $X \in \R^{n \times \ldots \times n}$ of multilinear rank at most $(k, \ldots, k, n, \ldots, n)$ (where the first $p$ entries of the tuple are $k$ and the last $q - p$ entries are $n$), such that $\|X - A\|_R$ is minimized.
\end{problem}

We give an $O(p \cdot (\sqrt{k} \log k + \sqrt{\log (np)}))$-approximation algorithm with bicriteria rank $O(k \log^2 k)$ --- we formally state our result in Appendix \ref{sec:robust_loss_fn}. We make the assumption that $p < q$ only for this robust loss function --- all of our results for Frobenius norm Tucker decomposition hold even if $\,p = q$. We make this assumption since it allows us to apply techniques for $\ell_{1, 2}$-norm matrix low rank approximation - different techniques might be required if $p = q$.

\begin{remark}
In all of our results, unidentifiability does not affect our guarantees, since we are not recovering the ground truth. Instead, we compare the error to the best tensor which can be represented according to the given network with a specified rank.
\end{remark}

\subsubsection{Generalization: Tensor Network Approximation}

We extend our results summarized in Subsubsection \ref{subsubsec:tensor_train_bicriteria} to general tensor networks, which we define informally below. See Subsection \ref{subsec:general_tensor_network_contraction} for a formal definition.

\begin{definition}[General Tensor Network Contraction --- Informal]
Let $\calG = (V, E)$ be a graph. For each $v \in V$, suppose $U_v$ is a tensor with $\deg(v)$ modes of dimension $k$ (corresponding to each edge incident to $v$) and one mode of dimension $n$. We define the contraction of $\calG$ to be the tensor obtained by contracting $U_u$ and $U_v$ for every edge $(u, v) \in \calG$, until we are left with a single vertex.  For convenience, we denote this contraction by $\calG(\{U_v \mid v \in V\})$. We will informally refer to $k$ as the rank of the network according to $\calG$.
\end{definition}

The Tensor Train decomposition is a special case of tensor networks with $\calG$ being a line graph. We first extend our Tensor Train decomposition algorithm to obtain a bicriteria $(1+\eps)-$approximation algorithm for the special case of tree network decomposition.

\begin{theorem}[Theorem \ref{thm:analysis_of_tree_network_bicriteria} --- Informal]
Let $A \in \R^{n \times \ldots \times n}$ be a $q$-mode tensor and $\calT$ be a tree. Then, our algorithm outputs a tree network $M$ of rank $O(\frac{qk}{\epsilon} \log(\frac{q}{\delta}))$, such that
$$\|M - A\|_F \leq (1 + \epsilon) \min_T \|T - A\|_F$$
where the minimum on the right-hand side is taken over all tensors $T$ with rank at most $k$ according to $\calT$. The running time of our algorithm is $O(q \cdot \nnz(A)) + n \cdot (\frac{qk}{\epsilon \delta})^{O(d)}$.
\end{theorem}

We also show how to apply the tree network decomposition algorithm to approximate tensors with low ranks under a particular graph $\calG$ which may have cycles. Many such network-based decompositions have been used in different applications.

\begin{figure}[htb]
     \centering
     \begin{subfigure}[b]{0.2\textwidth}
         \centering
         \includegraphics[width=0.95\textwidth]{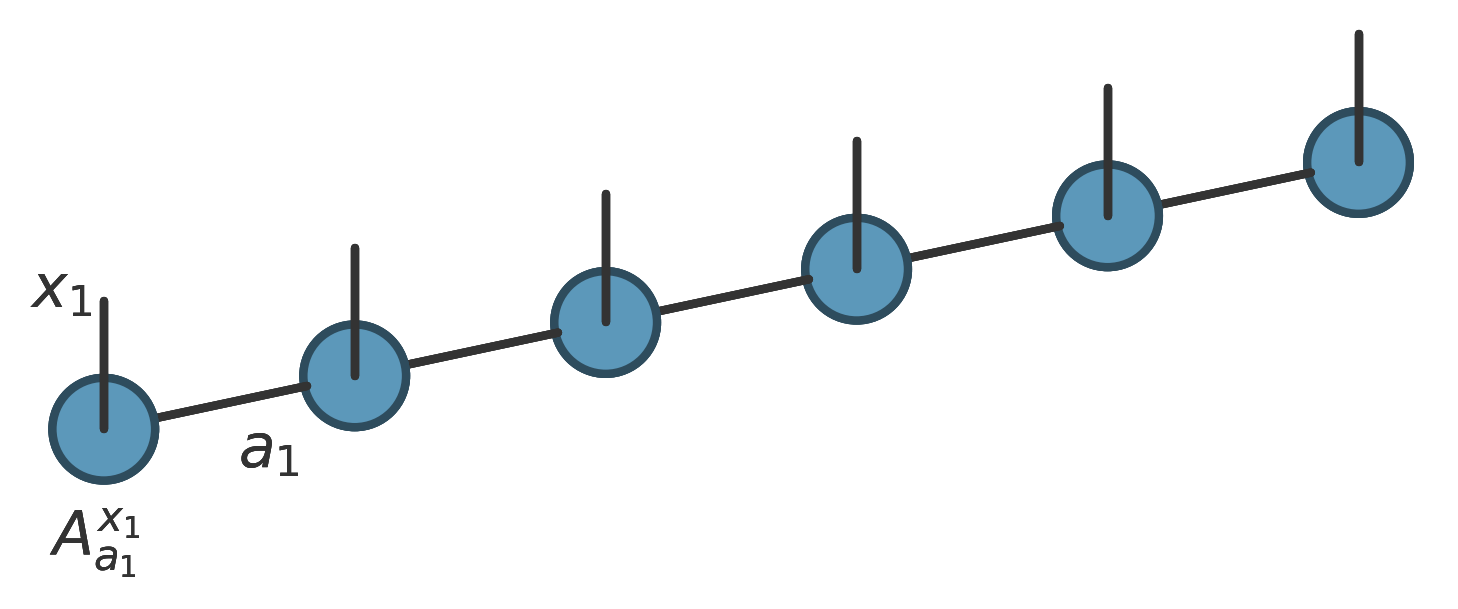}
         \caption{Matrix Product State/Tensor Train}
     \end{subfigure}
     \hfill
     \begin{subfigure}[b]{0.15\textwidth}
         \centering
         \includegraphics[width=0.75\textwidth]{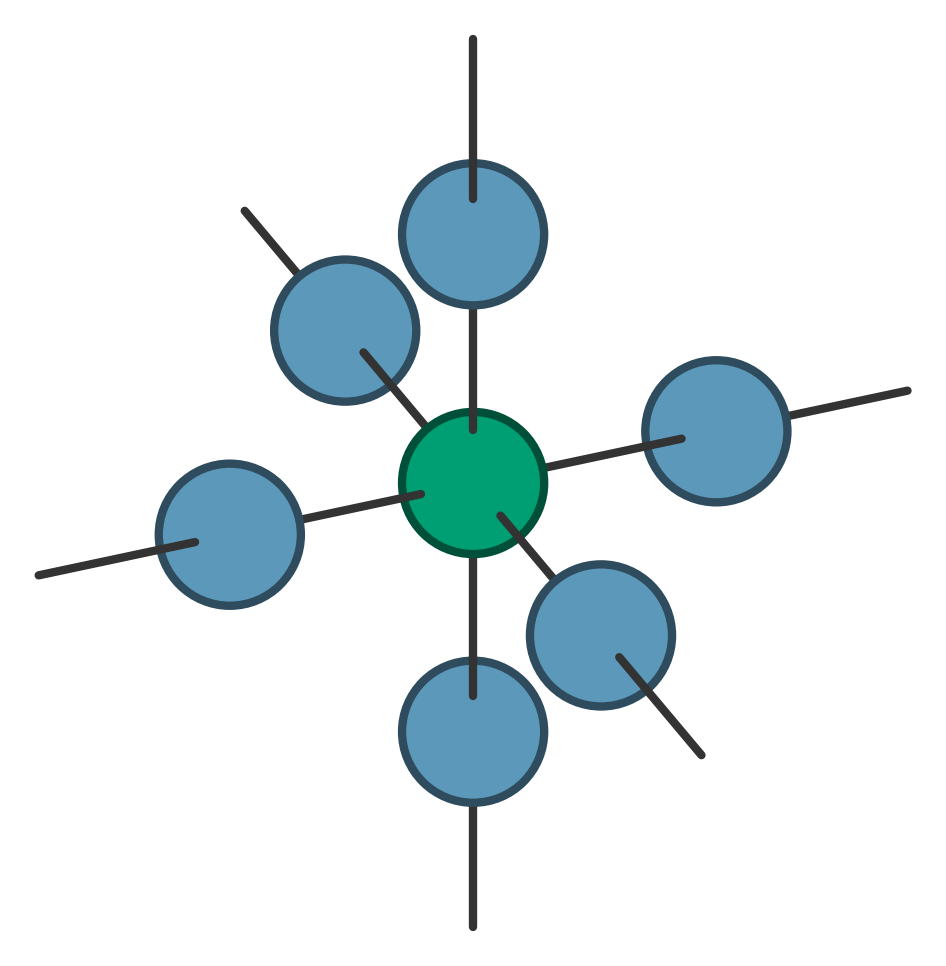}
         \caption{Tucker}
     \end{subfigure}
     \hfill
     \begin{subfigure}[b]{0.2\textwidth}
         \centering
         \includegraphics[width=0.9\textwidth]{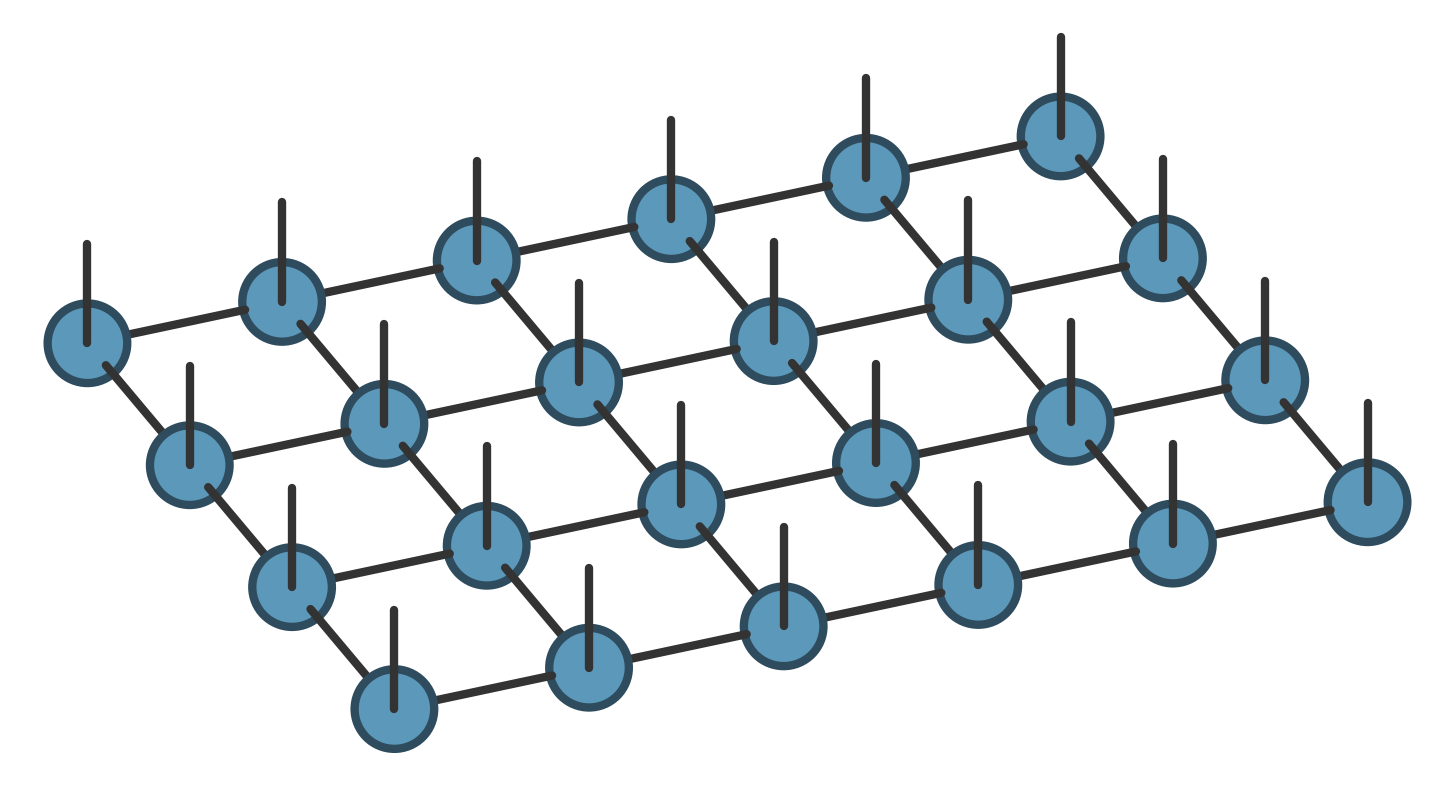}
         \caption{Projected Entangled Pair State}
     \end{subfigure}
     \hfill
   \begin{subfigure}[b]{0.2\textwidth}
         \centering
         \includegraphics[width=0.9\textwidth]{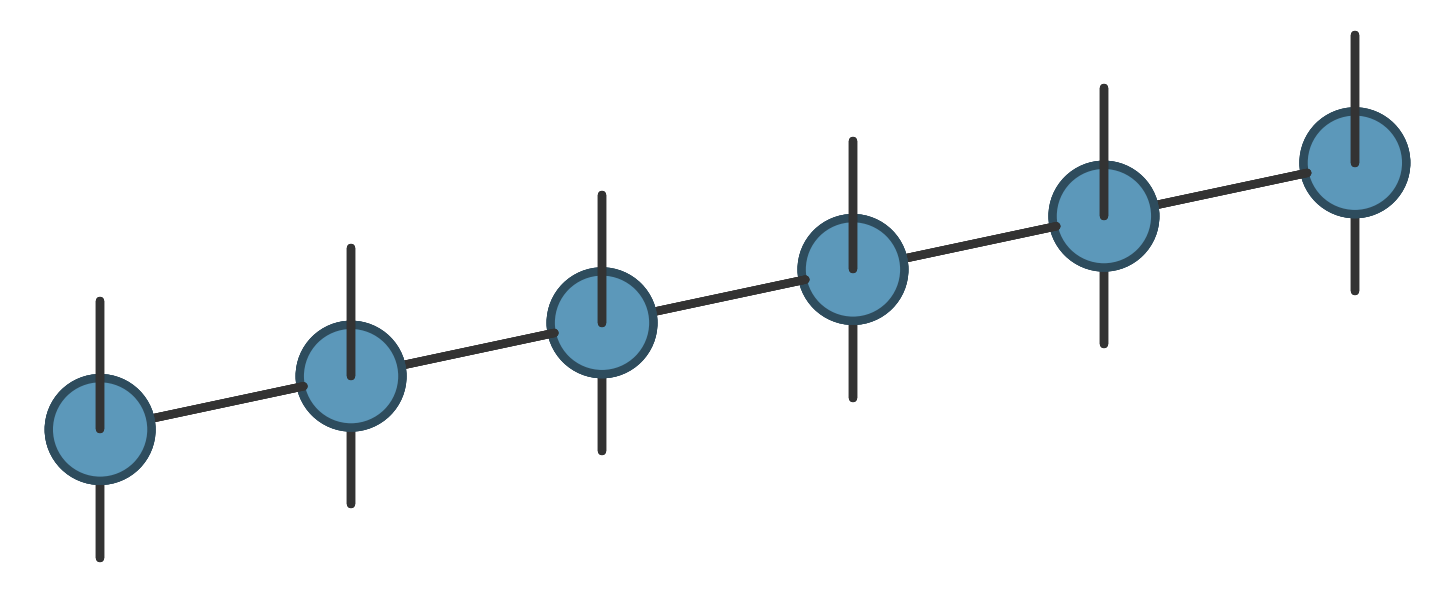}
         \caption{Matrix Product Operator}
     \end{subfigure}
    \hfill
     \begin{subfigure}[b]{0.2\textwidth}
         \centering
         \includegraphics[width=0.9\textwidth]{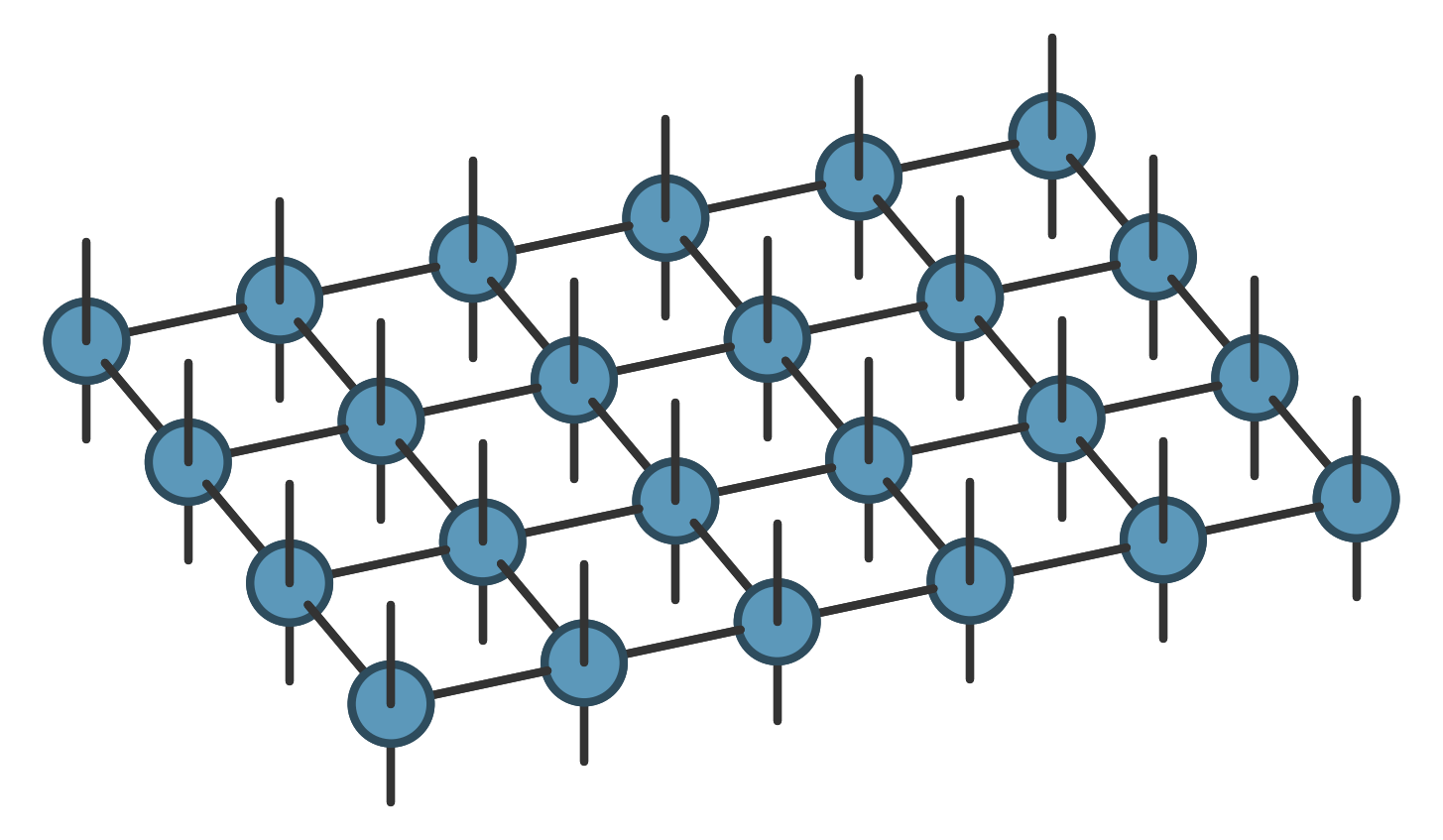}
         \caption{Projected Entangled Pair Operator}
     \end{subfigure}
     \caption{Popular Tensor Networks.}
     \label{fig:tns}
\end{figure}



\paragraph{Tensors with Low General Tensor Network Rank}

We show how to find a binary tree $\calT$ on a vertex set $V'$ containing $O(q)$ nodes, and corresponding tensors $U_v$, again with a mode corresponding to each edge $e \in T$ of dimension poly$(k^{\deg(G)\textrm{tw}(G)}q/\varepsilon)$, for which
$$\|\calT(\{U_v \mid v \in V'\}) -  A\|_F^2 \leq (1+\varepsilon) \cdot \|\calG(\{U^*_v \mid v \in V\}) - A\|_F^2,$$
in $O(\nnz(A) q + n \cdot \poly(k^{\deg(G)\textrm{tw}(G)}q/\varepsilon)$ time, where tw$(G)$ is the treewidth of $G$, $\deg(G)$ is the maximum degree of a vertex, and $\nnz(A)$ is the number of non-zero entries of $A$.  Although our output tensor $\calT(\{U_v \mid v \in V'\})$ does not have the same topology as the original network $\calG$, it performs as well, up to a $(1+\varepsilon)$-factor, as the best tensor with network topology $\calG$, and moreover, the number of parameters in the tree network $\calT$ is at most a factor $k^{O(\deg(G) \textrm{tw}(G))}$ larger than that in $\calG$.
The overall tree network will have $q \cdot n \cdot k^{O(\deg(G)^2 \tw(G)^2)}$ parameters, instead of $n^q$.

The main idea of our algorithm is that, given a graph $\calG = (V, E)$ and a corresponding set of factors $\{U_v \mid v \in V\}$ where the ranks on the edges in $E$ are all of rank $k$, there exists a binary tree network $\calT$ with corresponding factors $\{W_v \mid v \text{ is a vertex of }\calT\}$ which is equivalent, i.e., 
$$\calT(\{W_v \mid v \text{ is a vertex of }\calT\}) = \calG(\{U_v \mid v \in V\})$$
Crucially, $\calT$ also has low rank. Specifically, the ranks on the edges of $\calT$ are at most $k^{\deg(\calG)\tw(\calG)}$. This is shown in subsection \ref{subsec:binary_tree_existential_result}.

Using the observation above, we construct the binary tree $\calT$ by contracting the edges of $\calG$ in a particular order, and each vertex $w$ resulting from the contraction of an edge $(u, v)$ has two children $u'$ and $v'$ corresponding to $u$ and $v$. In order for the edges of $\calT$ to have low rank, it is crucial for the degree of the vertices $w$ obtained from contraction to be small. For this to hold, we make use of the following result from \cite{ms08_simulating_quantum_tensor_contraction}: there is an order in which we can contract the edges of $\calG$, such that the largest degree of any vertex at any point during the contraction is at most $O(\deg(\calG) \tw(\calG))$. Moreover, this order can be computed in $\poly(|V|) e^{O(\deg(\calG) \tw(\calG))}$ time. The technical details of the algorithm are presented in Section \ref{sec:general_networks}.

\subsection{Comparison to Prior Work}

The study of fast algorithms for tensor decompositions is a vast subject with work in numerical linear algebra, scientific computing, and theoretical computer science. We discuss several other approaches to these problems below, but we emphasize the main differences with our work here:
\begin{enumerate}
    \setlength\itemsep{0.1em}
    \item There is no previous work for general tensor network low rank approximation, despite work that motivates studying this question \cite{yl18}.
    \item We obtain relative error $(1+\varepsilon)$-approximations, whereas, with the exception of \cite{swz19_tensor_low_rank}, previous work either gets an additive error or at best a fixed constant approximation factor (e.g. \cite{huber2017randomized} and \cite{grasedyck10_hierarchical_tucker} which obtain $O(\sqrt{q})$ approximation algorithms).
    \item Previous works do not run in $\nnz(A) \cdot q$ time, but rather at least $\nnz(A) \poly(kq)$ time (such as \cite{huber2017randomized} which takes at least $O(qs^2 \nnz(A))$ time where $s$ is their output rank, and \cite{grasedyck10_hierarchical_tucker} for Hierarchical Tucker decomposition which requires $O(n^q)$ time due to the use of SVD).
\end{enumerate}
In addition to the above, a major contribution of our work is to introduce many of the advanced and recent techniques from randomized numerical linear algebra to the study of more advanced tensor decompositions. We note that some such tools have been developed, such as TensorSketch \cite{pham2013fast} and its optimizations \cite{akkpvwz20_tensor_sketch_outer_product}, but they do not directly apply to Tensor Train decompositions, for example. In this work, we develop new algorithmic tools, such as sparse affine and subspace embeddings for tensors represented implicitly, such as those represented as a sequence of $\circ$ operations, 
as well as new analytical tools such as the propagation of JL-moment guarantees across the nodes of a tree network. We note that the idea of implicit tensor maps is seen in prior works such as \cite{implicit_projection_map}. However, the context and goal are vastly different, the map in \cite{implicit_projection_map} is specific to a subspace, where ours is an embedding. \cite{ms2022_tensor_network_embeddings_concurrent_and_independent} that has a similar spirit in developing tensor embeddings does not present an end-to-end algorithm and guarantees for approximate tensor decomposition, discussed in detail in Section \ref{sec:additional_related_work}.

\begin{table}
    \centering
    \begin{tabular}{|c|c|c|c|}
    \hline
    Work & Running Time & Approximation Factor \\ \hline
    \cite{grasedyck10_hierarchical_tucker} & $O(n^{\frac 3 2 q})$ & $\sqrt{q}$ \\ \hline
    \cite{sun_lowranktucker} & $O\left(\left(\frac{s(1 - s/n)^q}{1 - s/n} + qk\right)n^q\right)$  & $2q$    \\ \hline
    \cite{MinsterRachel2020RAfL} & $O(\sum_{j=1}^q (k^jn^{q-j+1} + k^jn^{q-j}))$  & $O(\sqrt{q})$    \\ \hline
    \textbf{This work} & $O(q \cdot \nnz(A)) + n \cdot (qk/\eps)^{O(q^2k^2\log q \frac 1 \eps)}$ & $1 + \eps$    \\ \hline
    \end{tabular}
    \caption{Comparison of algorithms for the Tucker decomposition. $\omega$ is the matrix multiplication complexity exponent.}
    \label{tab:tucker_results_comparison}
\end{table}

Besides \cite{swz19_tensor_low_rank}, there are some prior works on tensor decompositions using the same family of techniques as ours, such as \cite{sun_lowranktucker} with Tucker-factor-wise sketches, \cite{pmlr20_tensorized_random_proj} with randomized projections, \cite{TensorSketch} with subspace embeddings for tensors, and \cite{leverage_score_tensor_ring} with leverage score sampling. Several algorithms using more standard regression techniques have been developed for tensor decompositions. This includes Alternative Least Squares algorithms such as \cite{fahrbach2021fast_ridge_regression}, and SVD-like algorithms such as \cite{MinsterRachel2020RAfL}, \cite{oseledets2011tensor}, \cite{huber2017randomized}, and \cite{faster_tt_sparse_2019}.
These works are generally significantly more computationally expensive, or have worse guarantees, or both. In particular, they do not achieve $(1+\eps)-$approximation within a practical time complexity. We discuss them in more detail later in section \ref{sec:additional_related_work}.

In \cite{dbchmprs2021_arxiv_tt_rounding} which studies the related problem of Tensor Train rounding, a Tensor Train embedding similar to ours is proposed --- the embedding consists of a Tensor Train network where each factor has i.i.d. Gaussian entries. This work also proposes the ``right-to-left partial $\circ$'' operation to quickly apply this embedding to another Tensor Train. The work \cite{dbchmprs2021_arxiv_tt_rounding} does not analyze the error incurred using this embedding in their algorithms, and the embedding is not used as a subspace embedding or affine embedding. Instead, it is used to approximately perform the QR decomposition of matrices of the form $T_{X, n} T_{Y, n + 1 : N}$, where $Y$ is a Tensor Train of length $N$ given as the input to the Tensor Train rounding algorithm, and $T_{Y, n + 1 : N}$ is a matrix formed from cores $n + 1$ through $N$ of $Y$. Cores $n + 1$ through $N$ of the Gaussian Tensor Train network are applied to $T_{Y, n + 1 : N}$, and the column span of the resulting matrix is computed.

The concurrent and independent \cite{sos_tensor_ring} follows a different line of techniques towards $\circ$ based decompositions. Compared to ours, their Tensor Ring decomposition result makes several additional assumptions, notably $\forall i, j, k\in [q] |T_{i, j, k} - Tr(U_i^*U_jU_k^*)| <= \eta$ and the analogous second order constraint for the Sum-of-Squares technique. Their error guarantee is poly$(n, k, \max_{i \in q}||U^*_i||_F, 1/m)\eta^c$ over component-wise parameter distance ($m$ is a lower bound on the condition number of the fused matrix of an arbitrary combination of two modes in the optimal solution), but each output component is of rank $k$. To remove the $n$ dependency in the error guarantee, they need additional smoothed analysis assumptions. Our work takes a different approach to achieve a bicriteria algorithm and gets $(1+\eps)$ relative error guarantees with a comparable runtime.

We note that whether one can efficiently contract tensor networks with cycles is a question that has been studied in quantum physics and quantum computation (e.g.,  \cite{ms08_simulating_quantum_tensor_contraction}, \cite{Haferkamp_2020}, \cite{o2019parameterization}, to name a few examples). To our knowledge, the problem of approximating a given input tensor $A$ by a general tensor network of low rank has not been extensively studied. However, \cite{yl18} gives examples where the rank of a tensor $A$ with respect to one tensor network can be significantly lower than the rank with respect to another. We also note that hardness results for contracting PEPS networks (e.g., \cite{Haferkamp_2020}) suggest that either converting PEPS networks to equivalent tree networks is hard to do efficiently, or that an equivalent tree network would require a much larger rank.
 
%

\newpage
\appendix
\section*{Table of Contents for the Appendices}

\begin{itemize}
    \item In Section \ref{sec:preliminaries} we introduce sketching tools that are used throughout the paper, and also include the formal definition of general tensor network $\circ$.
    \item In Section \ref{sec:tree_networks} we introduce our bicriteria $(1 + \eps)$-approximation algorithm for tree networks.
    \item In Section \ref{sec:general_networks} we give our approximation algorithm for general tensor networks.
    \item In Section \ref{sec:frobenius_norm_1_plus_eps_rank_exactly_k} we give our fixed-parameter tractable algorithms for Tucker decomposition, Tensor train decomposition, and CP decomposition with output rank exactly $k$.
    \item In Section \ref{sec:rank_1_hardness} we give our hardness result based on Conjecture \ref{conjecture:2_to_4_norm_running_time} about approximating the $2 \to 4$ norm.
    \item In Section \ref{sec:robust_loss_fn} we give our algorithm for Tucker-$(p, q)$ decomposition with a robust loss function.
    \item In Section \ref{sec:additional_related_work} we discuss additional related works.
\end{itemize}

\newpage
\section{Preliminaries} 

\label{sec:preliminaries}

\subsection{General Tensor Network Contraction $\circ$}
\label{subsec:general_tensor_network_contraction}

\begin{definition}[General Tensor Network Contraction $\circ$]
\label{def:tensor_network_contraction}
Let $\calG = (V, E)$ be a connected graph, $q, n, k \in \N$, such that $|V| = q$. Suppose $\{U_v \mid v \in V\}$ is a collection of tensors, such that for each $v \in V$, $U_v$ has one mode for each of the edges incident to $v$ (of dimension $k$), and an additional mode (of dimension $n$), for a total of $\deg(v) + 1$ modes. Then, we let $\calG(\{U_v \mid v \in V\})$ be the tensor obtained by contracting $\{U_v \mid v \in V\}$ in the manner specified by $\calG$. Specifically, $\calG(\{U_v \mid v \in V\})$ is the tensor defined as follows:
\begin{enumerate}
    \setlength\itemsep{0.05em}
    \item Let $G = \calG$, and let $\calS = \{U_v \mid v \in V\}$.
    \item While $G$ has more than one vertex:
    \begin{enumerate}
        \setlength\itemsep{0.03em}
        \item Select an arbitrary edge $e = (u, v)$ in $G$.
        \item Let $G' = (V', E')$ be the graph obtained by contracting the edge $e$ in $G$, and let $w$ be the vertex in $G'$ that replaces $u$ and $v$.
        \item Remove $U_u$ and $U_v$ from $\calS$.
        \item Define $U_w$ to be the tensor obtained by contracting $U_u$ and $U_v$ along their respective modes corresponding to the edge $e$.
        \item Add $U_w$ to $\calS$.
    \end{enumerate}
    \item Let $T$ be the single tensor remaining in $\calS$. Then, $T$ is the desired tensor.
\end{enumerate}
\end{definition}

\subsection{Sketching Tools} \label{subsec:prelim_sketching_tools}

\subsubsection{Definitions} \label{subsubsec:sketching_definitions}

Let us first introduce some basic definitions. See \cite{woodruff14_sketching_as_a_tool} for a detailed exposition.

\begin{definition} [Approximate Matrix Product]
\label{def:approx_matrix_product}
Let $\calD$ be a distribution over random matrices $S \in \R^{r \times n}$. We say $\calD$ satisfies the $(\eps, \delta)$-approximate matrix product property if, for any matrices $A \in \R^{n \times d_1}, B \in \R^{n \times d_2}$,
$$\mathbf{Pr}_{S \sim \calD}\Big[\|A^TS^TSB - A^TB\|_F \geq \eps \|A\|_F \|B\|_F\Big] \leq \delta$$
\end{definition}

To show that $S$ satisfies the $(\eps, \delta)$-approximate matrix product property, it suffices to show that it satisfies the $(\eps, \delta, \ell)$-JL Moment Property:

\begin{definition}[$(\eps, \delta, \ell)$ JL Moment Property --- Definition 20 of \cite{kn14_sparser_jl_transforms}] \label{def:jl_moment_property}
Let $\calD$ be a distribution over random matrices in $S \in \R^{r \times d}$. We say $S$ satisfies the $(\eps, \delta, \ell)$ JL moment property if for all $x \in \R^d$ with $\|x\|_2 = 1$, $\mathbb{E}_{S \sim \calD}[|\|Sx\|_2^2 - 1|^\ell] \leq \eps^\ell \delta$.
\end{definition}

The $(\eps, \delta, \ell)$-JL moment property implies the approximate matrix product property:

\begin{theorem}[Theorem 21 of \cite{kn14_sparser_jl_transforms}]
\label{thm:jl_property_implies_approx_matrix_product}
Suppose a distribution $\calD$ over random matrices $S \in \R^{r \times n}$ satisfies the $(\eps, \delta, \ell)$-JL moment property, for some $\ell \geq 2$, and $\eps, \delta \in (0, 1/2)$. Then it satisfies the $(3\eps, \delta)$-approximate matrix product property.
\end{theorem}

The $(\eps, \delta, \ell)$-JL moment property is also preserved under certain operations:

\begin{lemma}[Lemma 14 of \cite{akkpvwz20_tensor_sketch_outer_product}]
\label{lemma:jl_moment_with_identity_matrix}
Let $r, d, k \in \N$ and suppose $S \in \R^{r \times d}$ has the $(\eps, \delta, \ell)$-JL moment property for some $\eps, \delta \in (0, 1)$ and $\ell \in \N$. Then, $S \times I_k$ has the $(\eps, \delta, \ell)$-JL moment property, where $I_k$ is the $k \times k$ identity matrix.
\end{lemma}

\begin{lemma}[Lemma 15 from \cite{akkpvwz20_tensor_sketch_outer_product}]
\label{lemma:jl_moment_composition_of_jl_maps}
Let $\eps, \delta \in (0, 1)$ and $k \in \N$. Suppose for $i \in [k]$ that $M_i \in \R^{d_{i + 1} \times d_i}$ has the $(\frac{\eps}{\sqrt{2k}}, \delta, \ell)$-JL moment property, and in addition suppose that the $M_i$'s are independent. Then, $M_k M_{k - 1} \ldots M_1$ has the $(\eps, \delta, \ell)$-JL moment property.
\end{lemma}

\begin{lemma}[Corollary of Lemmas \ref{lemma:jl_moment_with_identity_matrix} and \ref{lemma:jl_moment_composition_of_jl_maps}]
\label{lemma:jl_moment_kronecker_products}
Let $\eps, \delta \in (0, 1)$, and $\ell \in \N$. Suppose $r_1, r_2, d_1, d_2 \in \N$, and suppose $S \in \R^{r_1 \times d_1}$ and $T \in \R^{r_2 \times d_2}$ each have the $(\eps/2, \delta, \ell)$-JL moment property. Then, $S \times T$ has the $(\eps, \delta, \ell)$-JL moment property.
\end{lemma}
\begin{proof}
If $S$ and $T$ each have the $(\eps/2, \delta, \ell)$-JL moment property, then $S \times I_{r_2}$ and $I_{d_1} \times T$ have the $(\eps/2, \delta, \ell)$ JL moment property by Lemma \ref{lemma:jl_moment_with_identity_matrix}, meaning that $S \times T = (S \times I_{r_2})(I_{d_1} \times T)$ has the $(\eps, \delta, \ell)$-JL moment property by Lemma \ref{lemma:jl_moment_composition_of_jl_maps}.
\end{proof}

\begin{definition}[Subspace Embeddings --- Definition 2.1 of \cite{woodruff14_sketching_as_a_tool}]
Let $A \in \R^{n \times d}$. Then, a matrix $S \in \R^{r \times n}$ is a $(1 \pm \eps)$ $\ell_2$ subspace embedding for $A$ if, for all $x \in \R^d$, $\|SAx\|_2 = (1 \pm \eps) \|Ax\|_2$.
\end{definition}

\begin{definition}[Affine Embedding --- Theorem 7.10 of \cite{cw17_lra_regression_nnz_time}]
Let $A \in \R^{n \times r}$ and $B \in \R^{n \times d}$, and let $\eps > 0$. Then, we say $S$ is an $\eps$-affine embedding for $A$ and $B$ if, for all $X \in \R^{r \times d}$, $\|S(AX - B)\|_F = (1 \pm \eps) \|AX - B\|_F$.
\end{definition}

\begin{lemma}[JL Moment Property implies Affine Embedding]
\label{lemma:jl_moment_implies_affine_embedding}
Let $A \in \R^{n \times k}$ and $B \in \R^{n \times d}$. Suppose $S \in \R^{r \times n}$ is a random matrix with the $(\frac{\eps}{3k}, \frac{\delta}{3}, \ell)$-JL moment property for some integer $\ell \geq 2$. Then, $S$ is an $\eps$-affine embedding for $A$ and $B$ with probability at least $1 - \delta$.
\end{lemma}
\begin{proof}
By Theorem 7.10 of \cite{cw17_lra_regression_nnz_time}, it suffices to have the following:
\begin{itemize}
    \item $S$ should have the $(\frac{\eps}{\sqrt{k}}, \frac{\delta}{3})$ approximate matrix product property.
    \item $S$ should be a $(1 \pm \eps)$ $\ell_2$ subspace embedding for $A$, with probability at least $1 - \frac{\delta}{3}$.
    \item $\|S(AX^* - B)\|_F^2 = (1 \pm \eps) \|AX^* - B\|_F^2$, where $X^* = \argmin_X \|AX - B\|_F$, with probability at least $1 - \frac{\delta}{3}$.
\end{itemize}
Note that by Theorem \ref{thm:jl_property_implies_approx_matrix_product}, if $S$ satisfies the $(\frac{\eps}{3\sqrt{k}}, \frac{\delta}{3}, \ell)$ JL moment property for some $\ell \geq 2$, then $S$ will have the $(\frac{\eps}{\sqrt{k}}, \frac{\delta}{3})$ approximate matrix product property. This is true because $\forall k\in \N, k\geq \sqrt{k}, \frac{\eps}{\sqrt{k}}\leq \frac{\eps}{k}$. In addition, for $S$ to be a $(1 \pm \eps)$ $\ell_2$ subspace embedding for $A$, it suffices to have $\|U^T S^T S U - I\|_2 \leq \eps$ where $U$ is an orthonormal basis for the column span of $A$. However, if $S$ has the $(\tau, \frac{\delta}{3})$ approximate matrix product property for some $\tau > 0$, then
$$\|U^T S^T S U - I\|_2 \leq \|U^T S^T S U - I\|_F \leq \tau \|U\|_F^2 \leq \tau k$$
with probability at least $1 - \frac{\delta}{3}$, and if this holds, then $S$ is a subspace embedding for $A$. Thus, for $S$ to be a subspace embedding for $A$ with probability at least $1 - \frac{\delta}{3}$, it is enough for $S$ to have the $(\frac{\eps}{k}, \frac{\delta}{3})$ approximate matrix product property, and this holds as long as $S$ has the $(\frac{\eps}{3k}, \frac{\delta}{3}, \ell)$ JL moment property by Theorem \ref{thm:jl_property_implies_approx_matrix_product}. Finally, if $S$ has the $(\eps, \frac{\delta}{3}, \ell)$ JL moment property, then by Lemma \ref{lemma:jl_moment_with_identity_matrix}, $S \times I_d$ also has the $(\eps, \frac{\delta}{3}, \ell)$ JL moment property. Therefore, if we let $\widetilde{b} = \tensorvec(AX^* - B)$, then
$$\mathbb{E}[|\|S(AX^* - B)\|_F^2 - \|AX^* - B\|_F^2|^\ell]  = E[|\|(S \times I_d)\widetilde{b}\|_2^2 - \|\widetilde{b}\|_2^2|^\ell] \leq \frac{\eps^\ell \delta \|\widetilde{b}\|_2^{2\ell}}{3} = \frac{\eps^\ell \delta \|AX^* - B\|_F^{2\ell}}{3}$$
Here the first inequality is because $S \times I_d$ has the $(\eps, \frac{\delta}{3}, \ell)$ JL moment property. Therefore, by Markov's inequality,
\begin{equation}
\begin{split}
\PrBig{|\|S(AX^* - B)\|_F^2 &- \|AX^* - B\|_F^2| \leq  \eps \|AX^* - B\|_F^2} \\
& \leq \PrBig{|\|S(AX^* - B)\|_F^2 - \|AX^* - B\|_F^2|^\ell \leq \eps^\ell \|AX^* - B\|_F^{2\ell}} \\
& \leq \frac{\mathbb{E}[|\|S(AX^* - B)\|_F^2 - \|AX^* - B\|_F^2|^\ell]}{\eps^\ell \|AX^* - B\|_F^{2\ell}} \\
& \leq \frac{\eps^\ell \delta \|AX^* - B\|_F^{2\ell}/3}{\eps^\ell \|AX^* - B\|_F^{2\ell}} \\
& = \frac{\delta}{3}
\end{split}
\end{equation}
and therefore, with probability at least $1 - \frac{\delta}{3}$, $\|S(AX^* - B)\|_F^2 = (1 \pm \eps) \|AX^* - B\|_F^2$. In summary, if $S$ has the $(\frac{\eps}{3k}, \frac{\delta}{3}, \ell)$ JL moment property, then $S$ satisfies each of the properties in the statement of Theorem 7.10 of \cite{cw17_lra_regression_nnz_time} with probability at least $1 - \frac{\delta}{3}$ each. Thus, with probability at least $1 - \delta$, $S$ is an $\eps$-affine embedding for $A$ and $B$.
\end{proof}

\subsubsection{Countsketch Matrices} \label{subsubsec:countsketch}

Let us also recall properties of Countsketch, a sketch which can be applied to matrices in input-sparsity time.

\begin{definition}[Countsketch --- Section 3 of \cite{Countsketch2004}]
A matrix $S\in\R^{r\times n}$ is a Countsketch matrix if it is constructed in the following fashion: For each column $S(:, j)$, sample a row number $i \in [r]$ uniformly at random. Assign $\{1, -1\}$ to $S_{ij}$  with equal probability. The remaining entries are $0$.
\end{definition}

If a matrix $A \in \R^{n \times d}$ is represented as a list of its nonzero entries, and $S \in \R^{r \times n}$ is a Countsketch matrix, then $SA$ can be computed in $O(\nnz(A))$ time. \cite{cw17_lra_regression_nnz_time} shows that $SA$ is a $(1\pm \varepsilon) \, \ell_2$ subspace embedding for $A$ with $r = \poly(k/\varepsilon)$ rows. This allows us to use Countsketch to reduce the size of tensor slices to $\poly(k/\varepsilon)$ in several algorithms described below. First note that Countsketch matrices have the JL moment property:

\begin{lemma}[Countsketch has JL Moment Property --- Follows from Lemma 7.11 of \cite{cw17_lra_regression_nnz_time}]
\label{lemma:countsketch_jl_moment_property}
Let $S \in \R^{r \times n}$ be a Countsketch matrix. Then, if $r \in \Omega(\frac{1}{\eps^2 \delta})$, then $S$ has the $(\eps, \delta, 2)$-JL moment property.
\end{lemma}
\begin{proof}
As shown in the proof of Lemma 7.11 in \cite{cw17_lra_regression_nnz_time}, for any matrix $A$ and Countsketch matrix $S$ specified in the Lemma statement, $\Var{\|SA\|_F^2} \leq \frac{6}{r} \|A\|_F^2$, and $\mathbb{E}[\|SA\|_F^2] = \|A\|_F^2$. In particular, for any unit vector $x \in \R^n$, $\mathbb{E}[(\|Sx\|_2^2 - 1)^2] \leq \frac{6}{r}$, and this implies that for $r \in \Omega(\frac{1}{\eps^2 \delta})$, $S$ has the $(\eps, \delta, 2)$-JL moment property.
\end{proof}

As a result, a Countsketch matrix $S$ is a $(1 \pm \eps)$ $\ell_2$ subspace embedding for an $n \times d$ matrix $A$ with $O(\frac{d^2}{\eps^2})$ rows:

\begin{theorem}[Countsketch is a $(1 \pm \eps)$ $\ell_2$ Subspace Embedding --- \cite{nn13_osnap_matrices} Theorem 3, \cite{mm13_input_sparsity_subspace_embeddings} Theorem 1]
\label{thm:countsketch_subspace_embedding}
Let $S \in \R^{r \times n}$ be a Countsketch matrix with $r = O(\frac{d^2}{\eps^2\delta})$, and let $A \in \R^{n \times d}$. Then with probability $1 - \delta$, $S$ is an $\ell_2$ subspace embedding for $A$, i.e. for all $x \in \R^d$, $(1 - \eps) \|Ax\|_2 \leq \|SAx\|_2 \leq (1 + \eps) \|Ax\|_2$.
\end{theorem}
Countsketch also provides an \textit{affine embedding} --- it preserves the error for any solution to \textit{multiple-response} $\ell_2$ regression:

\begin{theorem}[Countsketch is a $(1 \pm \eps)$ $\ell_2$ Affine Embedding --- Theorem 7.13 of \cite{cw17_lra_regression_nnz_time}]
\label{thm:countsketch_affine_embedding}
Let $n, k, d \in \N$, and let $S \in \R^{r \times n}$ be a Countsketch matrix with $r = O(\frac{k^2}{\eps^2 \delta})$. Then, for any $A \in \R^{n \times k}$ and $B \in \R^{n \times d}$, with probability $1 - \delta$, $S$ is a $(1 \pm \eps)$ $\ell_2$ affine embedding for $A$ and $B$.
\end{theorem}
\begin{proof}
This follows from Theorem 7.13 of \cite{cw17_lra_regression_nnz_time}. (Note that Theorem 7.13 of \cite{cw17_lra_regression_nnz_time} gives a bound of $O(\frac{k^2}{\eps^2} \log^7(k/\eps))$ on the required number of rows in $S$, but this can be reduced to $O(\frac{k^2}{\eps^2})$ by using the subspace embedding result of \cite{nn13_osnap_matrices, mm13_input_sparsity_subspace_embeddings}, see Theorem \ref{thm:countsketch_subspace_embedding} above. For the dependence on the failure probability $\delta$, note that $S$ is a $(1 \pm \eps)$ subspace embedding with probability $1 - \delta$ as long as it has at least $O(\frac{k^2}{\eps^2 \delta})$ rows. It also preserves the squared Frobenius norm of a fixed matrix up to a $(1 \pm \eps)$ factor, as long as it has at least $O(\frac{1}{\delta \eps^2})$ rows, by the proof of Lemma 7.11 in \cite{cw17_lra_regression_nnz_time}. Finally, by the proof of Lemma 7.3 in \cite{cw17_lra_regression_nnz_time}, it achieves the desired approximate matrix product property with probability $1 - \delta$ as long as it has at least $O(\frac{1}{\eps^2 \delta})$ rows.
\end{proof}

\subsubsection{Projection-Cost Preserving Sketches (PCPs)}

We now introduce the concept of a projection-cost-preserving sketch (PCP) and prove that Countsketch matrices with $\poly(k/\eps)$ rows/columns can be used as PCPs.

\begin{definition}[Projection-cost Preserving Sketch --- Definition 1 of \cite{cemmp15_pcps}]
\label{def:pcp}
Let $A \in \R^{n \times d}$ and $k \in \N$, $\eps \in (0, 1)$. Then, $\widetilde{A} \in \R^{n \times d'}$ is a rank-$k$ projection-cost preserving sketch of $A$ if for all orthogonal projection matrices $P$ onto subspaces of rank at most $k$, 
$$(1 - \eps) \|A - PA\|_F^2 \leq \|\widetilde{A} - P\widetilde{A}\|_F^2 + c \leq (1 + \eps) \|A - PA\|_F^2$$
for some scalar $c \geq 0$ that is independent of $P$. We also say $\widetilde{A} \in \R^{n \times d'}$ is an $(\eps, c, k)$-projection-cost preserving sketch or $(\eps, c, k)$-PCP of $A$ (as in Definition 1 of \cite{MuscoCameronPCP}).
\end{definition}

Dense Gaussian matrices can be used to obtain $(\eps, 0, k)$-PCPs:

\begin{theorem}[Corollary 7 of \cite{MuscoCameronPCP}]
Let $A \in \R^{n \times d}$ and $k \in \N$, and let $\eps, \delta \in (0, 1)$. If $S \in \R^{m \times d}$ has i.i.d. $\mathcal{N}(0, 1/m)$ entries, where $m \in \Omega(\frac{k + \log(1/\delta)}{\eps^2})$, then $AS^T$ is an $(\eps, 0, k)$-projection cost preserving sketch of $A$ with probability at least $1 - \delta$.
\end{theorem}

Countsketch matrices can also be used to obtain $(\eps, 0, k)$-PCPs:

\begin{theorem}[Corollary of Theorem 2 of \cite{MuscoCameronPCP}]
\label{thm:countsketch_pcps_zero_additive}
Let $A \in \R^{n \times d}$ and $k \in \N$, and let $\eps, \delta \in (0, 1)$. If $S \in \R^{m \times d}$ is a Countsketch matrix, where $m \in \Omega(\frac{k^2}{\eps^2 \delta})$, then $AS^T$ is an $(\eps, 0, k)$-projection cost preserving sketch of $A$ with probability at least $1 - \delta$.
\end{theorem}
\begin{proof}
This follows from Theorem 2 of \cite{MuscoCameronPCP}. For the first property of that theorem, note that $S$ is an $\frac{\eps}{3}$ subspace embedding for $A_k$ with probability $1 - O(\delta)$ as long as it has at least $\Omega(\frac{k^2}{\eps^2 \delta})$ rows. For the second and third properties, note that by the proof of Lemma 7.3 in \cite{cw17_lra_regression_nnz_time}, $S$ achieves the desired approximate matrix product property with probability $1 - O(\delta)$ as long as it has at least $O(\frac{k}{\eps^2 \delta})$ rows. Finally, by the proof of Lemma 7.11 in \cite{cw17_lra_regression_nnz_time}, it achieves the desired norm preservation property as long as it has at least $O(\frac{1}{\eps^2 \delta})$ rows.
\end{proof}

It was also previously noted by \cite{cemmp15_pcps} (Lemma 11 and Theorem 12 of that work) that Countsketch matrices could be used to obtain PCPs.

The following lemma is a consequence of the definition of a PCP:

\begin{lemma}[Lemma 3 of \cite{cemmp15_pcps}] \label{lemma:pcp_preserving_approximation_error}
Let $A \in \R^{n \times d}$, $k \in \N$, $\eps \in (0, 1)$. Let $\widetilde{A} \in \R^{n \times d'}$ be a rank-$k$ projection-cost preserving sketch for $A$. Let $P^* = \argmin_P \|A - PA\|_F$ and $\widetilde{P}^* = \argmin_P \|\widetilde{A} - P\widetilde{A}\|_F$, where the argmins are taken over all rank-$k$ orthogonal projection matrices. Suppose $\widetilde{P} \in \R^{n \times n}$ is a rank-$k$ orthogonal projection matrix such that $\|\widetilde{A} - \widetilde{P}\widetilde{A}\|_F^2 \leq \gamma \|\widetilde{A} - \widetilde{P}^* \widetilde{A}\|_F^2$, for some $\gamma \geq 1$. Then, $\|A - \widetilde{P}A\|_F^2 \leq \frac{\gamma \cdot (1 + \eps)}{(1 - \eps)} \|A - P^* A\|_F^2$.
\end{lemma}

\subsubsection{Multiple-Response $\ell_2$ Regression Using Sign Matrices} \label{subsubsec:multiple_response_l2_regression_sign_matrices}

We will use the following result to efficiently perform multiple-response $\ell_2$ regression:

\begin{theorem}[Theorem 3.1 of \cite{cw09_numerical_linear_algebra_streaming}]
\label{theorem:original_sign_matrix_multiple_regression}
Let $\eps, \delta \in (0, 1)$. Let $A \in \R^{n \times d_1}$ and $B \in \R^{n \times d_2}$, with $A$ of rank at most $k$. Let $S \in \R^{m \times n}$ be a matrix whose entries are each chosen from $\{-\frac{1}{\sqrt{m}}, \frac{1}{\sqrt{m}}\}$ uniformly at random, with $m = O(k \log(1/\delta)/\eps)$. If $X^* = \argmin_X \|AX - B\|_F$, and $\widehat{X} = \argmin_X \|SAX - SB\|_F$, then with probability $1 - O(\delta)$, $\|A\widehat{X} - B\|_F \leq (1 + \eps) \|AX^* - B\|_F$. 
\end{theorem}

We will also make use of the following version of Theorem 3.1 of \cite{cw09_numerical_linear_algebra_streaming}.

\begin{theorem}[Strenghtened Version of Theorem 3.1 of \cite{cw09_numerical_linear_algebra_streaming}]
\label{thm:approximate_multiple_l2_regression}
Let $\eps, \delta, \tau \in (0, 1)$. Let $A \in \R^{n \times d_1}$, $B \in \R^{n \times d_2}$, with $A$ of rank at most $k$. Let $S \in \R^{m \times n}$ be a random matrix such that
\begin{itemize}
    \item $S$ is a $(1 \pm \frac{1}{3})$ $\ell_2$ subspace embedding for the column span of $A$ with probability at least $1 - \delta$.
    \item $S$ has the $(\sqrt{\eps/k}, \delta)$-approximate matrix product property.
    \item $E[S_i^T S_j]$ is $1$ if $i = j$ and $0$ otherwise, where $S_i$ denotes the $i^{th}$ column of $S$.
\end{itemize}
Then, with probability $1 - O(\delta)$, if $X^* = \argmin_X \|AX - B\|_F$, then for all $\widehat{X} \in \R^{d_1 \times d_2}$, such that $\|SA\widehat{X} - SB\|_F \leq (1 + \tau) \min_X \|SAX - SB\|_F$, it holds that $\|A\widehat{X} - B\|_F \leq (1 + O(\eps) + O(\tau/\delta)) \|AX^* - B\|_F$. In particular, this holds if $S \in \R^{m \times n}$ is a matrix whose entries are each chosen from $\{-\frac{1}{\sqrt{m}}, \frac{1}{\sqrt{m}}\}$ uniformly at random, with $m = O(k \log(1/\delta)/\eps)$.
\end{theorem}

\begin{remark}
We will apply the above theorem setting $\tau = \eps$ and $\delta = \frac{1}{10}$ --- thus, we do not optimize the $O(\tau/\delta)$ term.
\end{remark}

\begin{proof}
The above theorem is essentially the same as Theorem 3.1 of \cite{cw09_numerical_linear_algebra_streaming}, with the difference being that in Theorem 3.1 of \cite{cw09_numerical_linear_algebra_streaming}, it was stated that if $\widetilde{X}$ is the minimizer of $\|SAX - SB\|_F$, then $\|A\widetilde{X} - B\|_F \leq (1 + \eps) \|AX^* - B\|_F$, while we modify the proof slightly to show that this also holds with $(1 + \eps)$-approximate solutions $\widehat{X}$ in the place of $\widetilde{X}$, and state it slightly more generally to allow us to apply this theorem with matrices $S$ which are not sign matrices. We give a proof for completeness.

Let $\widehat{X} \in \R^{d_1 \times d_2}$ such that $\|SA\widehat{X} - SB\|_F \leq (1 + \tau) \min_X \|SAX - SB\|_F$. In addition, define $\widetilde{X} := \argmin_X \|SAX - SB\|_F$. Finally, let $U \in \R^{n \times k}$ be an orthonormal basis for the column span of $A$. As in \cite{cw09_numerical_linear_algebra_streaming}, first let us show that $\|A(\widetilde{X} - X^*)\|_F \leq O(\sqrt{\eps})\|AX^* - B\|_F$. We can upper bound $\|A(\widetilde{X} - X^*)\|_F$ as follows:
\begin{equation}
\begin{split}
\|A(\widetilde{X} - X^*)\|_F
& = \|UU^T A(\widetilde{X} - X^*)\|_F \\
& = \|U^T A (\widetilde{X} - X^*)\|_F \\
& = \|U^T S^T S U \cdot U^T A(\widetilde{X} - X^*) + (I - U^T S^T S U) \cdot U^T A (\widetilde{X} - X^*)\|_F \\
& \leq \|U^T S^T S U \cdot U^T A(\widetilde{X} - X^*)\|_F + \|(I - U^T S^T S U) \cdot U^T A (\widetilde{X} - X^*)\|_F \\
& \leq \|U^T S^T S U \cdot U^T A(\widetilde{X} - X^*)\|_F + \|I - U^T S^T S U\|_2 \|U^T A (\widetilde{X} - X^*)\|_F \\
& \leq \|U^T S^T S A(\widetilde{X} - X^*)\|_F + \frac{1}{3} \|A(\widetilde{X} - X^*)\|_F
\end{split}
\end{equation}
Here, the first inequality is by the triangle inequality, and the second inequality is by $\|CD\|_F \leq \|C\|_2 \|D\|_F$ with $C = I - U^T S^T S U$ and $D = U^T A (\widetilde{X} - X^*)$. The last inequalitiy holds with probability at least $1 - O(\delta)$ by our first assumption on $S$. Rearranging gives
\begin{equation} \label{eq:first_bound_3.1}
\begin{split}
\|A(\widetilde{X} - X^*)\|_F \leq \frac{3}{2} \|U^T S^T S A (\widetilde{X} - X^*)\|_F
\end{split}
\end{equation}
Thus, it suffices to obtain a bound on $\|U^T S^T SA (\widetilde{X} - X^*)\|_F$. First, observe that by the normal equations for the problem $\min_X \|SAX - SB\|_F$,
$$U^T S^T S (A\widetilde{X} - B) = A^T S^T S (A\widetilde{X} - B) = 0$$
and therefore,
$$U^T S^T S A(\widetilde{X} - X^*) = U^T S^T S (A\widetilde{X} - B + B - AX^*) = U^T S^T S(B - AX^*)$$
Combining this with Equation \ref{eq:first_bound_3.1} and using the fact that $S$ has the $(\sqrt{\eps/k}, \delta)$ approximate matrix product property, we find that
\begin{equation}
\begin{split}
\|A (\widetilde{X} - X^*)\|_F
& \leq \frac{3}{2} \|U^T S^T S (B - AX^*)\|_F \\
& \leq \frac{3}{2} \cdot \sqrt{\frac{\eps}{k}} \|U\|_F \|B - AX^*\|_F \\
& = O(\sqrt{\eps}) \|B - AX^*\|_F
\end{split}
\end{equation}
Here the second inequality holds with probability at least $1 - \delta$ since $S$ has the $(\sqrt{\eps/k}, \delta)$ approximate matrix product property, and the first equality is because $\|U\|_F = \sqrt{k}$.

Next, let us bound $\|A(\widehat{X} - \widetilde{X})\|_F$. By the Pythagorean theorem and the normal equations for the problem $\min_X \|SAX - SB\|_F$,
$$\|SA\widehat{X} - SB\|_F^2 = \|SA\widetilde{X} - SB\|_F^2 + \|SA(\widetilde{X} - \widehat{X})\|_F^2$$
Therefore,
\begin{equation}
\begin{split}
\|SA(\widehat{X} - \widetilde{X})\|_F^2
& = \|SA\widehat{X} - SB\|_F^2 - \|SA\widetilde{X} - SB\|_F^2 \\
& \leq (1 + \tau) \|SA\widetilde{X} - SB\|_F^2 - \|SA\widetilde{X} - SB\|_F^2 \\
& = \tau \|SA\widetilde{X} - SB\|_F^2 \\
& \leq \tau \|SA X^* - SB\|_F^2 \\
& \leq \frac{\tau}{\delta} \|AX^* - B\|_F^2
\end{split}
\end{equation}
Here the first inequality is because $\widehat{X}$ is a $(1 + \tau)$-approximate solution to the problem $\min_X \|SAX - SB\|_F$. The second inequality is because $\widetilde{X} = \argmin_X \|SAX - SB\|_F$. Finally, the third inequality holds with probability $1 - \delta$ by Markov's inequality and the fact that for any matrix $M \in \R^{n \times c}$, $E[\|SM\|_F^2] = E[\Tr(M^T S^T S M)] = \Tr(M^T M) = \|M\|_F^2$, where we used the fact that $E[S_i^T S_j]$ is $1$ if $i = j$ and $0$ otherwise.

In summary,
$$\|SA(\widehat{X} - \widetilde{X})\|_F^2 \leq \frac{\tau}{\delta} \|AX^* - B\|_F^2$$
with probability $1 - \delta$, and since $S$ is a subspace embedding for $A$ with probability $1 - \delta$, this implies that
$$\|A(\widehat{X} - \widetilde{X})\|_F \leq O\Big(\sqrt{\frac{\tau}{\delta}}\Big) \|AX^* - B\|_F$$
Thus, by the triangle inequality,
\begin{equation}
\begin{split}
\|A(\widehat{X} - X^*)\|_F
& \leq \|A(\widehat{X} - \widetilde{X})\|_F + \|A(\widetilde{X} - X^*)\|_F \\
& \leq O\Big(\sqrt{\frac{\tau}{\delta}} + \sqrt{\eps}\Big) \|AX^* - B\|
\end{split}
\end{equation}
Finally, by the normal equations for the original problem $\min_X \|AX - B\|_F$, the columns of $A$ are orthogonal to those of $AX^* - B$, and by the Pythagorean theorem,
\begin{equation}
\begin{split}
\|A\widehat{X} - B\|_F^2
& \leq \|A(\widehat{X} - X^*)\|_F^2 + \|AX^* - B\|_F^2 \\
& \leq O\Big(\sqrt{\eps} + \sqrt{\frac{\tau}{\delta}}\Big)^2 \|AX^* - B\|_F^2 + \|AX^* - B\|_F^2 \\
& \leq (1 + O(\eps) + O(\tau/\delta)) \|AX^* - B\|_F^2
\end{split}
\end{equation}
with probability $1 - O(\delta)$. Taking square roots gives
$$\|A\widehat{X} - B\|_F \leq \sqrt{1 + O(\eps) + O(\tau/\delta)} \|AX^* - B\|_F \leq (1 + O(\eps) + O(\tau/\delta)) \|AX^* - B\|_F$$
as desired.

We additionally show that this holds if $S \in \R^{m \times n}$ is a random matrix whose entries are each chosen from $\{-\frac{1}{\sqrt{m}}, \frac{1}{\sqrt{m}}\}$ uniformly at random, with $m = O(k \log(1/\delta)/ \eps)$. The subspace embedding property holds due to the following lemma (which can be applied with $\eps = \frac{1}{3}$):

\begin{lemma}[Lemma 3.4 of \cite{cw09_numerical_linear_algebra_streaming}]
\label{lemma:sign_matrix_subspace_embedding}
Let $k \in \N$ and $\eps, \delta > 0$, and $m = O(k \log(1/\delta)/\eps^2)$. There is an absolute constant $\eta > 0$ such that if $S \in \R^{m \times n}$ has entries from $\{-\frac{1}{\sqrt{m}}, \frac{1}{\sqrt{m}}\}$ (which are at least $\eta(k + \log(1/\delta))$-wise independent), then for any $U \in \R^{n \times k}$ with orthonormal columns, with probability at least $1 - \delta$, $\|U^T S^T S U - I\|_2 \leq \eps$.
\end{lemma}

\begin{remark}
Note that there is a typo in Lemma 3.4 of \cite{cw09_numerical_linear_algebra_streaming}, which should say $m = O(k \log(1/\delta)/\eps^2)$. This is because it applies Theorem 2.2 with $\|A\|_F = \|B\|_F = 1$, and $\delta$ replaced by $\frac{\delta}{\eta^k}$ for some constant $\eta > 1$. This does not affect the overall proof, since when applying Lemma 3.4, it is enough to take $\eps = \frac{1}{3}$.
\end{remark}

In addition, the approximate matrix product property is due to the following theorem from \cite{cw09_numerical_linear_algebra_streaming}:

\begin{theorem}[Theorem 2.2 of \cite{cw09_numerical_linear_algebra_streaming}]
\label{thm:sign_matrix_approx_matrix_product}
Let $\delta, \eps \in (0, 1)$, and $A \in \R^{n \times a}$ and $B \in \R^{n \times b}$. Then, for $m = \Theta(\log(1/\delta)/\eps^2)$ and $S \in \R^{m \times n}$ having entries drawn uniformly at random from $\{-\frac{1}{\sqrt{m}}, \frac{1}{\sqrt{m}}\}$ (which are at least $C \log(1/\delta)$-wise independent for some constant $C > 0$), $\PrBig{\|A^T S^T S B - A^T B\|_F \leq \eps \|A\|_F \|B\|_F} \geq 1 - \delta$.
\end{theorem}

\end{proof}

We will also use the following version of the above result to perform multiple-response $\ell_2$ regression, with sketching matrices other than dense sign matrices. This version will be used when we want very low failure probability (and thus cannot afford the $O(\tau/\delta)$ term in the approximation error) but do not need $(1 + \eps)$-approximate solutions of the sketched problem to remain good solutions.

\begin{theorem}[Essentially Theorem 3.1 of \cite{cw09_numerical_linear_algebra_streaming}]
\label{thm:general_l2_regression_sketching}
Let $\eps, \delta \in (0, 1/2)$. Let $A \in \R^{n \times d_1}$, $B \in \R^{n \times d_2}$, with $A$ of rank at most $k$. Let $S \in \R^{m \times n}$ be a random matrix which has the $(\sqrt{\frac{\eps}{k}}, \frac{\delta}{3})$ approximate matrix product property, and is a $(1 \pm \frac{1}{3})$ $\ell_2$ subspace embedding for $A$ with probability at least $1 - \delta$. If $X^* = \argmin_X \|AX - B\|_F$, and $\widehat{X} = \argmin_X \|SAX - SB\|_F$, then with probability $1 - O(\delta)$, $\|A\widehat{X} - B\|_F \leq (1+O(\varepsilon)) \|AX^* - B\|_F$.
\end{theorem}

\begin{proof}
The proof is essentially the same as that of Theorem 3.1 of \cite{cw09_numerical_linear_algebra_streaming}, with slight modifications --- we give the proof for completeness. First, if $U \in \R^{n \times k}$ is an orthonormal basis for the column span of $A$, then with probability at least $1 - \delta$, $\|U^T S^T S U - I\|_2 \leq \frac{1}{3}$ (by the definition of subspace embedding). Now, let $\widehat{X} = \argmin_X \|SAX - SB\|_F$. As before,
\begin{equation}
\begin{split}
\|A(\widehat{X} - X^*)\|_F
& = \|UU^T A (\widehat{X} - X^*)\|_F \\
& = \|U^T A (\widehat{X} - X^*)\|_F \\
& = \|U^T S^T S U \cdot U^T A(\widehat{X} - X^*) + (I - U^T S^T S U) \cdot U^T A(\widehat{X} - X^*)\|_F \\
& \leq \|U^T S^T S U \cdot U^T A (\widehat{X} - X^*)\|_F + \|(I - U^T S^T S U) \cdot U^T A(\widehat{X} - X^*)\|_F \\
& \leq \|U^T S^T S U U^T A (\widehat{X} - X^*)\|_F + \frac{1}{3} \|U^T A(\widehat{X} - X^*)\|_F \\
& = \|U^T S^T S U U^T A (\widehat{X} - X^*)\|_F + \frac{1}{3} \|A(\widehat{X} - X^*)\|_F
\end{split}
\end{equation}
Here the second inequality is because $S$ is a subspace embedding for $U$. Thus,
\begin{equation}
\begin{split}
\|A(\widehat{X} - X^*)\|_F \leq \frac{3}{2}\|U^T S^T S UU^T A(\widehat{X} - X^*)\|_F = \frac{3}{2}\|U^T S^T S A(\widehat{X} - X^*)\|_F
\end{split}
\end{equation}
where the equality is because $UU^T A = A$. Now, by the normal equations for the sketched problem $\min_X \|SAX - SB\|_F$,
$$U^T S^T S(A\widehat{X} - B) = A^T S^T S A(\widehat{X} - B) = 0$$
meaning that
$$U^T S^T S A(\widehat{X} - X^*) = U^T S^T S(B - A\widehat{X}) + U^T S^T S A(\widehat{X} - X^*) = U^T S^T S(B - AX^*)$$
Therefore,
\begin{equation}
\begin{split}
\|A(\widehat{X} - X^*)\|_F
& \leq \frac{3}{2} \|U^T S^T S A(\widehat{X} - X^*)\|_F \\
& = \frac{3}{2} \|U^T S^T S (B - AX^*)\|_F \\
& \leq \frac{3}{2} \sqrt{\frac{\eps}{k}} \|U\|_F \|B - AX^*\|_F \\
& \leq \frac{3\sqrt{\eps}}{2} \|B - AX^*\|_F
\end{split}
\end{equation}
Here the second inequality holds with probability at least $1 - \frac{\delta}{3}$ since $S$ has the $(\sqrt{\frac{\eps}{k}}, \frac{\delta}{3})$ approximate matrix product property. Thus,
\begin{equation}
\begin{split}
\|A\widehat{X} - B\|_F^2
& = \|A(\widehat{X} - X^*)\|_F^2 + \|AX^* - B\|_F^2 \\
& \leq O(\eps) \|AX^* - B\|_F^2 + \|AX^* - B\|_F^2 \\
& = (1 + O(\eps)) \|AX^* - B\|_F^2
\end{split}
\end{equation}
Here the first equality is by the Pythagorean theorem together with the normal equations for the problem $\min_X \|AX - B\|_F$. This completes the proof.
\end{proof}

We will apply the above lemma with sketching matrices of the form $S \times I$ where $S$ is a dense sign matrix. To show that $S \times I$ is a subspace embedding we will use the following lemma:

\begin{lemma}
\label{lemma:kronecker_product_subspace_embedding}
Let $\eps_1, \eps_2 > 0$, and let $n_1, r_1, s_1, n_2, r_2, s_2 \in \N$ and suppose $S \in \R^{s_1 \times n_1}$ is a $(1 \pm \eps_1)$ $\ell_2$ subspace embedding for $U$. In addition, let $V \in \R^{n_2 \times r_2}$ and suppose $T \in \R^{s_2 \times n_2}$ is a $(1 \pm \eps_2)$ $\ell_2$ subspace embedding for $V$. Then, $S \times T$ is a $(1 \pm \eps_1)(1 \pm \eps_2$ $\ell_2$ subspace embedding for $U \times V$, i.e. 
$$\|(S \times T)(U \times V)x\|_2 = (1 \pm \eps_1)(1 \pm \eps_2) \|(U \times V)x\|_2$$
for all $x \in \R^{r_1r_2}$.
\end{lemma}
\begin{proof}
Let $x \in \R^{r_1r_2}$, and let $X \in \R^{r_1 \times r_2}$ so that $X_{ij}$ is equal to the entry of $x$ in the coordinate $(i, j)$. Then, 
\begin{equation}
\begin{split}
\|(S \times T)(U \times V)x\|_2 = \|(SU \times TV)x\|_2 = \|SU X V^T T^T\|_2  \\ = (1 \pm \eps_1)(1 \pm \eps_2) \|UXV^T \|_2 = (1 \pm \eps_1)(1 \pm \eps_2) \|(U \times V)x\|_2
\end{split}
\end{equation}
where the third equality is because $S$ is a subspace embedding for $U$, and the fourth equality is because $T$ is a subspace embedding for $V$. This proves the lemma.
\end{proof}

In addition, we will need the fact that dense sign matrices have the JL moment property.

\begin{lemma}[Sign Matrices Have JL Moment Property --- Lemma 2.3 of \cite{cw09_numerical_linear_algebra_streaming}]
\label{lemma:sign_matrix_jl_moment_property}
Let $\eps, \delta \in (0, 1)$, and let $S \in \R^{m \times n}$ have entries which are i.i.d. and uniformly random in $\{-\frac{1}{\sqrt{m}}, \frac{1}{\sqrt{m}}\}$. If $m \in \Omega(\frac{\log(1/\delta)}{\eps^2})$, then $S$ has the $(\eps, \delta, \log(1/\delta))$-JL moment property.
\end{lemma}
\begin{proof}
By Lemma 2.3 of \cite{cw09_numerical_linear_algebra_streaming}, if $S \in \R^{m \times n}$ is a rescaled dense sign matrix, then for any unit vector $x \in \R^n$, and $p \in \N$,
$$\left(\mathbb{E}[|\|Sx\|_2^2 - 1|^{2p}]\right)^{1/p} \leq O\Big(\frac{p}{m}\Big)$$
Thus, for any $p \in \N$,
$$\left(\mathbb{E}[|\|Sx\|_2^2 - 1|^p]\right)^{1/p} \leq (\mathbb{E}[|\|Sx\|_2^2 - 1|^{2p}])^{1/(2p)} \leq O\Big(\sqrt{\frac{p}{m}}\Big)$$
Thus, if $m \geq \Theta(\frac{\log(1/\delta)}{\eps^2})$, then for all $p \in \N$,
$$\left(\mathbb{E}[|\|Sx\|_2^2 - 1|^p]\right)^{1/p} \leq O\Big(\eps \sqrt{\frac{p}{\log(1/\delta)}}\Big)$$
Thus, $S$ in fact has the Strong $(\eps, \delta)$-JL Moment Property (see Definition 19 in \cite{akkpvwz20_tensor_sketch_outer_product}). By Remark 1 of \cite{akkpvwz20_tensor_sketch_outer_product}, $S$ therefore has the $(\eps, \delta, \log(1/\delta))$-JL moment property as long as it has at least $\Theta(\frac{\log(1/\delta)}{\eps^2})$ rows.
\end{proof}


\subsubsection{Kronecker Product Regression in $\ell_2$ Norm} \label{subsubsec:kronecker_product_regression}

Recall that Tucker-$q$ decomposition is equivalent to finding $q$ factors $U^{1}, \ldots, U^{q}$ and a core tensor $G$ such that $\|(U^{1} \times \ldots \times U^{q})\tensorvec(G) - A\|_2^2$ is minimized. Note that if the factors $U^{1}, \ldots, U^{q}$ are known, then this is an $\ell_2$ regression problem, where the design matrix is a Kronecker product of $q$ $n \times k$ matrices. This $\ell_2$ regression problem can be solved very efficiently:

\begin{theorem}[$\ell_2$ Kronecker Product Regression --- Theorem 3.1 and Algorithm 1 of \cite{djssw19_kronecker_product_regression}]
\label{thm:kronecker_product_regression}
Let $A_1, A_2, \ldots, A_q$, where $A_i \in \R^{n_i \times d_i}$. Let $n = \prod_i n_i$ and $d = \prod_i d_i$. Let $b \in \R^n$. Then, there is an algorithm which, in running time $\sum_{i = 1}^q \nnz(A_i) + \poly(d/(\eps\delta))$ and with success probability $1 - \delta$, returns $\widehat{x} \in \R^d$ such that $\|(A_1 \times \ldots \times A_q)\widehat{x} - b\|_2 \leq (1 + \eps) \min_x \|(A_1 \times \ldots \times A_q)x - b\|_2$, and also returns $\widehat{e} = (1 \pm \eps)\|(A_1 \times \ldots \times A_q)\widehat{x} - b\|_2 - \|b\|_2$ (where we use the notation $a = (1 \pm \eps)b$ to indicate that $a \in [(1 - \eps)b, (1 + \eps)b]$).
\end{theorem}

\begin{proof}
This follows from Theorem 3.1 in \cite{djssw19_kronecker_product_regression} with minor modifications. The main difference is that we desire an algorithm which also returns an estimate of the error $\|(A_1 \times \ldots \times A_q)\widehat{x} - b\|_2$, in addition to returning the solution $\widehat{x}$ itself. To achieve the desired result, it suffices to show that if $D$ is a leverage score sampling matrix as generated by \cite{djssw19_kronecker_product_regression}, then it gives a weak $\eps$-affine embedding, as defined in Theorem 7.10 of \cite{cw17_lra_regression_nnz_time}. Observe that if $D$ has $\poly(d/(\eps\delta))$ nonzero entries, then it is a subspace embedding (as shown in Theorem 2.11 of \cite{woodruff14_sketching_as_a_tool}). In addition, by Proposition 3.6 of \cite{djssw19_kronecker_product_regression}, $D$ satisfies the approximate matrix product property with error parameter $\frac{\eps}{\sqrt{d}}$ with probability $1 - \delta$, as long as $D$ has $\poly(d/(\eps \delta))$ nonzero entries. Thus, both of the hypotheses in Theorem 7.10 of \cite{cw17_lra_regression_nnz_time} are satisfied by $D$ with probability $1 - \delta$, meaning that with probability $1 - \delta$, $S$ is a weak $\eps$-affine embedding for $(A_1 \times \ldots \times A_q, b)$, i.e., for all $x \in \R^d$, $\|D(A_1 \times \ldots \times A_q)x - Db\|_2 - \|Db\|_2 = (1 \pm \eps)\|(A_1 \times \ldots \times A_q)x - b\|_2 - \|b\|_2$. Thus, $\widehat{x}$ can be found using the pseudo-inverse of $D(A_1 \times \ldots \times A_q)$ as in \cite{djssw19_kronecker_product_regression}, and $\widehat{e}$ can be computed as $\|D(A_1 \times \ldots \times A_q)x - Db\|_2 - \|Db\|_2$. Note that the running time needed to compute $\widehat{e}$ is $\poly(d/(\eps\delta))$.
\end{proof}

\subsubsection{$\ell_{1, 2}$-norm Low Rank Approximation}

Several algorithms for $\ell_{1, 2}$-norm low rank approximation are given in \cite{cw15_subspace_approximation}. Based on the techniques of \cite{cw15_subspace_approximation}, \cite{jllmw21_distributed_lp_css} also gives an $O(1)$-approximation algorithm for $\ell_{1, 2}$ low rank approximation with bi-criteria rank $k \cdot \poly(\log k)$. The following is a modified version of the result presented in \cite{jllmw21_distributed_lp_css}, with the OSNAP matrices given in Theorem 9 of \cite{nn13_osnap_matrices} replaced by the sparse embeddings with $O(\frac{d \log d}{\eps^2})$ rows given in Theorem 4.2 of \cite{cohen16_nearly_tight_OSE}.

\begin{theorem}[Based on Theorem 1 of \cite{jllmw21_distributed_lp_css}, Theorem 4.2 of \cite{cohen16_nearly_tight_OSE}]
Let $A \in \R^{n \times d}$ and $k \in \N$. Then, there is an algorithm which outputs $U \in \R^{n \times O(k \log^2 k)}$ and $V \in \R^{O(k \log^2 k) \times d}$, such that $\|UV - A\|_{1, 2} \leq O(1) \min_{B \text{ rank }k} \|B - A\|_{1, 2}$. The running time of this algorithm is $(\nnz(A) + d^2) \cdot k \cdot \poly(\log k)$.
\end{theorem}
\begin{proof}
The proof is the same as that of Theorem 1 in \cite{jllmw21_distributed_lp_css}, with the sparse embedding matrix $S$ in the proof of that theorem now being a sparse embedding matrix with $O(k \log k)$ rows, rather than the large $k \cdot \poly(\log k)$ as in Theorem 9 of \cite{nn13_osnap_matrices}. Note that the bicriteria rank is $O(k \log^2 k)$ since the number of rows in $S$ is $O(k \log k)$, and the number of columns in the Lewis weight sampling matrix $S'$ (in the proof of Theorem 5.8 of \cite{jllmw21_distributed_lp_css}) thus has to be $O(k \log^2 k)$, in order for $S'$ to give an $\ell_1$ subspace embedding for the row span of $SA$.
\end{proof}

\newpage
\section{Approximation Algorithm for Tree Networks} \label{sec:tree_networks}

In this section, we consider the following problem:
\begin{problem} \label{problem:tree_networks}
Let $k, n, q \in \N$. Let $\calT = (V, E)$ be a tree, with vertices $V$ and edges $E$, with $|V| = q$. In addition, suppose that for each $v \in V$, $U_v$ is a tensor with $\deg(v) + 1$ modes. Here, for each edge $e$ incident to $v$, $U_v$ has a mode of dimension $k$, and $U_v$ has one additional mode of dimension $n$ that is not connected to any other tensor. Finally, define $\calT(\{U_v \mid v \in V\}) \in \R^{n \times \ldots \times n}$ to be the $q$-mode tensor obtained by contracting the tensors in $\{U_v \mid v \in V\}$ in the manner specified by the tree $\calT$. Then, given a tensor $A \in \R^{n \times \ldots \times n}$ and a tree $\calT = (V, E)$, we wish to find $U_v$ for $v \in V$ minimizing $\|\calT(\{U_v \mid v \in V \}) - A\|_F$.
\end{problem}

Since the input of the algorithm includes $\calT$, a question that one may ask is \textit{how to instantiate our algorithm for tree network decomposition} --- i.e., how to choose the tree $\calT$ and the corresponding rank $k$ given a tensor $A$. We note that one potential way to do this is by considering whether $A$ has low rank according to a more general tensor network (i.e., a graph which may have cycles) --- we define contraction of general tensor networks in Definition \ref{def:tensor_network_contraction}. This also answers the question of \textit{why tensor contraction networks may be of interest} in addition to other notions of decomposition such as CP, Tucker, and Tensor Train decomposition. For instance, the Tensor Ring decomposition on $q$ modes \cite{zhao2016tensor} corresponds to a cycle with $q$ vertices and aims to overcome various limitations of Tensor Train decompositions, such as the asymmetry of the dimensions of the first and last modes versus the internal ones. Also, the Projected Entangled Pairs States (PEPS) network is used as an ansatz for quantum wavefunctions, i.e., the ground state of two-dimensional Hamiltonians, see, e.g.,  \cite{orus2014practical}. The work of \cite{yl18} discusses general tensor networks further, and gives examples of tensors $A$ which have low Tensor Ring rank but much higher Tensor Train rank. In terms of the space complexity used to represent a network, a sparse network might only need $O(q \cdot n \cdot \poly(k))$ parameters, while other decompositions could use space which is exponential in $q$. In Section \ref{sec:general_networks}, we give an algorithm which, given a tensor $A$, a graph $\mathcal{G}$ and a target rank $k$, outputs a tree $\calT$ and a tensor $M$ which has rank at most $\widetilde{O}(\frac{qk^{\deg(\calG)\tw(\calG)}}{\epsilon} )$ according to $\calT$, such that $\|M - A\|_F \leq (1 + \epsilon) \min_T \|T - A\|_F$, where $T$ ranges over all tensors which have rank $k$ with respect to the graph $\calG$.

\subsection{Warm-up: Improving Previous Results of Bicriteria Algorithm for Tensor Train Decomposition}\label{subsection:improving_hsw_17}
First, we describe the Tensor Train decomposition algorithm of \cite{huber2017randomized}. Suppose $A \in \R^{n \times \ldots \times n}$ is the input tensor. Then, in \cite{huber2017randomized}, $U^q$ is found as follows. First, \cite{huber2017randomized} defines a sketching matrix $G \in \R^{n^{q - 1} \times s_q}$, which is instantiated as a Gaussian matrix (but as we discuss below, could also belong to any other family of oblivious sketching matrices which has the affine embedding property). Then, the rows of $U^q \in \R^{s_q \times n}$ form an orthonormal basis for the column span of $M_q(A) G$. Finally, $A$ is replaced by $A \circ U_q^T \in \R^{n \times \ldots \times n \times s_q}$ (i.e., the first $q - 1$ modes are $n$-dimensional).

In general, to find $U^i$ for $i \geq 2$, \cite{huber2017randomized} proceed as follows. Suppose $U^q \in \R^{s_q \times n}$ and $U^j \in \R^{s_j \times n \times s_{j + 1}}$ (for $j$ between $i + 1$ and $q - 1$) have already been found, the rows of $U^q$ form an orthonormal basis, and the rows of $M_1(U^j)$ form an orthonormal basis for $j$ between $i + 1$ and $q - 1$. Furthermore, suppose that at this point $A \in \R^{n \times \ldots \times n \times s_{i + 1}}$ (and $A$ now has $i + 1$ modes). Then $U^i$ is computed as follows --- first, $G \in \R^{n^{i - 1} \times s_i}$ is again an oblivious sketching matrix (which \cite{huber2017randomized} instantiate as a Gaussian matrix). Then, the rows of $M_1(U^i)$ form an orthonormal basis for the column span of $M_{\{i, i + 1\}}(A) G \in \R^{s_{i + 1} n \times s_i}$. In addition, once $U^i$ is found, $A$ is replaced by $A \circ (U^i)^T \in \R^{n \times \ldots \times n \times s_i}$ (i.e., $A$ has $i$ modes, and in the operation $A \circ (U^i)^T$, the last two modes of $A$ are contracted with $(U^i)^T$). Finally, after $U^2, \ldots, U^q$ have been found, $A \in \R^{n \times s_2}$, and \cite{huber2017randomized} defines $U^1 := A$. This completes our description of their algorithm.

\paragraph{Turning the Algorithm of \cite{huber2017randomized} into a $(1 + \eps)$-Approximation Algorithm.}

In \cite{huber2017randomized}, it is shown that this is a $\sqrt{q - 1}$-approximation algorithm, with a running time of $O(q s^2 \cdot \nnz(A) + qs^3 n)$, where $s = \max_{i \in [q]} s_i$ --- note that the bicriteria rank is also $s$. However, with a different analysis, one can in fact show that the algorithm of \cite{huber2017randomized} is a $(1 + \eps)$-approximation algorithm, with bicriteria rank $r = \poly(qk/\eps)$ and $O(qs^2 \cdot \nnz(A) + n \cdot \poly(qk/\eps))$ running time --- we discuss one such analysis now.

First, suppose $U^1_*, \ldots, U^q_*$ form the optimal Tensor Train decomposition of rank $k$ for $A$. Then, $(U^q_*)^T \in \R^{n \times k}$ is the solution to the multiple-response $\ell_2$ regression problem
$$\min_U \|U M_q(U^1_* \circ \cdots \circ U^{q - 1}_*) - M_q(A)\|_F$$
Then, one can multiply on the right by a sketching matrix $G_q$ which is an $\frac{\eps}{q}$ affine embedding for $M_q(U^1_* \circ \cdots \circ U^{q - 1}_*)$ and $M_q(A)$ with failure probability $O(\delta/q)$, meaning that for all $U \in \R^{n \times k}$, 
$$\|U M_q(U^1_* \circ \cdots \circ U^{q - 1}_*) - M_q(A)\|_F = (1 \pm \eps/q) \|U M_q(U^1_* \circ \cdots \circ U^{q - 1}_*) G_q - M_q(A) G_q\|_F$$
Then, if $(U^q)^T$ is defined to be $M_q(A) G_q (M_q(U^1_* \circ \cdots \circ U^{q - 1}_*) G_q)^\dagger$ (or equivalently, if $(U^q)^T$ is defined to be an orthonormal basis for the column span of $M_q(A) G_q$, and the change of basis matrix $(M_q(U^1_* \circ \cdots \circ U^{q - 1}_*) G_q)^\dagger$ is contracted with $U^{q - 1}_*$) then $U^q$ increases the error achieved by $U^1_* \circ \cdots \circ U^q_*$ by a factor of at most $(1 + O(\eps/q))$. Note that this step is exactly as in the algorithm of \cite{huber2017randomized}, which lets $U^q$ be an orthonormal basis for the row span of $G_q^T M_q(A)^T$.

Now, suppose $U^q, \ldots, U^{i + 1}$ have been found. Then, fixing the factors $U^q, \ldots, U^{i + 1}$ and $U^1_*, \ldots, U^{i - 1}_*$, the optimal solution $U^i$ with the other factors fixed is the solution to the regression problem
$$\min_U \|U (M_i(U^1_* \circ \cdots \circ U^{i - 1}_*) \times M_1(U^i \circ \cdots \circ U^q)) - M_i(A) \|_F$$
Now, multiplying $A$ by the sketching matrix $G$, as done by \cite{huber2017randomized}, corresponds to multiplying the matrices in the above objective on the right by $(G \times I)$, and the new objective is
$$\min_U \|U (M_i(U^1_* \circ \cdots \circ U^{i - 1}_*)G \times M_1(U^i \circ \cdots \circ U^q)) - M_i(A) (G \times I)\|_F$$
The sketching matrix $G$ can be selected so that this does not change the objective by more than $(1 \pm \eps/q)$ --- this can be achieved by selecting $G$ to be an affine embedding for $M_i(U^1_* \circ \cdots \circ U^{i - 1}_*)$ and $M_{\{1, \ldots, i - 1\}}(A)^T$, since for all $U \in \R^{k \times n \times k}$,
\begin{equation}
\begin{split}
\|U (M_i(U^1_* \circ \cdots \circ U^{i - 1}_*) & \times M_1(U^i \circ \cdots \circ U^q)) - M_i(A) \|_F  \\
& = \|M_1(U \circ U^i \circ \cdots \circ U^q)^T M_i(U^1_* \circ \cdots \circ U^{i - 1}_*) - M_{\{1, \ldots, i - 1\}}(A)^T \|_F
\end{split}
\end{equation}
and moreover,
\begin{equation}
\begin{split}
\|U (M_i(U^1_* \circ \cdots \circ U^{i - 1}_*)G & \times M_1(U^i \circ \cdots \circ U^q)) - M_i(A) (G \times I)\|_F \\
& = \|M_1(U \circ U^i \circ \cdots \circ U^q)^T M_i(U^1_* \circ \cdots \circ U^{i - 1}_*)G - M_{\{1, \ldots, i - 1\}}(A)^T G \|_F
\end{split}
\end{equation}
Thus, the approximation factor is only increased by a factor of $(1 + O(\eps/q))$ if $U^i$ is chosen to be
$$U^i := M_i(A) (G (M_i(U^1_* \circ \cdots \circ U^{i - 1}_*)G)^\dagger \times M_1(U^i \circ \cdots \circ U^q)^\dagger)$$
We can see that this definition of $U^i$ is equivalent to the definition of $U^i$ in \cite{huber2017randomized}, due to the following:
\begin{itemize}
    \item Note that multiplying $M_i(A)$ by $G \times I$ is equivalent to the usage of $G$ in \cite{huber2017randomized}.
    \item In addition, the key point is that the rows of $M_i(U^i \circ \cdots \circ U^q)$ are orthonormal, due to $U^i, \ldots, U^q$ being orthonormal bases. Thus, the pseudoinverse of $M_1(U^i \circ \cdots \circ U^q)$ is $M_1(U^i \circ \cdots \circ U^q)^T$, and therefore, multiplying $M_i(A) (G \times I)$ by $(I \times M_1(U^i \circ \cdots \circ U^q)^T)$ is equivalent to contracting $A$ with $(U^q)^T$, then with $(U^{q - 1})^T$, up to $(U^i)^T$. This is exactly as is done in \cite{huber2017randomized}.
    \item Finally, note that the change of basis matrix $(M_i(U^1_* \circ \cdots \circ U^{i - 1}_*)G)^\dagger$ can simply be folded into $U^{i - 1}$, and this is equivalent to letting $U^i$ be an orthonormal basis as is done in \cite{huber2017randomized}.
\end{itemize}
Thus, at each iteration of \cite{huber2017randomized}, if the sketching matrix $G$ is chosen properly, then $U^i$ as defined in \cite{huber2017randomized} is in fact a $(1 + \eps/q)$-approximate solution to the $\ell_2$ regression problem defined by $U^q, \ldots, U^{i + 1}$ and $U^1_*, \ldots, U^{i - 1}_*$.

\paragraph{Our $(1 + \eps)$-Approximation Algorithm for Tensor Train Decomposition with Improved Running Time}

The running time of the algorithm of \cite{huber2017randomized} (with the sketching matrix $G$ modified appropriately to obtain $(1 + \eps)$-approximation) is $O(q s^2 \nnz(A)) + n \cdot \poly(qk/\eps)$, where $s = \poly(qk/\eps)$ is the number of rows in $G$, and is also the bicriteria rank of the algorithm. We now show that this running time can be further reduced to $O(q \cdot \nnz(A)) + n \cdot \poly(qk/\eps)$ using our subspace embedding for matrices of the form $M_{\{1, \ldots, i - 1\}}(U^1 \circ \cdots \circ U^{i - 1})$ described above.

In the algorithm of \cite{huber2017randomized}, the main reason for the $O(q s^2 \cdot \nnz(A))$ running time is that in each iteration, $A$ is contracted with $(U^i)^T \in \R^{ns_{i + 1} \times s_i}$ --- thus, as stated in \cite{huber2017randomized}, before this $\circ$j is performed, $A$ has at most $s_{i + 1} \cdot \nnz(A)$ nonzero entries, meaning that the time needed to perform this $\circ$ is at most $O(s_i s_{i + 1} \cdot \nnz(A)) = O(s^2 \cdot \nnz(A))$. Furthermore, after this $\circ$ is performed, $A$ has at most $s_i \cdot \nnz(A)$ nonzero entries, as mentioned in \cite{huber2017randomized}. As we mentioned above, the contraction of $A$ with $(U^i)^T$ is needed as a step in computing the pseudo-inverse of $M_1(U^i \circ \cdots \circ U^q)$ and contracting it with $A$.

In order to improve the running time, we note that computing the pseudo-inverse of $M_1(U^i \circ \cdots \circ U^q)$ is in fact not required, and a weaker guarantee suffices --- that is, it suffices to obtain a subspace embedding for the row span of $M_1(U^i \circ \cdots \circ U^q)$. Note that we in fact need our sketch to be an affine embedding for $M_1(U^{i + 1} \circ \cdots \circ U^q)$ and $(G \times I)M_i(A)$, and we can inductively prove that this holds for our map $\calL$ that we defined earlier in this subsection. We want to show that for all $X$ of the appropriate dimensions
$$\|X M_1(U^{i + 1} \circ \cdots \circ U^q) - (G \times I)M_i(A)\|_F = \|X M_1(U^{i + 1} \circ \cdots \circ U^q) \calL - (G \times I)M_i(A) \calL\|_F$$
Suppose $S_q, \ldots, S_{i + 1}$ are the sketching matrices used when applying $\calL$ to $M_1(U^{i + 1} \circ \cdots \circ U^q)$ --- then, $\calL$ can be written as $(S_q \times I)(S_{q - 1} \times I) \ldots S_{i + 1}$ (where the identity matrices $I$ are of the appropriate dimensions). It suffices to show that the $(S_j \times I)$ are each affine embeddings with error $O(\eps/q^2)$ and failure probability $O(\delta/q^2)$ --- however, by reshaping appropriately, this just holds as long as $S_j$ itself is an affine embedding with error $O(\eps/q^2)$ and failure probability $O(\delta/q^2)$.

The key point is that this removes the $s^2$ factors in front of the running time. Note that we can reuse the $S_j$, i.e., when trying to find $U^i$, we can reuse the $S_q, \ldots, S_{i + 2}$ which were used when trying to find $U^{i + 2}$. Thus, in the iteration where $U^i$ is computed, it suffices to simply contract $A$ with $S_{i + 1}$, which takes $\nnz(A)$ time if $S_{i + 1}$ is a Countsketch matrix (see \cite{woodruff14_sketching_as_a_tool} for more information on Countsketch matrices). Over the course of $q$ iterations, this amounts to a running time of $O(q \cdot \nnz(A))$ instead of $O(q s^2 \cdot \nnz(A))$. Once we have applied $\calL$, we will compute the pseudo-inverse of $M_1(U^{i + 1} \circ \cdots \circ U^q) \calL$ --- however, this is simply a $k \times \poly(qk/\eps)$ matrix, meaning that we only incur an $n \cdot \poly(qk/\eps)$ term in the overall running time as a result of this computation. Thus, the overall running time for Tensor Train decomposition when using our new subspace embedding is $O(q \cdot \nnz(A)) + n \cdot \poly(qk/\eps)$.

\subsection{Bicriteria $(1+\eps)$ Approximation Algorithm for Tree Networks}

We obtain a $(1 + \eps)$-approximation algorithm with bicriteria rank $\poly(qk/\eps)$. The key idea of our algorithm is as follows. We root the tree arbitrarily. Then, we move from the leaves to the root of the tree, computing factors $U_v$ for nodes $v$ in that order. We compute the factor $U_v$ for a node $v$ once all of the factors for its subtree, excluding $v$ itself, have been computed --- aside from this, the order in which we process vertices can be arbitrarily chosen. To allow for efficiently computing $U_v$ given the factors in its subtree, we carefully apply Countsketch matrices in succession. We analyze Algorithm \ref{algorithm:tree_network_bicriteria} below.

\begin{algorithm}
\caption{$(1 + \eps)$-approximation algorithm for tree network decomposition with output rank $O(\frac{qk}{\eps} \log(\frac{q}{\delta}))$. Note that we can assume each of the leaves in $\calT$ has one mode of dimension $n$ --- if a leaf $v$ has no such mode, then $U_v$ is simply a $k$-dimensional vector, and can be contracted with its parent and ignored for the purposes of approximating $A$.}
\label{algorithm:tree_network_bicriteria}
\begin{algorithmic}
\Require $A \in \R^{n \times \ldots \times n}$ with $q$ modes, and a tree $\calT = (V, E)$ with $O(q)$ vertices, such that for any $u, v \in V$ with $(u, v) \in E$, the edge between $u$ and $v$ has rank $k$, and each vertex in $V$ has at most one mode of dimension $n$. Without loss of generality, each leaf in $V$ has one mode of dimension $n$. All vertices of $\calT$ have degree at most $d \in \N$.
\Ensure $\{U_v \mid v \in V\}$ such that for each $v \in V$, the dimension of $U_v$ corresponding to each of its outgoing edges in $\calT$ is at most $t = O(\frac{qk}{\eps} \log(q/\delta))$

\item []
\State // Process the leaves first
\State $L \gets $ A list of all the leaves of $\calT$
\For{$v \in L$}
    \State $T_v \gets $ An $O(\frac{q^3 k^2}{\eps^2 \delta}) \times n^{q - 1}$ Countsketch matrix
    \State $R_v \gets $ A $t \times O(\frac{q^3 k^2}{\eps^2 \delta})$ matrix whose entries are drawn i.i.d. from $\{-\frac{1}{\sqrt{t}}, \frac{1}{\sqrt{t}}\}$
    \State $i \gets $ the mode of $U_v$ of dimension $n$
    \State $\widetilde{U_v} \gets M_i(A) T_v^T R_v^T \in \R^{n \times t}$
    
    \item []
    \State // Store a sketch of $\widetilde{U_v}$, which we denote $M_v$, for future use.
    \State $S_v \gets $ An $s \times n$ Countsketch matrix, where $s = O(\frac{q^4 t^2 d^3}{\eps^2 \delta}) = O(\frac{q^6 k^2 d^3}{\eps^4 \delta} \log^2(\frac{q}{\delta}))$.
    \State $M_v \gets S_v \widetilde{U_v} \in \R^{s \times t}$
    \State $A \gets S_v \circ_i A$
\EndFor

\item []
\State // Now process each vertex $v \in V$ after processing its entire subtree.
\State $\calI \gets $ A list of the vertices in $V \setminus L$ where each vertex appears after the rest of the vertices in its subtree.
\For{$v \in \calI$}
    \State // First compute $\widetilde{U_v}$ using the sketches of the subtrees of the children of $v$.
    \State $M_{v-subtree} \gets \bigotimes_u M_u \in \R^{s^{\deg(u) - 1} \times t^{\deg(u) - 1} }$, where $u$ ranges over the children of $v$
    \State $\calG_1 \gets $ The modes of $A$ corresponding to row dimensions of $M_u$ for the children $u$ of $v$
    \State $\calG_2 \gets $ The remaining modes of $A$
    \State $\widetilde{U_v} \gets M_{v-subtree}^\dagger \circ_{\calG_1} A$
    
    \If{$v$ is not the root of $\calT$}
        \State $T_v \gets $ An $O(\frac{q^3 k^2}{\eps^2 \delta}) \times n^{|\calG_2|}$ Countsketch matrix
        \State $R_v \gets $ A $t \times O(\frac{q^3 k^2}{\eps^2 \delta})$ matrix whose entries are drawn i.i.d. from $\{-\frac{1}{\sqrt{t}}, \frac{1}{\sqrt{t}}\}$
        \State $\widetilde{U_v} \gets R_vT_v \circ_{\calG_2} \widetilde{U_v} \in \R^{t \times \ldots \times t \times t}$
    \EndIf

    \item []
    \State // Next, compute $M_v$.
    \State $M_v \gets M_{v-subtree} \circ \widetilde{U_v} \in \R^{s^{\deg(u) - 1} \times t}$ where the contraction is along the modes of $\widetilde{U_v}$ 
    \State \,\,\,\,\,\,\,\,\,\,\,\,\,\,\,\, corresponding to the edges between $u$ and its children
    \State $S_v \gets $ An $s \times s^{\deg(u) - 1}$ Countsketch matrix.
    \State $M_v \gets S_v M_v$
    \State $A \gets S_v \circ_{\calG_1} A$
\EndFor \\
\\

\Return $\{\widetilde{U_v} \mid v \in V\}$

\end{algorithmic}
\end{algorithm}

\begin{theorem} \label{thm:analysis_of_tree_network_bicriteria}
Let $k, q, n \in \N$, and let $A$ and $\calT$ be as specified in Algorithm \ref{algorithm:tree_network_bicriteria}. Then, Algorithm \ref{algorithm:tree_network_bicriteria} finds $\{U_v \mid v \in V\}$ with bicriteria rank $t = O(\frac{qk}{\eps} \log(\frac{q}{\delta}))$ such that 
$$\|\calT(\{U_v \mid v \in V\}) - A\|_F \leq (1 + \eps) \min_{U_v} \|\calT(\{U_v \mid v \in V\}) - A\|_F$$
with probability at least $1 - \delta$. The running time of Algorithm \ref{algorithm:tree_network_bicriteria} is $O(q \cdot \nnz(A)) + n \cdot (\frac{qk}{\eps\delta})^{O(d)}$.
\end{theorem}
\begin{proof}
For convenience, let $U_v^*$ for $v \in V$ be such that
$$\|\calT(\{U_v^* \mid v \in V\}) - A\|_F = \min_{\{U_v \mid v \in V\}} \|\calT(\{U_v \mid v \in V\}) - A\|_F$$
and the dimension of $U_v^*$ is $k$ along the modes corresponding to the edges incident to $v$ (along with potentially one mode of dimension $n$). Throughout the proof, we will let $\calS$ be the set of vertices $v \in V$ for which $\widetilde{U_v}$ has already been computed so far --- in our proof, we show by induction on $|\calS|$ that with probability $1 - O(\frac{|\calS| \delta}{|V|})$,
$$\|\calT(\{\widetilde{U_v} \mid v \in S\} \cup \{U_v^* \mid v \in V \setminus \calS\}) - A\|_F \leq \Big(1 + O\Big(\frac{\eps}{|V|}\Big)\Big)^{|\calS|} \|\calT(\{U^*_v \mid v \in V\}) - A\|_F$$
Note that $A$ changes throughout the course of Algorithm \ref{algorithm:tree_network_bicriteria} since we contract it repeatedly with Countsketch matrices --- for clarity, in this proof, we let $A$ denote the original value of $A$ provided at the time of input, and we let $A_v$ denote the value of $A$ when we are attempting to compute $\widetilde{U_v}$. 

Let $w \in V$ be the vertex that is currently being processed (i.e. if we are in the step where we process the leaves in $L$, then all the leaves prior to $w$ in $L$ have been processed, and otherwise if $w \in \calI$, then all the vertices in $\calI$ before $w$ have been processed). Assume that
\begin{equation}
\begin{split}
\|\calT(\{\widetilde{U_v} \mid v \in S\} \cup \{U_v^* \mid v \in V \setminus \calS\}) - A\|_F
& \leq \Big(1 + O\Big(\frac{\eps}{|V|}\Big)\Big)^{|\calS|} \|\calT(\{U_v^* \mid v \in V\}) - A\|_F
\end{split}
\end{equation}
Without loss of generality, assume that $w$ is not a leaf of $\calT$ --- the case where $w$ is a leaf is similar but simpler. Then, the optimal $U_w$ (with $\widetilde{U_v}$ for $v \in S$ and $U_v^*$ for $v \in V \setminus (S \cup \{w\})$ held fixed) is the minimizer for the multiple-response $\ell_2$ regression problem, defined as follows.

First, let $\calT_{above}$ denote the tree which is obtained by removing the subtree rooted at $w$ from $\calT$, and let $V_{above}$ denote the vertices of $\calT_{above}$. In addition, let $\calT_1, \ldots, \calT_{\deg(w) - 1}$ denote the subtrees rooted at each of the children of $w$, and let $V_i$ denote the vertices of $\calT_i$. Then,
\begin{equation}
\begin{split}
\calT(\{\widetilde{U_v} & \mid v \in \calS\} \cup \{U_v^* \mid v \in V \setminus \calS\}) \\
& = \calT_{above}(\{U_v \mid v \in V_{above}\}) \circ U_w \circ \Big(\bigotimes_{i = 1}^{\deg(w) - 1} \calT_i(\{U_v \mid v \in V_i\}) \Big)
\end{split}
\end{equation}
On the right hand side, for the sake of simplicity, we use $U_v$ to denote $\widetilde{U_v}$ for $v \in \calS$ and $U_v^*$ for $v \in V \setminus \calS$. Here, the contraction between $\calT_{above}(\{U_v \mid v \in V_{above}\})$ and $U_w$ happens along the mode corresponding to the edge between $w$ and its parent, and the contraction between $U_w$ and $\calT_i(\{U_v \mid v \in V_i\})$ happens along the mode corresponding to the edge between $w$ and its $i^{th}$ child (where the children can be ordered arbitrarily). Then, redefining $U_w^*$ as
\begin{equation}
\begin{split}
U_w^* := \argmin_{U_w} \Big\|\calT_{above}(\{U_v \mid v \in V_{above}\}) \circ U_w \circ \Big(\bigotimes_{i = 1}^{\deg(w) - 1} \calT_i(\{U_v \mid v \in V_i\}) \Big) - A\Big\|_F
\end{split}
\end{equation}
clearly can only decrease the error achieved by $\{\widetilde{U_v} \mid v \in \calS\} \cup \{U_v^* \mid v \in V \setminus \calS\}$. Thus, it suffices to show that for fixed values of $\widetilde{U_v}$ (for $v \neq w$) and $U_v^*$, $\widetilde{U_w}$ is a $(1 + O(\frac{\eps}{|V|}))$-approximate solution to this multiple-response $\ell_2$ regression problem with probability $1 - O(\frac{\delta}{|V|})$.

First, we analyze the effect of applying the Countsketch matrices $S_v$ for $v \in \calS$:

\begin{lemma}
With probability $1 - O(\frac{\delta}{|V|})$, for all tensors $U_w$ of the appropriate dimensions,
\begin{equation} \label{eq:affine_embedding_dp_countsketch}
\begin{split}
&\Big\|\calT_{above}(\{U_v \mid v \in V_{above}\})  \circ U_w \circ \Big(\bigotimes_{i = 1}^{\deg(w) - 1} M_i\Big) - A_v \Big\|_F \\
& = \Big(1 + O\Big(\frac{\eps}{|V|}\Big)\Big)
\Big\|\calT_{above}(\{U_v \mid v \in V_{above}\}) \circ U_w \circ \Big(\bigotimes_{i = 1}^{\deg(w) - 1} \calT_i(\{U_v \mid v \in V_i\}) \Big) - A \Big\|_F
\end{split}
\end{equation}
where if $v_i$ is the $i^{th}$ child of $w$, then we write $M_{v_i}$ as $M_i$ for simplicity. Here, for simplicity, we also write $\widetilde{U_v}$ as $U_v$ for $v \in \calS$, and we write $U_v^*$ as $U_v$ for $v \in V \setminus \calS$.
\end{lemma}
\begin{proof}
Note that when we contract a Countsketch matrix $S_v$ with $A$ and with $\calT(\{U_v \mid v \in V\})$, this has the same effect as matricizing $A$ and $\calT(\{U_v \mid v \in V\})$ along the appropriate modes and multiplying both by a matrix of the form $S_v \otimes I$ (where $I$ is an identity matrix of the appropriate dimensions). Thus, replacing $A$ with $A_v$, and replacing $\calT_i(\{U_v \mid v \in V_i\})$ with $M_i$ for $1 \leq i \leq \deg(w) - 1$, is equivalent to matricizing and then multiplying by a sequence of matrices of the form $S_v \otimes I$ where $v$ ranges over $\calS$.

Now, observe that for each $v \in \calS$, $S_v$ has the $(O(\frac{\eps}{q^{3/2}t \deg(w)}), O(\frac{\delta}{q \deg(w)}), 2)$ JL moment property. This is because by Lemma \ref{lemma:countsketch_jl_moment_property}, in order for a Countsketch matrix to have the \\ $(O(\frac{\eps}{q^{3/2}t \deg(w)}), O(\frac{\delta}{q \deg(w)}), 2)$ JL moment property, it must have  $\Omega(\frac{q^3 t^2 \deg(w)^2}{\eps^2} \cdot \frac{q \deg(w)}{\delta}) $ = \\ $\Omega(\frac{q^6 k^2 \deg(w)^3}{\eps^4 \delta} \log^2(\frac{q}{\delta}))$ rows --- in particular, $S_v$ has $O(\frac{q^6 k^2 d^3}{\eps^4 \delta} \log^2(\frac{q}{\delta}))$ rows where $d$ is the maximum degree of any vertex in $\calT$, which suffices. In addition, by Lemma \ref{lemma:jl_moment_with_identity_matrix}, it follows that $S_v \otimes I$ also has the $(O(\frac{\eps}{q^{3/2}t \deg(w)}), O(\frac{\delta}{q \deg(w)}), 2)$ JL moment property. \footnotemark \footnotetext{For convenience, we write this map as $S_v \otimes I$. However, this does not imply that $S_v$ is only applied to the first few modes of $A$ --- the subset of modes to which $S_v$ is being applied can depend on $v$.}

Now, observe that if $v_1 \in V_i$ and $v_2 \in V_j$ for $i \neq j$, then $S_{v_1} \otimes I$ and $S_{v_2} \otimes I$ commute, since the subset of modes of $A$ to which $S_{v_1}$ is applied is disjoint from the subset of modes of $A$ to which $S_{v_2}$ is applied. The linear map used to obtain $A_v$ from $A$ and $\bigotimes_{i = 1}^{\deg(w) - 1} M_i$ from $\bigotimes_{i = 1}^{\deg(w) - 1} \calT_i(\{U_v \mid v \in V_i\})$ can be written as $\prod_{v \in \calS} (S_v \otimes I)$. By the prior discussion, this map can also be written as
$$\prod_{i = 1}^{\deg(w) - 1} \Big(\prod_{v \in V_i} (S_v \otimes I) \Big) \prod_{v \in V_{above} \cap \calS} (S_v \otimes I)$$
Now, by Lemma \ref{lemma:jl_moment_composition_of_jl_maps}, for $1 \leq i \leq \deg(w) - 1$, $\prod_{v \in V_i} (S_v \otimes I)$ has the $(O(\frac{\eps}{q t \deg(w)}), O(\frac{\delta}{q \deg(w)}), 2)$ JL moment property, since $S_v \otimes I$ has the $(O(\frac{\eps}{q^{3/t} t \deg(w)}), O(\frac{\delta}{q \deg(w)}), 2)$ JL moment property for each $v \in V_i$, and moreover $|V_i| \leq |\calS| \leq |V| = O(q)$. Similarly, $\prod_{v \in V_{above} \cap \calS} (S_v \otimes I)$ has the $(O(\frac{\eps}{qt \deg(w)}), O(\frac{\delta}{q \deg(w)}), 2)$ JL moment property. Thus, each of $\prod_{v \in V_1} (S_v \otimes I), \ldots, \prod_{v \in V_{\deg(w) - 1}} (S_v \otimes I), \prod_{v \in V_{above} \cap \calS} (S_v \otimes I)$ has the $(O(\frac{\eps}{qt\deg(w)}), O(\frac{\delta}{q \deg(w)}), 2)$ JL moment property. 

Finally, by Lemma \ref{lemma:jl_moment_implies_affine_embedding}, with probability at least $1 - O(\frac{\delta}{q \deg(w)})$, $\prod_{v \in V_i} (S_v \otimes I)$ is a $(1 \pm O(\frac{\eps}{q \deg(w)}))$ affine embedding for $\calT_i(\{U_v \mid v \in V_i\})$ (matricized appropriately). Note here that we used the fact that the edge between $w$ and its $i^{th}$ child $v_i$ has rank $t$, meaning that $\prod_{v \in V_i} (S_v \otimes I)$ simply has to be an affine embedding for a subspace of dimension $t$. Similarly, $\prod_{v \in V_{above} \cap \calS} (S_v \otimes I)$ is a $(1 \pm O(\frac{\eps}{q \deg(w)}))$ affine embedding for $\calT_{above}(\{U_v \mid v \in V_{above}\})$ with probability at least $1 - O(\frac{\delta}{q \deg(w)})$. Thus, by a union bound, with probability at least $1 - O(\frac{\delta}{q})$, for all $U_w$ of appropriate dimensions,
\begin{equation}
\begin{split}
& \Big\|\calT_{above}(\{U_v  \mid v \in V_{above}\}) \circ U_w \circ \Big(\bigotimes_{i = 1}^{\deg(w) - 1} M_i\Big) - A_v \Big\|_F \\
& = \Big(1 + O\Big(\frac{\eps}{q\deg(w)}\Big)\Big)^{\deg(w)}
\Big\|\calT_{above}(\{U_v \mid v \in V_{above}\}) \circ U_w \circ \Big(\bigotimes_{i = 1}^{\deg(w) - 1} \calT_i(\{U_v \mid v \in V_i\}) \Big) - A \Big\|_F \\
& = \Big(1 + O\Big(\frac{\eps}{q}\Big)\Big)
\Big\|\calT_{above}(\{U_v \mid v \in V_{above}\}) \circ U_w \circ \Big(\bigotimes_{i = 1}^{\deg(w) - 1} \calT_i(\{U_v \mid v \in V_i\}) \Big) - A \Big\|_F
\end{split}
\end{equation}
as desired (note that this implies the lemma since $|V| = \Theta(q)$). 
\end{proof}

Next, we analyze the effect of applying the Countsketch matrix $T_v$ to the modes of $A$ corresponding to $\calT_{above}$:

\begin{lemma}
With probability $1 - O(\frac{\delta}{|V|})$, for all tensors $U_w$ of the appropriate dimensions,
\begin{equation}
\begin{split}
\Big\|T_v \circ \calT_{above}(\{U_V & \mid v \in V_{above}\})  \circ U_w \circ \Big(\bigotimes_{i = 1}^{\deg(w) - 1} M_i \Big) - T_v \circ A_v \Big\|_F \\
& = \Big(1 + O\Big(\frac{\eps}{|V|}\Big)\Big) \Big\|\calT_{above}(\{U_v \mid v \in V_{above}) \circ U_w \circ \Big(\bigotimes_{i = 1}^{\deg(w) - 1} M_i\Big) - A_v \Big\|_F
\end{split}
\end{equation}
\end{lemma}
\begin{proof}
After matricizing $\calT_{above}(\{U_v \mid v \in V_{above}\})$ and $A_v$ appropriately, it is sufficient to show that $T_v$ is a $(1 \pm O(\frac{\eps}{q}))$ affine embedding for a $k$-dimensional subspace (and $A_v$, matricized appropriately) with probability at least $1 - O(\frac{\delta}{q})$. However, since $T_v$ has $O(\frac{q^3 k^2}{\eps^2 \delta})$ rows, it has the $(O(\frac{\eps}{qk}), O(\frac{\delta}{q}), 2)$ JL moment property by Lemma \ref{lemma:countsketch_jl_moment_property}. It thus follows from Lemma \ref{lemma:jl_moment_implies_affine_embedding} that $T_v$ is a $(1 \pm O(\frac{\eps}{q}))$-approximation for $\calT_{above}(\{U_v \mid v \in V_{above}\})$ and $A_v$, matricized appropriately.
\end{proof}

Finally, applying the dense sign matrix $R_v$ leads to a $(1 + O(\frac{\eps}{|V|}))$-approximate solution to multiple-response $\ell_2$ regression:

\begin{lemma}
Define
\begin{equation}
\begin{split}
\widehat{U_w}
& := \argmin_{U_w} \Big\|R_v T_v \circ \calT_{above}(\{U_v \mid v \in V_{above}\}) \circ U_w \circ \Big( \bigotimes_{i = 1}^{\deg(w) - 1} M_i \Big) - R_v T_v \circ A_v \Big\|_F
\end{split}
\end{equation}
Then, with probability $1 - O(\frac{\delta}{|V|})$,
\begin{equation}
\begin{split}
\Big\|T_v \circ & \calT_{above}(\{U_v \mid v \in V_{above}\}) \circ \widehat{U_w} \circ \Big(\bigotimes_{i = 1}^{\deg(w) - 1} M_i \Big) - T_v \circ A_v \Big\|_F \\
& \leq \Big(1 + O\Big(\frac{\eps}{|V|}\Big)\Big) \min_{U_w} \Big\|T_v \circ \calT_{above}(\{U_v \mid v \in V_{above}\}) \circ U_w \circ \Big(\bigotimes_{i = 1}^{\deg(w) - 1} M_i \Big) - T_v \circ A_v \Big\|_F
\end{split}
\end{equation}
\end{lemma}
\begin{proof}
We wish to show that after contracting $R_v$ with $T_v \circ \calT_{above}(\{U_v \mid v \in V_{above}\})$ and $T_v \circ A_v$, along the mode corresponding to the row dimension of $T_v$, that the minimizer of the regression problem with $R_v$ is a $(1 + O(\frac{\eps}{q}))$-approximate solution to the regression problem without $R_v$. However, the regression problem
\begin{equation}
\begin{split}
\min_{U_w} \Big\|T_v \circ \calT_{above}(\{U_v \mid v \in V_{above}\}) \circ U_w \circ \Big(\bigotimes_{i = 1}^{\deg(w) - 1} M_i \Big) - T_v \circ A_v \Big\|_F
\end{split}
\end{equation}
can be rewritten in the following form:
\begin{equation} \label{eq:kronecker_product_regression_tree_networks}
\begin{split}
\min_{U_w} \Big\|U_w \circ \Big((T_v \circ \calT_{above}(\{U_v \mid v \in V_{above}\})) \otimes \Big(\bigotimes_{i = 1}^{\deg(w) - 1} M_i \Big)\Big) - T_v \circ A_v \Big\|_F
\end{split}
\end{equation}
Without loss of generality, assume $U_w$ has one mode of dimension $n$, in which case \ref{eq:kronecker_product_regression_tree_networks} is a multiple-response $\ell_2$ regression problem. Let $i_w$ be the mode of $U_w$ of dimension $n$ --- then, matricizing the above problem along mode $i_w$ gives
\begin{equation} \label{eq:kronecker_product_regression_tree_networks_matricized}
\begin{split}
\min_{U_w} \Big\|M_{i_w}(U_w)  \Big(M(T_v \circ (\calT_{above}(\{U_v \mid v \in V_{above}\}))) \otimes \Big(\bigotimes_{i = 1}^{\deg(w) - 1} M(M_i) \Big)\Big) - M_{i_w}(T_v \circ A_v) \Big\|_F
\end{split}
\end{equation}
Here, $M(T_v \circ \calT_{above}(\{U_v \mid v \in V_{above}\}))$ and $M(M_i)$ denote the matricizations of $\calT_{above}(\{U_v \mid v \in V_{above}\})$ and $M_i$ along the modes corresponding to their edges connecting them to $U_w$. By the definition of $T_v$, $M(T_v \circ \calT_{above}(\{U_v \mid v \in V_{above}\}))$ is a $k \times O(\frac{q^3 k^2}{\eps^2 \delta})$ matrix.

Now, it suffices to show that the conclusion of Theorem \ref{thm:general_l2_regression_sketching} holds for the problem
\begin{equation} \label{eq:kronecker_product_regression_tree_networks_dense_sign}
\begin{split}
\min_{U_w} \Big\|M_{i_w}(U_w) \Big(M(T_w \circ (\calT_{above}(\{U_v &\mid v \in V_{above}\})))  \\ &\otimes \Big(\bigotimes_{i = 1}^{\deg(w) - 1} M(M_i) \Big)\Big)(R_w \otimes I) - M_{i_w}(T_w \circ A_w)(R_w \otimes I) \Big\|_F
\end{split}
\end{equation}
i.e. that the minimizer of \ref{eq:kronecker_product_regression_tree_networks_dense_sign} is a $(1 + O(\frac{\eps}{q}))$-approximate solution to \ref{eq:kronecker_product_regression_tree_networks_matricized}. By Theorem \ref{thm:general_l2_regression_sketching}, it suffices to show that $R_w \otimes I$ has the $(O(\sqrt{\frac{\eps}{qk}}), O(\frac{\delta}{q}))$ approximate matrix product property, and that it is a $(1 \pm \frac{1}{3})$ subspace embedding for $\Big(M(T_w \circ (\calT_{above}(\{U_v \mid v \in V_{above}\}))) \otimes \Big(\bigotimes_{i = 1}^{\deg(w) - 1} M(M_i) \Big)\Big)$ with probability at least $1 - O(\frac{\delta}{q})$.

By Lemma \ref{lemma:sign_matrix_jl_moment_property}, $R_w$ has the $(O(\sqrt{\frac{\eps}{qk}}), O(\frac{\delta}{q}), O(\log(\frac{q}{\delta})))$ JL moment property as long as it has $O(\frac{qk}{\eps} \log(\frac{q}{\delta}))$ rows, which it does --- by Lemma \ref{lemma:jl_moment_with_identity_matrix}, this implies that $R_w \otimes I$ also has the $(O(\sqrt{\frac{\eps}{qk}}), O(\frac{\delta}{q}), O(\log(\frac{q}{\delta})))$ JL moment property. Thus by Theorem \ref{thm:jl_property_implies_approx_matrix_product}, $R_w \otimes I$ has the $(O(\sqrt{\frac{\eps}{qk}}), O(\frac{\delta}{q}))$ approximate matrix product property. In addition, by Lemma \ref{lemma:kronecker_product_subspace_embedding}, it suffices for $R_w$ to be a $(1 \pm \frac{1}{3})$ subspace embedding for $M(T_w \circ \calT_{above}(\{U_v \mid v \in V_{above}\}))$ which is a rank $k$ matrix --- by Lemma \ref{lemma:sign_matrix_subspace_embedding}, this holds as long as $R_w$ has $O(k \log(q/\delta))$ rows. This completes the proof of the lemma.
\end{proof}

Combining the previous three lemmas, we find that if
\begin{equation} \label{eq:Uw_regression_problem_tree_networks}
\begin{split}
\widehat{U_w}
& := \argmin_{U_w} \Big\|R_w T_w \circ \calT_{above}(\{U_v \mid v \in V_{above}\}) \circ U_w \circ \Big( \bigotimes_{i = 1}^{\deg(w) - 1} M_i \Big) - R_w T_w \circ A_v \Big\|_F
\end{split}
\end{equation}
then with probability $1 - O(\frac{\delta}{|V|})$,
\begin{equation}
\begin{split}
& \Big\|\calT_{above}(\{U_v  \mid v \in V_{above}\}) \circ \widehat{U_w} \circ \Big(\bigotimes_{i = 1}^{\deg(w) - 1} \calT_i(\{U_v \mid v \in V_i\}) \Big) - A \Big\|_F \\
& = \Big(1 + O\Big(\frac{\eps}{|V|}\Big)\Big)
\min_{U_w} \Big\|\calT_{above}(\{U_v \mid v \in V_{above}\}) \circ U_w \circ \Big(\bigotimes_{i = 1}^{\deg(w) - 1} \calT_i(\{U_v \mid v \in V_i\}) \Big) - A \Big\|_F
\end{split}
\end{equation}
However, the regression problem \ref{eq:Uw_regression_problem_tree_networks} can be rewritten as the Kronecker product regression problem
\begin{equation}
\begin{split}
\widehat{U_w}
 := \argmin_{U_w} \Big\| U_w \circ \Big(M(R_w T_w \circ \calT_{above}(\{U_v &\mid v \in V_{above}\})) \\ &\otimes \Big(\bigotimes_{i = 1}^{\deg(w) - 1} M(M_i) \Big) \Big) - M(R_w T_w \circ A_w) \Big\|_F
\end{split}
\end{equation}
Thus, the minimizer $\widehat{U_w}$ can be obtained by multiplying $M(R_w T_w \circ A_w)$ on the right by the pseudo-inverse of the Kronecker product $\Big(M(R_w T_w \circ \calT_{above}(\{U_v \mid v \in V_{above}\})) \otimes \Big(\bigotimes_{i = 1}^{\deg(w) - 1} M(M_i) \Big) \Big)$, which in turn is the Kronecker product of the pseudo-inverse of the individual components. Thus,
\begin{equation}
\begin{split} 
\widehat{U_w} 
& = (R_w T_w \circ \calT_{above}(\{U_v \mid v \in V_{above}\})^\dagger \circ (R_w T_w \circ A_w) \circ \Big(\bigotimes_{i = 1}^{\deg(w) - 1} M_i^\dagger \Big)
\end{split}
\end{equation}
Observe that $M_i$ and therefore $M_i^\dagger$ is known. $U_v$ is not known for all $v \in V_{above}$, meaning that $(R_w T_w \circ \calT_{above}(\{U_v \mid v \in V_{above}\}))^\dagger$ cannot be computed efficiently. However, $(R_w T_w \circ \calT_{above}(\{U_v \mid v \in V_{above}\}))^\dagger$ is a $t \times t$ matrix where $t = O(\frac{qk}{\eps} \log(\frac{q}{\delta}))$. Observe that, since we define
$$\widetilde{U_w} = (R_w T_w \circ A_w) \circ \Big(\bigotimes_{i = 1}^{\deg(w) - 1} M_i^\dagger\Big)$$
we have $\widehat{U_w} = (R_w T_w \circ \calT_{above}(\{U_v \mid v \in V_{above}\}))^\dagger \circ \widetilde{U_w}$. Thus, instead of contracting $(R_w T_w \circ \calT_{above}(\{U_v \mid v \in V_{above}\}))^\dagger$ with $\widetilde{U_w}$, we can contract it with $U_{w'}^*$, where $w'$ is the parent of $w$. In other words, we redefine $U_{w'}^*$ as $U_{w'}^* \circ (R_w T_w \circ \calT_{above}(\{U_v \mid v \in V_{above}\}))^\dagger$, where the contraction occurs along the mode of $U_{w'}^*$ corresponding to the edge between $w'$ and $w$. This has the effect that the edge between $w$ and $w'$ now has rank $t = O(\frac{qk}{\eps} \log(\frac{q}{\delta}))$ instead of $k$. In summary, if we define $\widetilde{U_w}$ and redefine $U_{w'}^*$ as above, and append $w$ to $\calS$, then it still holds that
$$\|\calT(\{\widetilde{U_v} \mid v \in S\} \cup \{U_v^* \mid v \in V \setminus \calS\}) - A\|_F \leq \Big(1 + O\Big(\frac{\eps}{|V|}\Big)\Big)^{|\calS|} \|\calT(\{U^*_v \mid v \in V\}) - A\|_F$$
with probability $1 - O(\frac{\delta |\calS|}{|V|})$, as desired. This completes the induction.

Finally, we calculate the running time of Algorithm \ref{algorithm:tree_network_bicriteria}. First, we calculate the time required to process a leaf $v \in L$. The time needed to compute $\widetilde{U_v} = M_i(A) T_v^T R_v^T$ is $\nnz(A) + n \cdot O(\frac{q^3 k^2}{\eps^2 \delta}) \cdot t$, which is $\nnz(A) + n \cdot \poly(\frac{qk}{\eps})$. We then compute $M_v = S_v \widetilde{U_v}$, which takes $O(snt) = n \cdot \poly(\frac{qk}{\eps})$ time. We then perform the update $A \gets S_v \circ_i A$, which takes $\nnz(A)$ time --- this is because over the course of the algorithm, $\nnz(A_v)$ is at most $\nnz(A)$, since $\nnz(A_v)$ is obtained by contracting $A$ with a series of Countsketch matrices. Since there are $|L| = O(q)$ leaves, the total time required to process the leaves is $O(q \cdot \nnz(A)) + n \cdot \poly(\frac{qk}{\eps})$. 

Next, we calculate the time needed to process vertices $v \in V \setminus L$. If $v$ is not the root of $T$, then $\widetilde{U_v}$ is defined as $(M^\dagger_{v-subtree} \circ_{\calG_1} A \circ_{G_2} (R_vT_v)$. We can efficiently perform this computation in the following manner:
\begin{enumerate}
    \item Contract $A$ with $T_v$, which takes $\nnz(A)$ time. After contracting $A$ with $T_v$, we obtain a tensor with dimensions $s \times \ldots \times s \times O(\frac{q^3 k^2}{\eps^2 \delta}) \times n$, i.e. it has dimension $s$ on the $\deg(w) - 1$ modes corresponding to the children of $w$. (Without loss of generality, we are assuming that $w$ has a single mode of dimension $n$.)
    \item Contract $A \circ T_v$ with $R_v$, which takes time $s^{\deg(w) - 1} \cdot O(\frac{q^3k^2}{\eps^2 \delta}) \cdot n \cdot t = n \cdot \poly(\frac{qk}{\eps \delta})^{O(\deg(w))} = n \cdot (\frac{qk}{\eps\delta})^{O(d)}$.
    \item Contract $A \circ (R_vT_v)$ with $M_{v-subtree}^\dagger$. This can be done by separately contracting $A \circ (R_vT_v)$ with $M_u$ for each child $u$ of $v$. There are $O(d)$ such contractions, and since $M_u \in \R^{s \times t}$ and $A \circ (R_vT_v)$ has dimension $s \times \ldots \times s \times t \times n$, each contraction takes $O(t \cdot s^{\deg(v) - 1} \cdot t \cdot n) = n \cdot \poly(\frac{qk}{\eps \delta})^{\deg(v)} = n \cdot (\frac{qk}{\eps\delta})^{O(d)}$ time. Thus, this step overall takes $n \cdot (\frac{qk}{\eps \delta})^{O(d)}$ time.
\end{enumerate}
Thus, for a single vertex $v \in V \setminus L$, computing $\widetilde{U_v}$ takes $\nnz(A) + n \cdot (\frac{qk}{\eps\delta})^{O(d)}$ time. Next, we compute $M_v$. This is done by contracting $M_{v-subtree}$ with $\widetilde{U_v}$ --- this takes at most $n \cdot (\frac{qk}{\eps\delta})^{O(d)}$ time. Finally, contracting $A$ with $S_v$ takes $\nnz(A)$ time. Thus, processing a single vertex $v \in V \setminus L$ takes $\nnz(A) + n \cdot (\frac{qk}{\eps\delta})^{O(d)}$ time, and the overall running time of Algorithm \ref{algorithm:tree_network_bicriteria} is $O(q \cdot \nnz(A)) + n \cdot (\frac{qk}{\eps\delta})^{O(d)}$.
\end{proof}

\newpage
\section{Tensors with Low General Tensor Network Rank} \label{sec:general_networks}

In this section, we consider the case when the $q$-mode tensor $A \in \R^{n \times \ldots \times n}$ can be approximated well by a tensor network with rank $k$ according to a graph $\calG$. More specifically, suppose we are given a graph $\calG = (V, E)$ with $|V| = q$, and $k \in \N$. Then, for a collection of tensors $\{U_v \mid v \in V\}$, such that for each $v \in V$, $U_v$ has one mode of dimension $n$ and $\deg(v)$ modes of dimension $k$, we define $\calG(\{U_v \mid v \in V\})$ to be the tensor obtained by contracting the $U_v$ in the manner specified by the edges of $\calG$. In this section, we give algorithms for approximating $A$ by a tree network of small rank, which achieves error not much larger than $\min_{U_v} \|\calG(\{U_v \mid v \in V\}) - A\|_F$.

\subsection{Existence of a Binary Tree Network with Low Relative Error and Rank $k^{O(\deg(G) \tw(G))}$}
\label{subsec:binary_tree_existential_result}

First, we show the following existential result. Suppose $\calG = (V, E)$ is a graph, and $\{U_v \mid v \in V\}$ is a collection of tensors such that $U_v$ has one mode of dimension $n$ and $\deg(v)$ modes of dimension $k$. Then, there exists a binary tree $\calT = (V_\calT, E_\calT)$, and a collection of tensors $\{\widehat{U_v} \mid v \in V_\calT\}$ such that $\widehat{U_v}$ has one mode of dimension $n$, as well as $\deg_{\calT}(v)$ modes of dimension $k^t$ (where $t$ is small as long as the maximum degree and treewidth of $\calG$ are small), such that
$$\calG(\{U_v \mid v \in V\}) = \calT(\{\widehat{U_v} \mid v \in V_\calT\})$$
First, we formally define contraction of an edge in a graph:
\begin{definition}[Edge Contraction]
Let $G = (V, E)$ be a graph, and let $(u, v) \in E$. Then, contracting the edge $(u, v)$ results in a new graph $G' = (V', E')$, defined as follows. $V'$ is equal to $V$, but with $u$ and $v$ removed and replaced with a new vertex $w \not\in V$. In addition, all edges in $E$ which are incident to $u$ and $v$ in the original graph $G$ are now incident to $w$ instead. (Note that if $G$ has a vertex $z \neq u, v$ such that $(z, u), (z, v) \in E$, then in the new graph $G'$, we consider the edge $(z, w)$ to occur twice.)
\end{definition}

As part of the construction of the binary tree $\calT$, we contract all of the edges of $\calG$ in sequence (and form new vertices which will be members of $\calT$) until only one vertex of $\calG$ remains. In order to have $\calT(\{\widehat{U_v} \mid v \in V_\calT\}) = \calG(\{U_v \mid v \in V\})$ with a rank of at most $k^t$ on all the edges of $\calT$, it is enough for the degree of all the vertices of $G$ to be at most $t$ throughout the contraction process. The following result from \cite{ms08_simulating_quantum_tensor_contraction} achieves this with $t = O(\deg(G) \tw(G))$, where $\deg(G)$ denotes the maximum degree of any vertex of $G$, and $\tw(G)$ denotes the treewidth of $G$.

\begin{lemma} [Proposition 4.2 of \cite{ms08_simulating_quantum_tensor_contraction}]
\label{lemma:contraction_order}
Let $G = (V, E)$ be a graph. There exists an ordering $\pi$ of $E$ such that if the edges of $G$ are contracted in the order specified by $\pi$, then the degree of any vertex in $G$ at any point in the contraction process (including parallel edges, and counting self-loops as $1$ edge) is at most $O(\deg(G) \tw(G))$. Such an ordering can be computed in $|V|^{O(1)} e^{O(\deg(G) \tw(G))}$ time.
\end{lemma}
\begin{proof}
By Proposition 4.2 of \cite{ms08_simulating_quantum_tensor_contraction}, the maximum degree over the course of the contraction process is upper bounded by the treewidth $\tw(G^*)$ of the line graph $G^*$ of $G$, and such an ordering can be computed given a tree decomposition with this width. Furthermore, since $G^*$ has at most $|V|^2$ vertices, a tree decomposition of $G^*$ with treewidth $O(\tw(G^*))$ can be computed in $|V|^{O(1)} e^{O(\tw(G^*))}$ time, by Theorem 4.3 of \cite{ms08_simulating_quantum_tensor_contraction}. Finally, by Lemma 4.4 of \cite{ms08_simulating_quantum_tensor_contraction}, $\tw(G^*) \leq \deg(G) (\tw(G) + 1) - 1$, leading to the desired running time.
\end{proof}

\begin{algorithm}
\caption{Converting a general tensor network of rank $k$ into a tree network of rank $k^t$ given a contraction ordering with maximum degree $t$}
\label{algorithm:contract_general_network_into_tree}
\begin{algorithmic}
\Require A graph $\calG = (V, E)$ with a contraction ordering with maximum degree $t$, $\{U_v \mid v \in V\}$ such that each $U_v$ has one mode of dimension $n$ and $\deg(v)$ modes of dimension $k$.
\Ensure A binary tree $\calT = (V_\calT, E_\calT)$ and $\{\widehat{U_v} \mid v \in V_\calT\}$ such that each $\widehat{U_v}$ has at most one mode of dimension $n$, and each mode of $\widehat{U_v}$ corresponding to an edge of $\calT$ has dimension at most $k^t$

\item []
\State // Initialization
\State $\pi \gets $ contraction ordering of $\calG$ with maximum degree $t$
\For{$v \in V$}
    \State $V_v \gets \{v\}$
    \State $E_v \gets \varnothing$
    \State $\calT_v \gets (V_v, E_v)$ (i.e. the singleton tree consisting of $v$)
\EndFor

\item []
\State // Performing edge contractions
\For{$e = (u, v) \in \pi$}
    \State $L_{1, u} \gets $ the modes of $U_u$ which correspond to edges in $\calG$
    \State $L_{2, u} \gets $ the remaining modes of $U_u$
    \State $U_u, U_u' \gets $ tensors such that 
    \State \,\,\,\,\,\,\,\,\,\,\,\,\,\,\,\, (1) $U_u \circ U_u'$ is equal to the original $U_u$
    \State \,\,\,\,\,\,\,\,\,\,\,\,\,\,\,\, (2) The edge connecting $U_u$ and $U_u'$ is of rank $k^t$
    \State \,\,\,\,\,\,\,\,\,\,\,\,\,\,\,\, (3) The remaining modes of $U_u$ correspond to the modes in $L_{1, u}$
    \State \,\,\,\,\,\,\,\,\,\,\,\,\,\,\,\, (4) The remaining modes of $U_u'$ correspond to the modes in $L_{2, u}$
    \State $L_{1, v} \gets $ the modes of $U_v$ which correspond to edges in $\calG$
    \State $L_{2, v} \gets $ the remaining modes of $U_v$
    \State $U_v, U_v' \gets $ tensors such that 
    \State \,\,\,\,\,\,\,\,\,\,\,\,\,\,\,\, (1) $U_v \circ U_v'$ is equal to the original $U_v$
    \State \,\,\,\,\,\,\,\,\,\,\,\,\,\,\,\, (2) The edge connecting $U_v$ and $U_v'$ is of rank $k^t$
    \State \,\,\,\,\,\,\,\,\,\,\,\,\,\,\,\, (3) The remaining modes of $U_v$ correspond to the modes in $L_{1, v}$
    \State \,\,\,\,\,\,\,\,\,\,\,\,\,\,\,\, (4) The remaining modes of $U_v'$ correspond to the modes in $L_{2, v}$
    \State $w \gets $ the vertex resulting from the contraction of $u$ and $v$
    \State $U_w \gets U_u \circ U_v$ where the contraction occurs along the edge $(u, v)$
    \State $\calT_u \gets \calT_u$ with $u$ replaced by $u'$
    \State $\calT_v \gets \calT_v$ with $v$ replaced by $v'$
    \State $\calT_w \gets $ the tree such that
    \State \,\,\,\,\,\,\,\,\,\,\,\,\,\,\,\, (1) $w$ is the root
    \State \,\,\,\,\,\,\,\,\,\,\,\,\,\,\,\, (2) The children of $w$ are $u'$ and $v'$
    \State \,\,\,\,\,\,\,\,\,\,\,\,\,\,\,\, (3) The subtrees of $u'$ and $v'$ are $\calT_u$ and $\calT_v$ respectively
\EndFor

\State $w \gets $ the single vertex of $\calG$ resulting from the contraction
\State \Return $\calT_w$, $\{U_v \mid v \text{ is a vertex of } \calT_w\}$
\end{algorithmic}
\end{algorithm}

\begin{theorem} \label{thm:contract_general_network_into_tree}
Let $\calG = (V, E)$ be a graph, with corresponding factors $\{U_v \mid v \in V\}$. Then the tree network $\calT_w(\{U_v \mid v \text{ is a vertex of }\calT_w\})$ returned by Algorithm \ref{algorithm:contract_general_network_into_tree} is equal to $\calG(\{U_v \mid v \in V\})$.
\end{theorem}
\begin{proof}
It suffices to show that at any point in Algorithm \ref{algorithm:contract_general_network_into_tree}, when contracting an edge $(u, v)$, it is possible to split $U_v$ into two tensors $U_v$ and $U_{v'}$ such that the edge connecting $U_v$ and $U_{v'}$ has rank at most $k^t$. This follows from the fact that (excluding the vertices in $\calT_v$) $v$ has degree at most $t$ within $\calG$. Thus, the modes of $U_v$ corresponding to the edges incident to $v$ in $\calG$ together have dimension at most $k^t$ (since the rank on any edge in $\calG$ is $k$). The desired result then follows from matricizing $U_v$ with the row dimension corresponding to these modes.
\end{proof}

\subsection{Approximating $A$ with the Resulting Tree Network}

We use the following procedure when contracting $\calG$. At each point in this procedure, each vertex $u$ in $\calG$ has a binary tree $\calT_u$ attached to it ($\calT_u$ is only connected to $\calG$ through $u$ and initially consists only of $u$). When contracting an edge $(u, v)$, we create vertices $u'$ and $v'$ such that $u$ is connected to $u'$ and $u'$ is connected to $T_u$ (i.e., $u$ is no longer directly connected to the rest of $T_u$), and the same is true of $v$ and $v'$. If $U_u$ has any modes of dimension $n$, then those modes will now be a part of the factor $U_{u'}$ corresponding to $u'$, and all the modes that initially connected $U_u$ to $\calT_u$ will now belong to $U_{u'}$. This means that aside from the edge connecting $u$ and $u'$, the only modes of $U_u$ are the ones corresponding to the edges between $u$ and the rest of $\calG$. Since $u$ has degree at most $O(\deg(\calG)\tw(\calG))$, this means that these modes together have dimension $k^{O(\deg(\calG) \tw(\calG))}$, meaning $U_u$ has rank $k^{O(\deg(\calG) \tw(\calG))}$ when matricized along these modes, and we can let the edge between $u$ and $u'$ have rank $k^{O(\deg(\calG) \tw(\calG))}$. The final tree output by our algorithm will only consist of edges of the form $(u, u')$ obtained in this way, meaning that all of the edges in this tree will have rank $k^{O(\deg(\calG)\tw(\calG))}$. Thus, we can approximate $A$ using our bicriteria algorithm for tree networks, with $\calT$ as the target tree and $k^{O(\deg(\calG)\tw(\calG))}$ as the target rank.

\begin{algorithm}
\caption{Algorithm to approximate $A$ using a binary tree network with ranks $k^{O(\deg(\calG) \tw(\calG))}$, achieving at most the error of a general network according to a graph $\calG$ (up to a $1 + \eps$ factor). Note that in the second step of this algorithm, we apply Algorithm \ref{algorithm:contract_general_network_into_tree} --- Algorithm \ref{algorithm:contract_general_network_into_tree} can simply use the graph $\calG$ together with the contraction order to output the binary tree $\calT$, by ignoring the steps involving tensor contraction.}
\label{algorithm:approx_general_network_with_tree}
\begin{algorithmic}
\Require $A \in \R^{n \times \ldots \times n}$, a graph $\calG = (V, E)$ such that each vertex has one mode of dimension $n$, and each edge has dimension $k$
\Ensure A binary tree $\calT$ with $O(q)$ vertices, $\{U_v \mid v \text{ is a vertex of }\calT\}$ such that $U_v$ has at most one mode of dimension $n$ and every other mode of $U_v$ has dimension $O(\frac{qk^{\deg(\calG)\tw(\calG)}}{\eps} \log(\frac{q}{\delta}))$

\State $\pi \gets $ contraction order of $\calG$ with maximum degree $O(\deg(\calG) \tw(\calG))$
\State $\calT \gets $ the tree that is output by Algorithm \ref{algorithm:contract_general_network_into_tree} given $\calG$ and $\pi$
\State $\{U_v \mid v \text{ is a vertex of }\calT\} \gets $ the output of Algorithm \ref{algorithm:tree_network_bicriteria} with $A$ as the input tensor and $\calT$ as the tree, and target rank $k^{\deg(\calG)\tw(\calG)}$
\State \Return $\calT$, $\{U_v \mid v \text{ is a vertex of }\calT\}$
\end{algorithmic}
\end{algorithm}

\begin{theorem}
Let $A \in \R^{n \times \ldots \times n}$ be a $q$-mode tensor, and let $\calG = (V, E)$ be a graph with $|V| = q$ such that each vertex has one mode of dimension $n$, and the edges each have dimension $k$. Then, Algorithm \ref{algorithm:approx_general_network_with_tree} outputs a binary tree $\calT$ and factors $\{U_v \mid v \text{ is a vertex of } \calT\}$ such that
\begin{equation}
\begin{split}
\|\calT(\{U_v \mid v \text{ is a vertex of }\calT\}) - A\|_F
& \leq (1 + \eps) \|\calG(\{U_v \mid v \in V\}) - A\|_F
\end{split}
\end{equation}
with probability at least $1 - \delta$. The running time of Algorithm \ref{algorithm:approx_general_network_with_tree} is $O(q \cdot \nnz(A)) + n \cdot \poly(\frac{q}{\eps\delta}) k^{O(\deg(\calG)\tw(\calG))}$.
\end{theorem}
\begin{proof}
The approximation guarantee follows directly from Theorem \ref{thm:contract_general_network_into_tree} and Theorem \ref{thm:analysis_of_tree_network_bicriteria}. We now calculate the running time. By Lemma \ref{lemma:contraction_order}, calculating the contraction order $\pi$ can be done in $\poly(q) e^{O(\deg(\calG)\tw(\calG)}$ time. In addition, given the contraction ordering, the remaining steps of Algorithm \ref{algorithm:contract_general_network_into_tree} take $\poly(q)$ time. Finally, applying Algorithm \ref{algorithm:tree_network_bicriteria} with $A$ and $\calT$ as inputs, and target rank $k^{O(\deg(\calG)\tw(\calG)}$, takes $O(q \cdot \nnz(A)) + n \cdot \Big(\frac{qk^{\deg(\calG)\tw(\calG)}}{\eps \delta}\Big)^{O(1)}$ time (note that the maximum degree of $\calT$ is $O(1)$ since it is a binary tree). Thus, the overall running time of Algorithm \ref{algorithm:approx_general_network_with_tree} is $O(q \cdot \nnz(A)) + n \cdot \poly(\frac{q}{\eps\delta}) k^{O(\deg(\calG)\tw(\calG))}$.
\end{proof}

\newpage
\section{Fixed Parameter Tractable (FPT)  Algorithms}
\label{sec:frobenius_norm_1_plus_eps_rank_exactly_k}

\subsection{Preliminaries on Gaussian Matrices with Exactly $k$ Rows}

In this section, we first develop the properties of random matrices with exactly $k$ rows and i.i.d. $N(0, 1/k)$ entries. We show that these can be used to obtain subspace embeddings and perform multiple-response $\ell_2$ regression with $e^{-\poly(k/\eps)}$ success probability --- our algorithm for Tucker-$(p, q)$ decomposition shown in the next subsection will instantiate these $\ell_2$ regression sketches $e^{\poly(k/\eps)}$ times to obtain a constant probability of success. We start by showing the subspace embedding property.

\begin{lemma} [$e^{-\poly(k/\eps)}$ Success Probability Subspace Embedding]
\label{lemma:gaussian_with_exactly_k_rows_o(1)_subspace_embedding}
Let $A \in \R^{n \times k}$ have rank $k$, and let $S \in \R^{k \times n}$ have i.i.d. $N(0, 1/k)$ entries. Then, with probability at least $e^{-\Theta(k^2 \log k)}$, for all $x \in \R^k$, $\frac{2}{3} \|Ax\|_2 \leq \|SAx\|_2 \leq \frac{2}{3} \|Ax\|_2$.
\end{lemma}
\begin{proof}
Let $U \in \R^{n \times k}$ be an orthonormal basis for the column span of $A$ --- then it suffices to show that with the desired probability, for all $x \in \R^k$, $\frac{2}{3} \|x\|_2 \leq \|SUx\|_2 \leq \frac{4}{3} \|x\|_2$. First observe that by the rotational invariance of the Gaussian distribution, $SU \in \R^{k \times k}$ has i.i.d. $N(0, 1/k)$ entries. Let $M = SU$, and let $d > 1$ be an absolute constant. For $i \in [k]$ and $j \in [k]$, note that the probability density function of $M_{ij}$ is $$f_X(x) = \frac{1}{\sigma\sqrt{2\pi}} e^{-\frac{(x - \mu)^2}{2\sigma^2}}$$
For $i\in [k], j\in[k]$, note that the probability density function of $M_{ij}$ is $f(x) = \sqrt{\frac{k}{2\pi}}e^{-\frac{kx^2}{2}}$. Thus, for fixed $i, j\in [k], i\neq j$, and a small constant $E < \frac 1 2$, $$\PrBig{|M_{ij}|\leq \frac{E}{k^d}}\geq \sqrt{\frac 2 \pi}\cdot\frac{Ee^{-\frac 1 2}}{k^d}$$ This is true since for $x\in [-E/k^d, E/k^d]$, because $\mathcal{N}(0, \frac 1 k)$ has $\mu = 0$, $f(x)\geq f(E/k^d)$. Expressing the cumulative density function with the Riemann integral, since the p.d.f is non-negative,
$$\PrBig{-E/k^d \leq M_{ij}\leq E/k^d}\geq 2\cdot\frac{E}{k^d}\cdot f(E/k^d) $$
and $e^{-\frac{kE^2}{2k^{2d - 1}}} \geq e^{-1/2}$ when $k, d\geq 1$ and $E < 1$.
\begin{align*}
    \PrBig{|M_{ij}|\leq \frac{E}{k^d}}&\geq \sqrt{\frac{k}{2\pi}}e^{-\frac{kx^2}{2}}\cdot \frac{2E}{k^d} \\
    &= \sqrt{\frac 2 \pi}e^{-\frac{k(E/k^d)^2}{2}}\cdot \frac{E}{k^{d - \frac 1 2}} \\
    &= \sqrt{\frac 2 \pi}e^{-\frac{E^2}{2k^{2d - 1}}}\cdot \frac{E}{k^{d - \frac 1 2}} \\
    &\geq \sqrt{\frac 2 \pi}\frac{Ee^{-1/2}}{k^d}
\end{align*}
The last inequality is true because $k^{d - \frac 1 2}\leq k^d$ for $k, d\geq 1$.
Similarly,
\begin{align*}
    \PrBig{M_{ii} = 1\pm \frac{E}{k^d}} &\geq f(2)\cdot \frac{2E}{k^d} \\
    &\geq \sqrt{\frac 2 \pi}\cdot e^{-2k}\cdot \frac{E}{k^d}
\end{align*}

Now, let $\calE$ denote the event that $\|M - I\|_\infty \leq \frac{E}{k^d}$ --- here, $I \in \R^{k \times k}$ is the identity matrix, and $\|\cdot \|_\infty$ denotes the entrywise $\ell_\infty$ norm of a matrix. Then, by the above,
\begin{equation}
\begin{split}
\PrBig{\calE}
& \geq \PrBig{|M_{ij}|  \geq \frac{E}{k^d} \text{ for } i \neq j \text{ and } |M_{ii}| = 1 \pm \frac{E}{k^d} \text{ for } i \in [k]} \\
& = \prod_{i \neq j} \PrBig{|M_{ij}| \geq \frac{E}{k^d}} \prod_{i \in [k]} \PrBig{|M_{ii}| = 1 \pm \frac{E}{k^d}} \\
& \geq \prod_{i \neq j} \sqrt{\frac 2\pi}\frac{Ee^{-1/2}}{k^d} \prod_{i \in [k]}\sqrt{\frac 2\pi} \frac{Ee^{-2k}}{k^d} \\
& \geq \Omega\Big(\frac{1}{k^d}\Big)^{k^2} \cdot \Omega\Big(\frac{e^{-2k}}{k^d}\Big)^{k} \\
& \geq e^{-\Theta(k^2 \log k)} \cdot e^{-\Theta(k^2)} e^{-\Theta(k \log k)} \\
& \geq e^{-\Theta(k^2 \log k)}
\end{split}
\end{equation}
Furthermore, note that if $\calE$ holds, then we can write $M = I + T$, where $\|T\|_\infty \leq \frac{E}{k^d}$, meaning that
$$\|T\|_2 \leq \|T\|_F \leq k \|T\|_\infty \leq \frac{E}{k^{d - 1}}$$
meaning that for all $x \in \R^k$,
$$|\|SUx\|_2 - \|x\|_2| \leq \|Tx\|_2 \leq \frac{E}{k^{d - 1}} \|x\|_2 \leq E \|x\|_2 \leq \frac{1}{3} \|x\|_2$$
where the second-to-last inequality holds for $d \geq 1$, and the last inequality holds for $E \leq \frac{1}{3}$. In summary, with probability at least $e^{-\Theta(k^2 \log k)}$, for all $x \in \R^k$,
$$\frac{2}{3} \|Ax\|_2 \leq \|SAx\|_2 \leq \frac{4}{3} \|Ax\|_2$$
This completes the proof of the lemma.
\end{proof}

Next, we show that $S \in \R^{k \times n}$ with i.i.d. $N(0, 1/k)$ entries also satisfies a special case of the approximate matrix product property in Definition \ref{def:approx_matrix_product}, where the success probability is now $e^{-\poly(k/\eps)}$ instead of $1 - \delta$. 

\begin{lemma}[$e^{-\poly(kt/\eps)}$ Success Probability Approximate Matrix Product]
\label{lemma:gaussian_with_exactly_k_rows_approximate_matrix_product}
Let $n, k, t \in \N$ and $\eps > 0$, and $A \in \R^{n \times k}$ and $B \in \R^{n \times t}$. Assume that for $i, j \in [k]$ distinct, $\langle A_i, A_j \rangle = 0$, and for $i, j \in [t]$ distinct, $\langle B_i, B_j \rangle = 0$. Further assume that for all $i \in [k]$ and $j \in [t]$, $\langle A_i, B_j \rangle = 0$. Let $S \in \R^{k \times n}$ have i.i.d. $N(0, 1/k)$ entries. Then, with probability $e^{-\Theta(k^2 \log k)} \cdot (\frac{\eps}{k})^{O(kt)}$, the following hold:
\begin{itemize}
    \item $\|A^T S^T S B\|_F^2 \leq \frac{\eps}{k} \|A\|_F^2 \|B\|_F^2$, and
    \item For all $x \in \R^k$, $\frac{2}{3} \|Ax\|_2 \leq \|SAx\|_2 \leq \frac{4}{3} \|Ax\|_2$.
\end{itemize}
\end{lemma}

\begin{remark}
Note that the second property is the conclusion of Lemma \ref{lemma:gaussian_with_exactly_k_rows_o(1)_subspace_embedding} --- here we restate it to clarify that the approximate matrix product and the subspace embedding property hold simultaneously (note that we cannot union bound to show this, as in usual arguments, since the success probabilities are exponentially small).
\end{remark}

\begin{proof}
Observe that if we let $M = [A, B]$, then the columns of $M$ are orthogonal, meaning the entries of $SM$ are jointly independent. In particular, let $\calE$ be the event that for all $x \in \R^k$, $\frac{2}{3} \|Ax\|_2 \leq \|SAx\|_2 \leq \frac{4}{3} \|Ax\|_2$ --- by Lemma \ref{lemma:gaussian_with_exactly_k_rows_o(1)_subspace_embedding}, this holds with probability $e^{-\Theta(k^2 \log k)}$. Since $SA$ and $SB$ are independent, if we condition on $\calE$, the entries of $SB$ are jointly independent.

Now, it suffices to show that, for each $i \in [k]$ and $j \in [t]$, if we condition on $\calE$, then with probability $e^{-\poly(kt/\eps)}$, $|A_i^T S^T S B_j|^2 \leq \frac{\eps}{k} \|A_i\|_2^2 \|B_j\|_2^2$. Note that if $\calE$ holds, then $\|SA_i\|_2 \leq \frac{4}{3}\|A_i\|_2$, meaning that
$$|A_i^T S^T S B_j|^2 \leq \|SA_i\|_2^2 \|SB_j\|_2^2 \leq O(1) \|A_i\|_2^2 \|SB_j\|_2^2$$
where the first inequality is by the Cauchy-Schwarz inequality and the second is by conditioning on $\calE$.

Let us now analyze the probability that $\|SB_j\|_2^2 \leq \Theta(\frac{\eps}{k}) \|B_j\|_2^2$. Since the Gaussian distribution is $2$-stable, $SB_j$ is a random vector in $\R^k$ with $N(0, \frac{\|B_j\|_2^2}{k})$ entries. For convenience, let $X = SB_j$. Then, for $i \in [k]$, $Y_i = \frac{X_i}{\|B_i\|_2/\sqrt{k}}$ is drawn from $N(0, 1)$, meaning
$$\PrBig{|X_i| \leq \Theta\Big(\frac{\sqrt{\eps}}{k}\Big)\|B_j\|_2} \geq \PrBig{|Y_i| \leq \Theta\Big(\frac{\sqrt{\eps}}{\sqrt{k}}\Big)} \geq \Theta\Big(\frac{\sqrt{\eps}}{\sqrt{k}}\Big)$$
Thus,
$$\PrBig{|X_i| \leq \Theta\Big(\frac{\sqrt{\eps}}{k}\Big) \|B_j\|_2 \text{ for all }i \in [k]} \geq \Theta\Big(\sqrt{\frac{\eps}{k}}\Big)^k \geq \Big(\frac{\eps}{k}\Big)^{O(k)}$$
Let $\calE_j$ denote the event that each coordinate of $SB_j$ is at most $O(\frac{\sqrt{\eps}}{k}) \|B_j\|_2$ in absolute value --- we have shown that $\Pr{\calE_j} \geq (\eps/k)^{O(k)}$. Note that if $\calE_j$ holds, then
$$\|SB_j\|_2 \leq \sqrt{k} \|SB_j\|_\infty \leq \sqrt{k} \cdot O\Big(\frac{\sqrt{\eps}}{k}\Big) \|B_j\|_2 \leq O\Big(\sqrt{\frac{\eps}{k}}\Big) \|B_j\|_2$$
In summary, if $\calE$ and $\calE_j$ hold, then for all $i \in [k]$,
$$|A_i^T S^T S B_j|^2 \leq O(1) \|A_i\|_2^2 \|SB_j\|_2^2 \leq O\Big(\frac{\eps}{k}\Big) \|A_i\|_2^2 \|B_j\|_2^2$$
meaning that if $\calE$ holds and $\calE_j$ holds for all $j \in [t]$, then
$$\|A^T S^T S B\|_F^2 = \sum_{i \in [k]} \sum_{j \in [t]} |A_i^T S^T S B_j|^2 \leq \sum_{i \in [k]} \sum_{j \in [t]} O\Big(\frac{\eps}{k}\Big) \|A_i\|_2^2 \|B_j\|_2^2 \leq O\Big(\frac{\eps}{k}\Big) \|A\|_F^2 \|B\|_F^2$$
The probability that $\calE$ holds and $\calE_j$ holds for all $j \in [t]$ is
$$\PrBig{\calE \cap \calE_1 \cap \ldots \cap \calE_t} = \PrBig{\calE} \cdot \prod_{i = 1}^t \PrBig{\calE_i} \geq e^{-\Theta(k^2 \log k)} \cdot \prod_{i = 1}^t \Big(\frac{\eps}{k}\Big)^{O(k)} = e^{-\Theta(k^2 \log k)} \cdot \Big(\frac{\eps}{k}\Big)^{O(kt)}$$
since the entries of $SB$ are jointly independent when conditioning on $\calE$. This proves the lemma.
\end{proof}

Using the above two lemmas, we show that random matrices with exactly $k$ rows and i.i.d. $N(0, 1/k)$ entries can be used for multiple-response $\ell_2$ regression with success probability $e^{-\poly(k/\eps)}$, following the proof of Theorem 3.1 in \cite{cw09_numerical_linear_algebra_streaming}.

\begin{lemma}[$e^{-\poly(kt/\eps)}$ Success Probability for Multiple-Response $\ell_2$ Regression] \label{lemma:gaussian_with_exactly_k_rows_multiple_response_regression}
Let $n, k, s, t \in \N$, $\eps > 0$, and $k, t < n$. Suppose $A \in \R^{n \times k}$ has rank $k$ and $B \in \R^{n \times s}$ has rank $t$. Let $R \in \R^{k \times n}$ have i.i.d. $N(0, 1/k)$ entries. If $X^* = \argmin \|AX - B\|_F^2$ and $\widehat{X} = \argmin \|RAX - RB\|_F^2$, then $\|A\widehat{X} - B\|_F^2 \leq (1 + \eps) \|AX^* - B\|_F^2$ with probability at least $e^{-\Theta(k^2 \log k)} (\frac{\eps}{k})^{O(kt)}$.
\end{lemma}
\begin{proof}
The proof follows that of Theorem 3.1 in \cite{cw09_numerical_linear_algebra_streaming}. By the normal equations and the Pythagorean theorem,
$$\|A\widehat{X} - B\|_F^2 = \|A(\widehat{X} - X^*) + AX^* - B\|_F^2 = \|A(\widehat{X} - X^*)\|_F^2 + \|AX^* - B\|_F^2$$
since the columns of $AX^* - B$ are orthogonal to those of $A$. Now we bound $\|A(\widehat{X} - X^*)\|_F$ as in the proof of Lemma 3.3 of \cite{cw09_numerical_linear_algebra_streaming}. Let $U \in \R^{n \times k}$ be an orthonormal basis for the column span of $A$. Condition on the event $\calE$, that $\|U^T R^T R (AX^* - B)\|_F^2 \leq \frac{\eps}{k} \|U\|_F^2 \|AX^* - B\|_F^2 = \eps \|AX^* - B\|_F^2$ and for all $x \in \R^k$, $\frac{2}{3} \|x\|_2 \leq \|RUx\|_2 \leq \frac{4}{3} \|x\|_2$ --- by Lemma \ref{lemma:gaussian_with_exactly_k_rows_approximate_matrix_product}, $\calE$ occurs with probability at least $e^{-\Theta(k^2 \log k)} e^{-k(k + t)}$. Then,
\begin{equation}
\begin{split}
\|A(\widehat{X} - X^*)\|_F
& = \|UU^T A(\widehat{X} - X^*)\|_F \\
& = \|U^T A(\widehat{X} - X^*)\|_F \\
& = \|(U^T R^T R U) U^T A(\widehat{X} - X^*) + (I - U^T R^T R U) U^T A(\widehat{X} - X^*)\|_F \\
& \leq \|U^T R^T R UU^T A(\widehat{X} - X^*)\|_F + \|(I - U^T R^T R U) U^T A (\widehat{X} - X^*)\|_F \\
& = \|U^T R^T R A(\widehat{X} - X^*)\|_F + \|I - U^T R^T R U\|_2 \|U^T A (\widehat{X} - X^*)\|_F \\
& \leq \|U^T R^T R A(\widehat{X} - X^*)\|_F + \frac{1}{3} \|U^T A (\widehat{X} - X^*)\|_F \\
& = \|U^T R^T R A(\widehat{X} - X^*)\|_F + \frac{1}{3} \|A (\widehat{X} - X^*)\|_F
\end{split}
\end{equation}
where the last inequality is because $R$ is a $(1 \pm \frac{1}{3})$ $\ell_2$ subspace embedding for $A$. Rearranging gives
$$\|A(\widehat{X} - X^*)\|_F \leq \frac{3}{2} \|U^T R^T R A (\widehat{X} - X^*)\|_F$$
We bound the right hand side using the inequality $\|U^T R^T R (AX^* - B)\|_F^2 \leq \eps \|AX^* - B\|_F^2$ (recall that we are conditioning on $\calE$). By the normal equations for the problem $\min_X \|RAX - RB\|_F$,
$$U^T R^T R (A\widehat{X} - B) = A^T R^T R (A\widehat{X} - B) = 0$$
meaning that
$$U^T R^T R A(\widehat{X} - X^*) = U^T R^T R (A\widehat{X} - B + B - AX^*) = U^T R^T R (B - AX^*)$$
and therefore, conditioning on $\calE$,
$$\|A(\widehat{X} - X^*)\|_F \leq \frac{3}{2} \|U^T R^T R A(\widehat{X} - X^*)\|_F = \frac{3}{2} \|U^T R^T R (B - AX^*)\|_F \leq O(\eps) \|AX^* - B\|_F^2$$
Therefore,
\begin{equation}
\begin{split}
\|A\widehat{X} - B\|_F^2
& = \|A(\widehat{X} - X^*)\|_F^2 + \|AX^* - B\|_F^2 \\
& = O(\eps) \|AX^* - B\|_F^2 + \|AX^* - B\|_F^2 \\
& = (1 + O(\eps)) \|AX^* - B\|_F^2
\end{split}
\end{equation}
and taking square roots gives
$$\|A\widehat{X} - B\|_F \leq \sqrt{1 + O(\eps)} \|AX^* - B\|_F \leq (1 + O(\eps)) \|AX^* - B\|_F$$
This completes the proof of the lemma.
\end{proof}
In the algorithms described below, the previous two lemmas will be used with $t\in O(\frac{k}{\eps})$

\subsection{FPT Tucker-$(p, q)$ Decomposition}
\label{subsection:fpt_tucker_rank_exactly_k}

We now prove the approximation error and running time guarantees of Algorithm \ref{algorithm:fpt_rank_exactly_k_1_plus_eps_approximation}.

\begin{proof}[Proof of Theorem \ref{thm:strongly_fpt_time_algorithm_tucker}]
We split the analysis of Algorithm \ref{algorithm:fpt_rank_exactly_k_1_plus_eps_approximation} into three parts: (1) we show that among the guesses $U^1, U^2, \ldots, U^p$, there is a good solution for Tucker-$(p, q)$, (2) we then show that, using the Kronecker product regression algorithm of \cite{djssw19_kronecker_product_regression} (with very minor modifications) we can identify that good solution $U^1, U^2, \ldots, U^p$, and (3) we analyze the running time.

\paragraph{Part 1: Existence of a Good Guess:} Let $T$ be the number of trials per mode in Algorithm \ref{algorithm:fpt_rank_exactly_k_1_plus_eps_approximation} --- that is, the number of guesses we make for each of $U^1, U^2, \ldots, U^p$. We will determine how to select a sufficiently large value of $T$ at the end of the algorithm. For the purpose of our analysis, for a mode $m \in [p]$, we use the notation $U^m_t$ to refer to the guess for $U^m$ that is generated on the $t^{th}$ trial. Throughout the proof, fix an optimal solution $U^1_*, U^2_*, \ldots, U^p_* \in \R^{n \times k}$, and $G_* \in \R^{k \times \ldots \times k \times n \times \ldots \times n}$. We will use an argument reminiscent of the ``iterative existential argument'' of \cite{swz19_tensor_low_rank} to show that with the desired success probability, there exist $\widehat{U^1}, \ldots \widehat{U^p}$, with $\widehat{U^1}$ among the guesses for $U^1$, $\widehat{U^2}$ among the guesses for $U^2$, and so on, such that
\begin{equation*}
\begin{split}
\min_G \|(\widehat{U^1} \times \ldots \times \widehat{U^p}  &\times I_n \times \ldots \times I_n) \tensorvec(G) - \tensorvec(A)\|_F  \\ &\leq (1 + \eps) \|(U^1_* \times \ldots \times U^p_* \times I_n \times \ldots \times I_n) \tensorvec(G_*) - \tensorvec(A)\|_F
\end{split}
\end{equation*}
For each mode $m \in [p]$, let $\calS_m \subset \R^{n \times k}$ be the set of all guesses made for $U^m$ over the course of Algorithm \ref{algorithm:fpt_rank_exactly_k_1_plus_eps_approximation}. Define
\begin{equation}
\begin{split}
\widehat{U^1} 
& = \argmin_{U \in \calS_1} \|(U \times U^2_* \times \ldots \times U^p_* \times I_n \times \ldots \times I_n) \tensorvec(G_*) - \tensorvec(A) \|_F \\
& = \argmin_{U \in \calS_1} \|U M_1(G_*) (U^2_* \times \ldots \times U^p_* \times I_n \times \ldots \times I_n)^T - M_1(A)\|_F
\end{split}
\end{equation}
For each $m \in [p]$, $m \neq 1, p$, define
\begin{equation} 
\begin{split} 
\widehat{U^m} 
& = \argmin_{U \in \calS_m} \|(\widehat{U^1} \times \ldots \times \widehat{U^{m - 1}} \times U \times U^{m + 1}_* \times \ldots \times U^p_* \times I_n \times \ldots \times I_n) \tensorvec(G_m) - \tensorvec(A)\|_F \\
& = \argmin_{U \in \calS_m} \|U M_m(G_m) (\widehat{U^1} \times \ldots \times \widehat{U^{m - 1}} \times U^{m + 1}_* \times \ldots \times U^p_* \times I_n \times \ldots \times I_n)^T - M_m(A)\|_F
\end{split}
\end{equation}
where $G_m$ is the optimal core tensor corresponding to $\widehat{U^1}, \ldots, \widehat{U^{m - 1}}, U^m_*, \ldots, U^p_*$, and define
\begin{equation}
\begin{split}
\widehat{U^p} 
& = \argmin_{U \in \calS_p} \|(\widehat{U^1} \times \ldots \times \widehat{U^{p - 1}} \times U \times I_n \times \ldots \times I_n) \tensorvec(G_p) - \tensorvec(A)\|_F \\
& = \argmin_{U \in \calS_p} \|U M_p(G_p) (\widehat{U^1} \times \ldots \times \widehat{U^{p - 1}} \times I_n \times \ldots \times I_n)^T - M_p(A)\|_F
\end{split}
\end{equation}
For convenience, define $W_1 = M_1(G_*) (U^2_* \times \ldots \times U^p_* \times I_n \times \ldots \times I_n)^T$, for $m \in [p]$, $m \neq 1, p$, define $W_m = M_m(G_m) (\widehat{U^1} \times \ldots \times \widehat{U^{m - 1}} \times U^{m + 1}_* \times \ldots \times U^p_* \times I_n \times \ldots \times I_n)^T$, and define $W_p = M_p(G_p) (\widehat{U^1} \times \ldots \times \widehat{U^{p - 1}} \times I_n \times \ldots \times I_n)^T$.

Now, fix a mode $m \in [p]$. First, let $\calE^m$ be the event that the Countsketch matrix $S_m$ is an affine embedding for $W_m$ and $M_m(A)$, that is, for all $U \in \R^{n \times k}$,
$$\|UW_m S_m^T - M_m(A) S_m^T\|_F = \Big(1 \pm O\Big(\frac{\eps}{p}\Big)\Big) \|UW_m - M_m(A)\|_F$$
Observe that if $S_m$ has $O(\frac{p^3k^2}{\eps^2})$ rows, then by Theorem \ref{thm:countsketch_affine_embedding}, $\calE^m$ holds with probability at least $1 - O(1/p)$ (note that $S_m$ is independent of $W_m$). Next, let $\calF^m$ be the event that the random sign matrix $T_m$ preserves $(1 + \frac{\eps}{p^2})$-approximate $\ell_2$ regression solutions to the problem $\min_U \|UW_m S_m^T - M_m(A) S_m^T\|_F$. In other words, let $\calF^m$ be the event that for any $U \in \R^{n \times k}$ for which
$$\|UW_mS_m^T T_m^T - M_m(A) S_m^T T_m^T\|_F \leq \Big(1 + O\Big(\frac{\eps}{p^2}\Big)\Big) \min_U \|UW_mS_m^T T_m^T - M_m(A) S_m^T T_m^T\|_F$$
it also holds that
$$\|U W_m S_m^T - M_m(A) S_m^T\|_F \leq \Big(1 + O\Big(\frac{\eps}{p}\Big)\Big) \min_U \|U W_m S_m^T - M_m(A) S_m^T\|_F$$
Then, by Lemma \ref{thm:approximate_multiple_l2_regression} (with $\eps$ replaced by $\frac{\eps}{p}$, $\tau$ replaced by $\frac{\eps}{p^2}$, and $\delta = \frac{1}{p}$), $\calF^m$ also holds with probability at least $1 - O(1/p)$ as long as $T_m$ has at least $O((pk \log p)/\eps)$ rows (note again that $T_m$ is independent of $S_m$ and $W_m$).

Now, for any $t \in [T]$, let $\calG^m_t$ be the event that the Gaussian matrix $R_{m, t}$ with $k$ rows and i.i.d. $N(0, 1/k)$ entries, produced on the $t^{th}$ trial, obtains a $(1 + \frac{\eps}{p^2})$-approximate $\ell_2$ regression solution to the above problem (note that the Gaussian matrix is denoted $R_m$ in Algorithm \ref{algorithm:fpt_rank_exactly_k_1_plus_eps_approximation}, but for the purposes of the analysis we denote it as $R_{m, t}$). In other words, $\calG^m_t$ is the event that if we define
$$U^m_t = \argmin_U \|U W_m S_m^T T_m^T R_{m, t}^T - M_m(A) S_m^T T_m^T R_{m, t}^T \|_F = M_m(A) S_m^T T_m^T R_{m, t}^T (W_mS_m^T T_m^T R_{m, t}^T)^\dagger$$
then
$$\|U^m_t W_m S_m^T T_m^T - M_m(A) S_m^T T_m^T\|_F \leq \Big(1 + O\Big(\frac{\eps}{p^2}\Big)\Big) \min_U \|U W_m S_m^T T_m^T - M_m(A) S_m^T T_m^T\|_F$$
By Lemma \ref{lemma:gaussian_with_exactly_k_rows_multiple_response_regression}, for a fixed $m \in [p]$ and $t \in [T]$, $\calG^m_t$ holds with probability at least $(\frac{\eps}{pk})^{O(\frac{pk^2 \log p}{\eps})}$, where we replaced $\eps$ with $\frac{\eps}{p^2}$ in the statement of Lemma \ref{lemma:gaussian_with_exactly_k_rows_multiple_response_regression}, and also used the fact that $T_m$ has $t = O(\frac{p k \log p}{\eps})$ rows. (Note that above we have defined $U^m_t = M_m(A) S_m^T T_m^T R_{m, t}^T (W_m S_m^T T_m^T R_{m, t}^T)^\dagger$, whereas in Algorithm \ref{algorithm:fpt_rank_exactly_k_1_plus_eps_approximation} we define $U^m_t = M_m(A) S_m^T T_m^T R_{m, t}^T$. These are essentially equivalent, since the factor $(W_m S_m^T T_m^T R_{m, t}^T)^\dagger \in \R^{k \times k}$ just amounts to a change of basis --- thus, for the purpose of the algorithm, this factor is not needed (note that it anyways cannot be computed efficiently) and for the purpose of the analysis, we can simply redefine $G_{m + 1}$ later in the analysis, using the value $U^m_t = M_m(A) S_m^T T_m^T R_{m, t}^T$ instead of $M_m(A) S_m^T T_m^T R_{m, t}^T (W_mS_m^T T_m^T R_{m, t}^T)^\dagger$. Let $f = (\frac{\eps}{pk})^{O(\frac{pk^2 \log p}{\eps})}$. Observe that if we select $T = O(\frac{p}{f}) = (\frac{pk}{\eps})^{O(\frac{pk^2 \log p}{\eps})}$, then with probability at least
$$1 - (1 - f)^T = 1 - (1 - f)^{O(p/f)} \geq 1 - e^{-O(p)}$$
there exists at least one $t \in [T]$ such that $\calG^m_t$ holds.

Let $\calG^m$ be the event that $\calG^m_t$ holds for at least one $t \in [T]$. Now, let $\calQ^m = \calE^m \cap \calF^m \cap \calG^m$. Observe that $\calQ^m$ holds with probability at least $(1 - O(1/p))^2 (1 - e^{-O(p)}) \geq (1 - O(1/p))$. If $\calQ^m$ holds, then $\calG^m$ holds, meaning that for some $t \in [T]$,
$$\|U^m_t W_m S_m^T T_m^T - M_m(A) S_m^T T_m^T\|_F \leq \Big(1 + O\Big(\frac{\eps}{p^2}\Big)\Big) \min_U \|U W_m S_m^T T_m^T - M_m(A) S_m^T T_m^T\|_F$$
In addition $\calF^m$ holds, meaning that
$$\|U^m_t W_m S_m^T - M_m(A) S_m^T\|_F \leq \Big(1 + O\Big(\frac{\eps}{p}\Big) \Big) \min_U \|U W_m S_m^T - M_m(A) S_m^T\|_F$$
Finally, $\calE^m$ holds, meaning that $S_m$ is a $(1 \pm O(\eps/p))$ $\ell_2$ affine embedding for $W_m$ and $M_m(A)$, and
$$\|U^m_t W_m - M_m(A)\|_F \leq \Big(1 + O\Big(\frac{\eps}{p}\Big)\Big) \min_U \|U W_m - M_m(A) \|_F$$
In summary, $\calQ^m$ holds with probability at least $1 - O(1/p)$ and implies that
$$\|U^m_t W_m - M_m(A)\|_F \leq \Big(1 + O\Big(\frac{\eps}{p}\Big)\Big) \min_U \|U W_m - M_m(A) \|_F$$
for some $t \in [T]$ --- in particular, by the definition of $\widehat{U^m}$, this implies that
\begin{equation} \label{eq:event_qm}
\begin{split}
\|\widehat{U^m} W_m - M_m(A)\|_F \leq \Big(1 + O\Big(\frac{\eps}{p}\Big)\Big) \min_U \|U W_m - M_m(A) \|_F
\end{split}
\end{equation}
Finally, observe that by reshaping and induction on $m$, we can show that if $\calQ^1, \ldots, \calQ^{m - 1}$ hold, then
$$\min_U \|UW_m - M_m(A)\|_F \leq \Big(1 + O\Big(\frac{\eps}{p}\Big)\Big)^{m - 1} \|(U^1_* \times \ldots \times U^p_* \times I_n \times \ldots \times I_n) \tensorvec(G_*) - \tensorvec(A)\|_F$$
To see why, note that we can iteratively replace $U^m_*$ with $\widehat{U^m}$ to lose a $(1 + \frac{\eps}{p})$ factor in the error (this follows from Equation \ref{eq:event_qm}, which holds when $\calQ^m$ holds for any $m \in [p]$) and replacing $G_{m - 1}$ with $G_m$ can only decrease the error. Thus, if $\calQ^1, \ldots \calQ^p$ all hold, then
\begin{equation} \label{eq:existence_of_solution}
\begin{split}
\|(\widehat{U^1} \times \ldots \times \widehat{U^p} & \times I_n \times \ldots \times I_n)  \tensorvec(G_p) - \tensorvec(A) \|_F \\
& \leq \Big(1 + O\Big(\frac{\eps}{p}\Big)\Big)^p \|(U^1_* \times \ldots \times U^p_* \times I_n \times \ldots \times I_n) \tensorvec(G_*) - \tensorvec(A)\|_F \\
& \leq (1 + O(\eps)) \|(U^1_* \times \ldots \times U^p_* \times I_n \times \ldots \times I_n) \tensorvec(G_*) - \tensorvec(A)\|_F
\end{split}
\end{equation}
where the second inequality is because $(1 + O(\eps/p))^p \leq e^{\frac{\eps}{p} \cdot p} = e^\eps$, and by the inequality $e^x \leq \frac{1}{1 - x}$ for $x > 0$. Let $\calQ$ be the event that Equation \ref{eq:existence_of_solution} holds --- then, $\calQ$ occurs with probability at least
$$\PrBig{\calQ^1 \cap \ldots \cap \calQ^p} \geq \prod_{i = 1}^p \PrBig{\calQ^i} \geq \Big(1 - O\Big(\frac{1}{p}\Big)\Big)^p \geq \frac{9}{10}$$
Here, the last inequality holds as long as $\calQ^m$ occurs with probability at least $1 - \frac{C}{p}$ for a sufficiently small constant $C$, and this can be achieved by setting the failure probabilities sufficiently small in $\calE^m$ and $\calF^m$, and multiplying the number of trials $T$ by a sufficiently large constant.

\paragraph{Part 2: Selecting the Best (or nearly best) Guess Using Kronecker Product Regression:} To identify $\widehat{U^1}, \ldots \widehat{U^p}$ defined above, we try all possible $p$-tuples of guesses. In other words, if for any mode $m \in [p]$, $\calS_m \subset \R^{n \times k}$ is the set of all guesses made by algorithm for $U^m$, then we try all $p$-tuples in $\calS_1 \times \ldots \times \calS_p$ (here the $\times$ symbol denotes the Cartesian product). Each $\calS_m$ has $T$ elements, meaning that there are overall $T^p$ guesses --- thus, we will instantiate the algorithm of \cite{djssw19_kronecker_product_regression} (with the parameters changed so that it can output a good estimate of the error, in addition to outputting a near-optimal solution --- the changes are described in Theorem \ref{thm:kronecker_product_regression}) with failure probability $\delta = \frac{1}{CT^p}$, where $C$ is an approprately large constant.

Note that if we directly apply the algorithm of \cite{djssw19_kronecker_product_regression} with $U^1 \times \ldots \times U^p \times \ldots I_n \times \ldots \times I_n$, the running time of this subroutine will be at least $n^{q - p}$, meaning the overall running time will be at least $n^{q - p} \cdot T^p$. Thus, we define the matrix $A_p$ and use a Countsketch matrix to obtain a PCP of $A_p$ as defined in Algorithm \ref{algorithm:fpt_rank_exactly_k_1_plus_eps_approximation}. Specifically, let $A_p \in \R^{n^p \times n^{q - p}}$ so that the entry of $A_p$ in the index $((i_1, i_2, \ldots, i_p), (i_{p + 1}, \ldots, i_q))$ is $A(i_1, i_2, \ldots, i_q)$, as in Algorithm \ref{algorithm:fpt_rank_exactly_k_1_plus_eps_approximation}. Then, for any tuple $(U^1, \ldots, U^p)$, note that the Kronecker product regression problem
$$\min_G \|(U^1 \times \ldots \times U^p \times I_n \times \ldots \times I_n) \tensorvec(G) - \tensorvec(A) \|_F$$
is equivalent to
$$\min_{X\in \R^{k^p\times n^{q-p}}} \|(U^1 \times \ldots \times U^p)X - A_p\|_F$$
Observe also that by Theorem \ref{thm:countsketch_pcps_zero_additive}, if $S_{\textsc{PROJ}}$ is a Countsketch matrix having $O(\frac{k^{2p}}{\eps^2})$ rows, then with probability at least $\frac{999}{1000}$, $S_{\textsc{PROJ}}$ preserves the error when projecting $A_p$ onto any $k^p$-dimensional subspace up to a $(1 \pm \eps)$ factor, i.e. if $P$ is any rank-$k^p$ projection matrix, then
$$\|A_p S_{\textsc{PROJ}}^T - PA_p S_{\textsc{PROJ}}^T\|_F = (1 \pm \eps) \|A_p - PA_p\|_F$$
Note that regression problem $\min_X \|(U^1 \times \ldots \times U^p) X - A_p\|_F$ is equivalent to computing the error when projecting $A_p$ onto the column span of the $k^p$-dimensional subspace spanned by $U^1 \times \ldots \times U^p$, and thus
$$\min_X \|(U^1 \times \ldots \times U^p)X - A_p S_{\textsc{PROJ}}^T\|_F = (1 \pm \eps) \min_X \|(U^1 \times \ldots \times U^p)X - A_p \|_F$$
For convenience, let $M = A_p S_{\textsc{PROJ}}^T$. Observe that by Theorem \ref{thm:kronecker_product_regression}, for each $i \in [O(k^{2p}/\eps^2)]$, it is possible to obtain an estimate $\widehat{e_i} = (1 \pm O(\eps)) \min_{X_i} \|(U^1 \times \ldots \times U^p)X_i - M_i\|_F^2 - \|M_i\|_F^2$ in $O(nkp) + (k^p/(\eps\delta))^{O(1)} = O(nkp) + (\frac{pk}{\eps})^{O(\frac{p^2k^2 \log p}{\eps})}$ running time (and failure probability $\delta = \frac{1}{CrT^p}$ where $r$ is the number of rows in the Countsketch matrix $S_{\textsc{PROJ}}$). Since $M$ has $O(\frac{k^{2p}}{\eps^2})$ columns, it is therefore possible to obtain (with probability $1 - \frac{1}{CT^p}$) an estimate
$$\widehat{e} = (1 \pm O(\eps)) \min_X \|(U^1 \times \ldots \times U^p) X - M\|_F^2 - \|M\|_F^2$$
in running time $n \cdot (k^p/\eps)^{O(1)} + (\frac{pk}{\eps})^{O(\frac{p^2k^2 \log p}{\eps})}$. Now recall that in Algorithm \ref{algorithm:fpt_rank_exactly_k_1_plus_eps_approximation}, we separately pre-compute $\|M\|_F^2 = \|A_p S_{\textsc{PROJ}}\|_F^2$ before evaluating the errors from each of the tuples $(U^1, \ldots, U^p)$ --- thus, we can use $\widehat{e} + \|M\|_F^2$ as our estimate of the error of $(U^1, \ldots, U^p)$. Observe that with probability $1 - \frac{1}{CT^p}$,
\begin{equation}
\begin{split}
\widehat{e} + \|M\|_F^2
& = (1 \pm O(\eps)) \min_X \|(U^1 \times \ldots \times U^p) X - A_p S_{\textsc{PROJ}} \|_F^2 \\
& = (1 \pm O(\eps)) \min_X \|(U^1 \times \ldots \times U^p) X - A_p\|_F^2 \\
& = (1 \pm O(\eps)) \min_G \|(U^1 \times \ldots \times U^p \times I_n \times \ldots \times I_n) \tensorvec(G) - \tensorvec(A) \|_F^2
\end{split}
\end{equation}
where the second equality is by the consequence of Theorem \ref{thm:countsketch_pcps_zero_additive}, since $S_{\textsc{PROJ}}$ preserves the error when projecting $A_p$ onto the subspace $U^1 \times \ldots \times U^p$. Thus, with probability $1 - \frac{1}{C}$, simultaneously for all tuples $(U^1, \ldots, U^p) \in \calS_1 \times \ldots \times S_p$, if $\widehat{e}$ is the error estimate for $(U^1, \ldots, U^p)$ defined above, then
$$\widehat{e} + \|M\|_F^2 = (1 \pm O(\eps)) \min_G \|(U^1 \times \ldots \times U^p \times I_n \times \ldots \times I_n) \tensorvec(G) - \tensorvec(A) \|_F^2$$
In particular, if $(\widetilde{U^1}, \ldots, \widetilde{U^p})$ is the $p$-tuple which achieves the smallest value of $\widehat{e} + \|M\|_F^2$ over the course of the algorithm, then
\begin{equation}
\begin{split}
\min_G \|(\widetilde{U^1} \times \ldots \times \widetilde{U^p} & \times I_n \times \ldots \times I_n)  \tensorvec(G) - \tensorvec(A) \|_F^2 \\
& \leq (1 + O(\eps)) \min_G \|\widehat{U^1} \times \ldots \times \widehat{U^p} \times I_n \times \ldots \times I_n) \tensorvec(G) - \tensorvec(A) \|_F^2 \\
& \leq (1 + O(\eps)) \|(U^1_* \times \ldots \times U^p_* \times I_n \times \ldots \times I_n) \tensorvec(G_*) - \tensorvec(A) \|_F^2
\end{split}
\end{equation}
Note that the second inequality is by Equation \ref{eq:existence_of_solution} (note that it holds with probability at least $\frac{9}{10}$ as mentioned in the discussion below Equation \ref{eq:existence_of_solution}). In addition, the first inequality holds with probability at least $1 - \frac{1}{C}$, where $C$ can be chosen to be a large constant.

\paragraph{Part 3: Running Time Analysis:} First let us analyze the running time of generating the guesses for $U^m$ and forming the sets $\calS_m$ for $m \in [p]$. Let us first calculate the time needed to form $\calS_m$ --- then, this running time will be multiplied by $p$. First, we multiply $M_m(A) \in \R^{n \times n^{q - 1}}$ on the right by a Countsketch matrix $S_m^T$, which takes $O(\nnz(A))$ time. We then multiply this matrix on the right by a sign matrix $T_m$, and this takes $O(n \cdot \frac{p k \log p}{\eps})$. Next, for $T$ trials, we create a Gaussian matrix $R_{m, t} \in \R^{k \times O(\frac{p k \log p}{\eps}))}$, and then multiply $M_m(A) S_m^T T_m^T \in \R^{n \times O(\frac{p k \log p}{\eps})}$ by $R_{m, t}^T$ --- this operation takes $O(\frac{npk^2 \log p}{\eps})$ time, and over the course of $T$ trials takes $O(\frac{npk^2 \log p}{\eps}) \cdot (\frac{pk}{\eps})^{O(\frac{pk^2 \log p}{\eps})} = n \cdot (\frac{pk}{\eps})^{O(\frac{pk^2 \log p}{\eps})}$ time. Thus, the running time needed to form $\calS_m$ for a single $m \in [p]$ is $O(\nnz(A)) + n \cdot (\frac{pk}{\eps})^{O(\frac{pk^2 \log p}{\eps})}$. The running time needed to form $\calS_m$ for all $m \in [p]$ is thus $O(p \cdot \nnz(A)) + n \cdot (\frac{pk}{\eps})^{O(\frac{pk^2 \log p}{\eps})}$.

Now, we analyze the running time needed to evaluate the error of all the tuples $(U^1, \ldots, U^p)$ in $\calS_1 \times \ldots \calS_p$. The time needed to compute $A_p S_{\textsc{PROJ}}^T$ and initialize $\textsc{MinError} \gets \|A\|_F^2$ is $O(\nnz(A))$. We must also pre-compute $\|M\|_F^2$, which can be done in $\nnz(M)$ time, which is at most $\nnz(A)$ time since there are at most $\nnz(A)$ nonzero entries of $M$. Now, the number of tuples $(U^1, \ldots, U^p) \in \calS_1 \times \ldots \times \calS_p$ is $T^p = (\frac{pk}{\eps})^{O(\frac{p^2 k^2 \log p}{\eps})}$. For each tuple, and for each of the $O(\frac{k^{2p}}{\eps^2})$ columns $M_j$ of $M$, we run $\textsc{ErrorEstimate}(U^1, \ldots, U^p, M_j, \delta = \frac{1}{1000r T^p})$ (the Kronecker product regression algorithm of \cite{djssw19_kronecker_product_regression}), and since the dimension $d$ of the regression problem is $k^p$, the running time of this step is $\sum_{m = 1}^p \nnz(U^m) + (\frac{k^p}{\eps \delta})^{O(1)} = O(nkp) + (\frac{pk}{\eps})^{O(\frac{p^2 k^2 \log p}{\eps})}$. This running time is incurred $T^p \cdot O(\frac{k^{2p}}{\eps^2}) = (\frac{pk}{\eps})^{O(\frac{p^2k^2 \log p}{\eps})}$ times, meaning the overall running time needed to find the best (up to a $1 \pm \eps$ factor) tuple $(U^1, \ldots, U^p)$ is $n \cdot (\frac{pk}{\eps})^{O(\frac{p^2 k^2 \log p}{\eps})}$. Combining the running times from the first and the second steps, we find that the overall running time is $O(p \cdot \nnz(A)) + n \cdot (\frac{pk}{\eps})^{O(\frac{p^2 k^2 \log p}{\eps})}$. This completes the proof.
\end{proof}

\begin{remark}
Note that we can also output a $(1 + \eps)$-approximate core tensor $G$ using the algorithm of \cite{djssw19_kronecker_product_regression}. In the case $p = q$ this will only lead to an additional term of $O(qnk) + k^{O(q)}/\eps^{O(1)}$ in the running time. In general the additional running time would be $n^{q - p} \cdot (O(pnk) + k^{O(p)}/\eps^{O(1)})$.
\end{remark}

\newpage
\subsection{FPT Tensor Train Decomposition}

In this section, we 
obtain a fixed-parameter tractable algorithm for Tensor Train decomposition with output rank exactly $k$. Our algorithm is shown in Algorithm \ref{algorithm:fpt_rank_exactly_k_tensor_train}, and we give guarantees for this algorithm in Theorem \ref{thm:fpt_tensor_train_rank_exactly_k}.

\begin{algorithm}
\caption{$(1 + \eps)$-approximation algorithm for Tensor Train decomposition with output rank exactly $k$ and fixed-parameter tractable running time. This algorithm is similar to our bicriteria algorithm for tree networks (specialized to the case of Tensor Train), but uses an exponential number of trials for each new mode (with Gaussian matrices having $k$ columns) to ensure that the output rank is exactly $k$. \textsc{ApproxBestCandidate} efficiently finds the candidate in $L_q$ that (approximately) minimizes the Frobenius norm distance to $A$, and is shown in Algorithm \ref{algorithm:tensor_train_fpt_find_best_candidate}.}
\label{algorithm:fpt_rank_exactly_k_tensor_train}
\begin{algorithmic}
\Require $A \in \R^{n \times \ldots \times n}$ with $q$ modes, $\eps, \delta \in (0, 1)$, $k \in \N$
\Ensure $U_1 \in \R^{n \times k}$, $U_2, \ldots, U_{q - 1} \in \R^{k \times n \times k}$, $U_q \in \R^{k \times n}$
\State $P_1 \gets $ An $r \times n^{q - 1}$ Countsketch matrix, where $r \in O(\frac{q^3 k^2}{\eps^2 \delta})$.
\State $L_1 \gets \{\}$   // This will contain the candidates for $U^1$.
\State $T \gets \log\Big(\frac{q}{\delta}\Big) \cdot \Big(e^{\Theta(k^2 \log k) + \Theta(qk^2/\eps \cdot \log(qk/\eps))}\Big)$
\For{$t = 1, \ldots, T$}
    \State $R_{1, t} \gets $ A $m \times r$ matrix whose entries are drawn i.i.d. from $\{-\frac{1}{\sqrt{m}}, \frac{1}{\sqrt{m}}\}$, where $m \in O(\frac{qk}{\eps})$
    \State $R_{1, t} \gets $ An $k \times m$ matrix whose entries are drawn i.i.d. from $N(0, \frac{1}{k})$
    \State $\widetilde{U^1_t} \gets M_1(A) P_1^T Q_1^T R_1^T$
    \State $L_1 \gets L_1 \cup \{\widetilde{U^1_t}\}$
\EndFor
\State $S_1 \gets $ An $s_1 \times n$ Countsketch matrix, where $s_1 = O(\frac{q^4 k^2}{\eps^2 \delta})$
\For{$i=2, \ldots, q - 1$}
    \State $P_i \gets $ An $n^{q - i} \times r$ Countsketch matrix, where $r = O(\frac{q^3k^2}{\eps^2 \delta})$
    \State $L_i \gets \{\}$    // This will contain the candidates for $(U^1, \ldots, U^i)$.
    
    \item []
    \State // Iterate through all the candidates in $L_{i - 1}$, and for each one, add $T$ candidates for $U^i$.
    \For{$(U^1, \ldots, U^{i - 1}) \in L_{i - 1}$}
        \For{$t = 1, \ldots, T$}
            \State $Q_{i, t} \gets $ A $m \times r$ matrix whose entries are drawn i.i.d. from $\{-\frac{1}{\sqrt{m}}, \frac{1}{\sqrt{m}}\}$, where 
            \State \,\,\,\,\,\,\,\,\,\,\,\,\,\,\,\,\, $m = O(k^2q/\eps)$
            \State // Here $\widetilde{U^i} \in \R^{n \times k^2}$ is the solution to the regression problem $$\min_{U} \|U (M_i(U^1 \circ \cdots \circ U^{i - 1}) \times M_1(U^{i + 1} \circ \cdots \circ U^q))(\calL_{i - 1}^T \times P_i^T Q_i^T R_i^T) - M_i(A) (\calL_{i - 1}^T \times P_i^T Q_i^T R_i^T)\|_F$$
            \State // and $\calL_{i - 1}\in \R^{s_{i-1}\times n^{i-1}}$ is as in Definition \ref{def:tensor_train_sketch_with_placeholders}, specified by the Countsketch matrices 
            \State // $S_1, \ldots, S_{i - 1}$.
        
            \State $R_{i, t} \gets $ An $k \times m$ matrix whose entries are i.i.d. $N(0, \frac{1}{k})$
            \State $\widetilde{U^i} \gets M_i(A) (\calL_{i - 1}^T \times P_i^T Q_{i, t}^T R_{i, t}^T)((M_i(U^1 \circ \cdots \circ U^{i - 1})\calL_{i - 1}^T)^\dagger \times I_k)$
            \State $L_i \gets L_i \cup \{(U^1, \ldots, U^{i - 1}, \widetilde{U^i})\}$
        \EndFor
    \EndFor
    
    \State $S_i \gets $ An $s_i \times s_{i - 1}n$ Countsketch matrix, where $s_i = O(\frac{q^4 k^2}{\eps^2 \delta})$
\EndFor
\algstore{TT_exact_k}
\end{algorithmic}
\end{algorithm}

\begin{algorithm}
\textbf{Algorithm \ref{algorithm:fpt_rank_exactly_k_tensor_train}} $(1+\eps)-$approximation algorithm for Tensor Train decomposition with output rank exactly $k$ continued.
\begin{algorithmic}
\algrestore{TT_exact_k}
\item []
\State // Now we find $U^q$ for each of the candidates for $(U^1, \ldots, U^{q - 1})$.
\State $L_q \gets \{\}$
\For{$(U^1, \ldots, U^{q - 1}) \in L_{q - 1}$}
    \State // Here $\widetilde{U^q} \in \R^{n \times k}$ is the solution to the regression problem $$\min_U \|U M_q(U^1 \circ \cdots \circ U^{q - 1})\calL_{q - 1}^T - M_q(A) \calL_{q - 1}^T\|_F$$
    \State // which can be solved exactly, and $\widetilde{U^q}$ has $k$ columns since the last mode of $U^{q - 1}$ has dimension $k$.
    \State $\widetilde{U^q} \gets M_q(A) \calL_{q - 1}^T(M_q(U^1\circ \cdots \circ U^{q-1})\calL_{q-1}^T)^\dagger$
    \State $L_q \gets L_q \cup \{(U^1, \ldots, U^{q - 1}, \widetilde{U^q})\}$
\EndFor

\State $(U^1, \ldots, U^q) \gets \textsc{ApproxBestCandidate}(L_q, A)$ \\
\Return $(U^1, \ldots, U^q)$
\end{algorithmic}
\end{algorithm}

\begin{algorithm}
\caption{This is a subroutine in Algorithm \ref{algorithm:fpt_rank_exactly_k_tensor_train}, which we call once we have found a list $L_q$ of $2^{\poly(qk/\eps)}$ candidates $(U^1, \ldots, U^q)$. We find the candidate $(U^1, \ldots, U^q)$ which approximately minimizes $\|U^1 \circ \cdots \circ U^q - A\|_F$, up to a $(1 + \eps)$ factor. This is achieved by first compressing $A$ with Countsketch matrices with $2^{\poly(qk/\eps)}$ rows, which contributes $O(q \cdot \nnz(A))$ to the running time of this subroutine, and then applying those same Countsketch matrices to each of the candidates $(U^1, \ldots, U^q)$ which contributes $|L_q| \cdot 2^{\poly(qk/\eps)} \cdot n$ to the overall running time. The number of rows in the Countsketch matrices is selected so that they approximately preserve $\|U^1 \circ \cdots \circ U^q - A\|_F$ for all candidates $(U^1, \ldots, U^q)$. \AM{Give a name to this subroutine.}}
\label{algorithm:tensor_train_fpt_find_best_candidate}
\begin{algorithmic}
\Require $L_q$ contains elements of the form $(U^1, \ldots, U^q)$ where $U^1 \in \R^{n \times k}, U^2, \ldots, U^{q - 1} \in \R^{k \times n \times k}, U^q \in \R^{k \times n}$, $A \in \R^{n \times \ldots \times n}$ has $q$ modes
\Ensure Returns a candidate $(U^1, \ldots, U^q) \in L_q$

\item[]
\State // First, we compress $A$ using Countsketch matrices.
\State $\widetilde{A} \gets A \in \R^{n \times \ldots \times n}$
\For{$i = 1, \ldots, q$}
    \State $S_i \gets $ An $s_1 \times n$ Countsketch matrix if $i = 1$, and an $s_i \times s_{i - 1}n$ Countsketch matrix otherwise \State \,\,\,\,\,\,\,\,\,\,\,\,\,\,\,\, where $s_i = \poly(qk2^{\poly(qk/\eps)}/\eps)$
    \State $\widetilde{A} \gets S_i \circ \calF(\widetilde{A})$, where $\calF$ is defined as in Definition \ref{def:fusing_first_two_modes}.
\EndFor

\item[]
\State // Next, we evaluate all the candidates. Here, $\calL_q(S_1, \ldots, S_q)$ is as in Definition \ref{def:tensor_train_sketch_with_placeholders}.
\State $(U^1_{best}, \ldots, U^q_{best}) \gets (0, \ldots, 0)$
\State $\textsc{MinCost} \gets \infty$
\For{$(U^1, \ldots, U^q) \in L_q$}
    \State $\textsc{CurrentCost} \gets \|\calL_q(S_1, \ldots, S_q)(U^1 \circ \cdots \circ U^q) - \widetilde{A}\|_F$
    \If{$\textsc{CurrentCost} < \textsc{MinCost}$}
        \State $\textsc{MinCost} \gets \textsc{CurrentCost}$
        \State $(U^1_{best}, \ldots, U^q_{best}) \gets (U^1, \ldots, U^q)$
    \EndIf
\EndFor

\State \Return $(U^1_{best}, \ldots, U^q_{best})$
\end{algorithmic}
\end{algorithm}

\begin{proof}[Proof of Theorem \ref{thm:fpt_tensor_train_rank_exactly_k}]
Define
\begin{equation}
\begin{split}
U^1_*, \ldots, U^q_* := \argmin_{U^1, \ldots, U^q} \|U^1 \circ \cdots \circ U^q - A\|_F
\end{split}
\end{equation}
\paragraph{Base Case: Finding Candidates for $U^1$}
Let $\calE_1$ denote the event that $L_1$ contains a $\widetilde{U^1} \in \R^{n \times k}$ such that
\begin{equation}
\begin{split}
\min_{U^2, \ldots, U^q} \|\widetilde{U^1} \circ U^2 \circ \cdots \circ U^q - A\|_F
& \leq \Big(1 + O\Big(\frac{\eps}{q}\Big)\Big) \|U^1_* \circ \cdots \circ U^q_* - A\|_F
\end{split}
\end{equation}
We will show that $\calE_1$ holds with probability at least $1 - O(\frac{\delta}{q})$. Note that we can write
\begin{equation}
\begin{split}
U^1_* = \argmin_U \|U \circ U^2_* \circ \cdots \circ U^q_* - A\|_F = \argmin_U \|U M_1(U^2_* \circ \cdots \circ U^q_*) - M_1(A)\|_F
\end{split}
\end{equation}
Let us define the following events:
\begin{itemize}
\item Let $\calE_{1, a}$ denote the event that $P_1$ is a $\Big(1 \pm O\Big(\frac{\eps}{q}\Big)\Big)$ affine embedding for the row span of $M_1(U^2_* \circ \cdots \circ U^q_*)$ and $M_1(A)$, that is, for all $U \in \R^{n \times k}$,
\begin{equation} \label{eq:e1a_base}
\begin{split}
\|U M_1(U^2_* \circ \cdots \circ U^q_*) P_1^T - M_1(A) P_1^T\|_F
& = \Big(1 \pm O\Big(\frac{\eps}{q}\Big)\Big) \|UM_1(U^2_* \circ \cdots \circ U^q_*) - M_1(A) \|_F
\end{split}
\end{equation}
By Theorem \ref{thm:countsketch_affine_embedding}, if we let the number of rows of $P_1$ be $r = O(\frac{q^3k^2}{\eps^2 \delta})$, then $\calE_{1, a}$ holds with probability $1 - O(\frac{\delta}{q})$.
\item For $t \in [T]$, let $\calE_{1, b, t}$ denote the event that for all $U$ satisfying
\begin{equation} \label{eq:e1b_base_hypothesis}
\begin{split}
\|U M_1(U^2_* \circ \cdots & \circ U^q_*) P_1^T Q_{1, t}^T -  M_1(A) P_1^T Q_{1, t}^T\|_F \\
& \leq \Big(1 + O\Big(\frac{\eps}{q}\Big)\Big) \min_U \|UM_1(U^2_* \circ \cdots \circ U^q_*) P_1^T Q_{1, t}^T - M_1(A) P_1^T Q_{1, t}^T\|_F
\end{split}
\end{equation}
it also holds that
\begin{equation} \label{eq:e1b_base_conclusion}
\begin{split}
\|U M_1(U^2_* \circ \cdots \circ U^q_*) P_1^T &- M_1(A) P_1^T\|_F \\ &\leq \Big(1 + O\Big(\frac{\eps}{q}\Big) \Big) \min_U \|UM_1(U^2_* \circ \cdots \circ U^q_*) P_1^T - M_1(A) P_1^T\|_F
\end{split}
\end{equation}
Applying Theorem \ref{thm:approximate_multiple_l2_regression} with $\delta = \frac{1}{3}$ and $\eps$ replaced by $O(\frac{\eps}{q})$, we find that $\calE_{1, b, t}$ holds with probability at least $\frac{2}{3}$ as long as $Q_1$ has $m = O(\frac{qk}{\eps})$ rows.

\item Finally, for $t \in [T]$, let $\calE_{1, c, t}$ denote the event that, if we define
\begin{equation} \label{eq:e1c_base_definition}
\begin{split}
\widehat{U^1_t} 
& = \argmin_{U} \|U M_1(U^2_* \circ \cdots \circ U^q_*) P_1^T Q_{1, t}^T R_{1, t}^T - M_1(A) P_1^T Q_{1, t}^T R_{1, t}^T\|_F \\
& = M_1(A) P_1^T Q_{1, t}^T R_{1, t}^T (M_1(U^2_* \circ \cdots \circ U^q_*) P_1^T Q_{1, t}^T R_{1, t}^T)^\dagger \\
& = \widetilde{U^1_t} X
\end{split}
\end{equation}
where $X = (M_1(U^2_* \circ \cdots \circ U^q_*) P_1^T Q_{1, t}^T R_{1, t}^T)^\dagger \in \R^{k \times k}$, then
\begin{equation} \label{eq:e1c_base_conclusion}
\begin{split}
\|\widehat{U^1_t} M_1(U^2_* \circ & \cdots \circ U^q_*) P_1^T Q_{1, t}^T  - M_1(A) P_1^T Q_{1, t}^T\|_F \\
& \leq \Big(1 + O\Big(\frac{\eps}{q}\Big)\Big) \min_U \|U M_1(U^2_* \circ \cdots \circ U^q_*) P_1^T Q_{1, t}^T - M_1(A) P_1^T Q_{1, t}^T\|_F
\end{split}
\end{equation}
By Lemma \ref{lemma:gaussian_with_exactly_k_rows_multiple_response_regression}, and since $M_1(A) P_1^T Q_1^T$ has rank at most $m$, this holds with probability at least $e^{-\Theta(k^2 \log k)} (\frac{\eps}{qk})^{O(km)} = e^{-\Theta(k^2 \log k)} (\frac{\eps}{qk})^{O(qk^2/\eps)}$.
\end{itemize}
Thus, for a fixed $t \in [T]$, $\calE_{1, b, t} \cap \calE_{1, c, t}$ holds with probability at least $e^{-\Theta(k^2 \log k)} (\frac{\eps}{qk})^{O(qk^2/\eps)}$, and the probability that for all $t \in [T]$, $\calE_{1, b, t} \cap \calE_{1, c, t}$ fail is therefore at most
\begin{equation}
\begin{split}
\Big(1 - e^{-\Theta(k^2 \log k)} \Big(\frac{\eps}{qk}\Big)^{O(qk^2/\eps)} \Big)^T
& = \Big(1 - \frac{1}{e^{\Theta(k^2 \log k) + \Theta(qk^2/\eps \cdot \log(qk/\eps))}}\Big)^T \\
& \leq e^{-\frac{T}{e^{\Theta(k^2 \log k) + \Theta(qk^2/\eps \cdot \log(qk/\eps))}}}
\end{split}
\end{equation}
Here in the first inequality we used the fact that $1 - x \leq e^{-x}$ for all $x \in \R$. Therefore, as long as
\begin{equation}
\begin{split}
T \geq \log\Big(\frac{q}{\delta}\Big) \cdot \Big(e^{\Theta(k^2 \log k) + \Theta(qk^2/\eps \cdot \log(qk/\eps))}\Big)
\end{split}
\end{equation}
with probability at least $1 - O(\frac{\delta}{q})$, there exists some $t \in [T]$ such that $\calE_{1, b, t} \cap \calE_{1, c, t}$ holds.

Note that if $\calE_{1, a}$ holds, and $t \in [T]$ is such that $\calE_{1, b, t}$ and $\calE_{1, c, t}$ both hold, then $\calE_1$ holds, since if $t \in [T]$ is such that $\calE_{1, b, t}$ and $\calE_{1, c, t}$ hold, and $\calE_{1, a}$ holds, then
\begin{equation}
\begin{split}
\min_{U^2, \ldots, U^q} \|\widetilde{U^1_t} \circ U^2 \circ \cdots \circ U^q - A\|_F
& \leq \|(\widetilde{U^1_t} X) \circ U^2_* \circ \cdots \circ U^q_* - A\|_F \\
& = \|\widetilde{U^1_t}X M_1(U^2_* \circ \cdots \circ U^q_*) - M_1(A) \|_F \\
& \leq \Big(1 + O\Big(\frac{\eps}{q}\Big)\Big) \|\widehat{U^1_t} M_1(U^2_* \circ \cdots \circ U^q_*)P_1^T - M_1(A) P_1^T\|_F \\
& \leq \Big(1 + O\Big(\frac{\eps}{q}\Big)\Big) \min_U \|U M_1(U^2_* \circ \cdots \circ U^q_*) P_1^T - M_1(A) P_1^T\|_F \\
& \leq \Big(1 + O\Big(\frac{\eps}{q}\Big)\Big) \min_U \|U M_1(U^2_* \circ \cdots \circ U^q_*) - M_1(A)\|_F \\
& = \Big(1 + O\Big(\frac{\eps}{q}\Big)\Big) \min_U \|U \circ U^2_* \circ \cdots \circ U^q_* - A\|_F
\end{split}
\end{equation}
Here the first inequality is because $X \in \R^{k \times k}$. The second inequality is by Equation \ref{eq:e1a_base}, since $\calE_{1, a}$ holds. The third inequality is by Equation \ref{eq:e1b_base_conclusion}, since $\calE_{1, b, t}$ holds (here note that Equation \ref{eq:e1b_base_hypothesis} holds, since Equation \ref{eq:e1c_base_conclusion} holds by our assumption that $\calE_{1, c, t}$ holds). Finally, the fourth inequality is by Equation \ref{eq:e1a_base} (since $\calE_{1, a}$ holds).

In summary, $\calE_1$ holds assuming $\calE_{1, a}$ holds and $\calE_{1, b, t}$ and $\calE_{1, c, t}$ hold simultaneously for some $t \in [T]$. As shown above, this holds with probability at least $1 - O(\frac{\delta}{q})$, as desired.

\paragraph{Inductive Step: Finding Candidates for $U^i$ for $i \in [q] \setminus \{1, q\}$}

In the following, let $C$ be an absolute constant that we choose appropriately. For $i \in [q]$, let $\calE_i$ be the event that there exists some $(\widetilde{U^1}, \ldots, \widetilde{U^i}) \in L_i$ such that
\begin{equation}
\begin{split}
\min_{U^{i + 1}, \ldots, U^q} \|\widetilde{U^1} \circ \cdots \circ \widetilde{U^i} \circ U^{i + 1} \circ \cdots \circ U^q - A\|_F
& \leq \Big(1 + \frac{C\eps}{q}\Big)^i \min_{U^1, \ldots, U^q} \|U^1 \circ \cdots \circ U^q - A\|_F
\end{split}
\end{equation}
We wish to show by induction that $\calE_i$ holds with probability at least $1 - \frac{Ci\delta}{q}$. Thus, assume $\calE_{i - 1}$ holds with probability at least $1 - \frac{C(i - 1) \delta}{q}$, that is, there exists $(\widetilde{U^1}, \ldots, \widetilde{U^{i - 1}}) \in L_{i - 1}$ such that
\begin{equation}
\begin{split}
\min_{U^i, \ldots, U^q} \|\widetilde{U^1} \circ \cdots \circ \widetilde{U^{i - 1}} \circ U^i \circ \cdots \circ U^q - A\|_F
& \leq \Big(1 + \frac{C\eps}{q}\Big)^{i - 1} \min_{U^1, \ldots, U^q} \|U^1 \circ \cdots \circ U^q - A\|_F
\end{split}
\end{equation}
For convenience, in the proof of the inductive step we will define
\begin{equation}
\begin{split}
U^i_*, \ldots, U^q_* := \argmin_{U^i, \ldots, U^q} \|\widetilde{U^1} \circ \cdots \circ \widetilde{U^{i - 1}} \circ U^i \circ \cdots \circ U^q - A\|_F
\end{split}
\end{equation}
Note that we are redefining $U^i_*, \ldots, U^q_*$, compared to our definition in the base case. From this definition, it follows that
\begin{equation}
\begin{split}
U^i_*
& = \argmin_{U} \|U (M_i(\widetilde{U^1} \circ \cdots \circ \widetilde{U^{i - 1}}) \times M_1(U^{i + 1} \circ \cdots \circ U^q)) - M_i(A) \|_F
\end{split}
\end{equation}

We now define the following auxiliary events:
\begin{itemize}
\item Let $\calE_{i, a}$ denote the event that for all $U \in \R^{n \times k^2}$,
\begin{equation}
\begin{split}
\|U (M_i(\widetilde{U^1} \circ & \cdots \circ \widetilde{U^{i - 1}}) \times M_1(U^{i + 1}_* \circ \cdots \circ U^q_*))(I \times P_i^T) - M_i(A) (I \times P_i^T)\|_F \\
& = \Big(1 + O\Big(\frac{\eps}{q}\Big)\Big) \|U (M_i(\widetilde{U^1} \circ \cdots \circ \widetilde{U^{i - 1}}) \times M_1(U^{i + 1}_* \circ \cdots \circ U^q_*)) - M_i(A)\|_F
\end{split}
\end{equation}
Note that the above is equivalent to
\begin{equation}
\begin{split}
\|M_{i + 1}(\widetilde{U^1} & \circ \cdots \circ \widetilde{U^{i - 1}} \circ U)^T M_1(U^{i + 1}_* \circ \cdots \circ U^q_*) P_i^T - M_{\{1, \ldots, i\}}(A) P_i^T\|_F \\
& = \Big(1 + O\Big(\frac{\eps}{q}\Big)\Big) \|M_{i + 1}(\widetilde{U^1} \circ \cdots \circ \widetilde{U^{i - 1}} \circ U)^T M_1(U^{i + 1}_* \circ \cdots \circ U^q_*) - M_{\{1, \ldots, i\}}(A) \|_F
\end{split}
\end{equation}
and thus it suffices for $P_i$ to be a $\Big(1 \pm O\Big(\frac{\eps}{q}\Big)\Big)$ affine embedding for the row span of $M_1(U^{i + 1}_* \circ \cdots \circ U^q_*)$ and $M_{\{1, \ldots, i\}}(A)$. By Theorem \ref{thm:countsketch_affine_embedding}, this holds with probability $1 - O(\frac{\delta}{q})$ as long as $P_i$ has $O(\frac{q^3k^2}{\eps^2 \delta})$ rows.

\item Let $\calE_{i, b}$ denote the event that for all $U \in \R^{n \times k^2}$,
\begin{equation}
\begin{split}
& \|U (M_i(\widetilde{U^1} \circ \cdots  \circ \widetilde{U^{i - 1}})\calL_{i - 1}^T \times M_1(U^{i + 1}_* \circ \cdots \circ U^q_*)P_i^T) - M_i(A)(\calL_{i - 1}^T \times P_i^T)\|_F \\
& = \Big(1 + O\Big(\frac{\eps}{q}\Big)\Big) \|U (M_i(\widetilde{U^1} \circ \cdots \circ \widetilde{U^{i - 1}}) \times M_1(U^{i + 1}_* \circ \cdots \circ U^q_*)P_i^T) - M_i(A)(I \times P_i^T)\|_F
\end{split}
\end{equation}
We can rewrite this as follows:
\begin{equation}
\begin{split}
\|\calL_{i - 1} & M_{\{1, \ldots, i - 1\}}(\widetilde{U^1} \circ \cdots \circ \widetilde{U^{i - 1}}) M_1(U \circ U^{i + 1}_* \circ \cdots \circ U^q_*)(I \times P_i^T) \\
&\,\,\,\,\,\,\,\,\,\,\,\,\,\,\,\,\,\,\,\,\,\,\,\,\,\,\,\,\,\,\,\,\,\,\,\,\,\,\,\,\,\,\,\,\,\,\,\,\,\,\,\,\,\,\,\,\,\,\,\,\,\,\,\,\,\,\,\,\,\,\,\,\,\,\,\,\,\,\,\,\,\,\,\,\,\,\,\,\,\,\,\,\,\,\,\,\,\,\,\,\,\,\,\,\,\,\,\,\,\,\,\,\,\,\,\,\,\,\,\,\,\,\,\,\,\, - \calL_{i - 1} M_{\{1, \ldots, i - 1\}}(A)(I \times P_i^T) \|_F \\
& = \Big(1 \pm O\Big(\frac{\eps}{q}\Big)\Big) \| M_{\{1, \ldots, i - 1\}}(\widetilde{U^1} \circ \cdots \circ \widetilde{U^{i - 1}}) M_1(U \circ U^{i + 1}_* \circ \cdots \circ U^q_*)(I \times P_i^T) \\
&\,\,\,\,\,\,\,\,\,\,\,\,\,\,\,\,\,\,\,\,\,\,\,\,\,\,\,\,\,\,\,\,\,\,\,\,\,\,\,\,\,\,\,\,\,\,\,\,\,\,\,\,\,\,\,\,\,\,\,\,\,\,\,\,\,\,\,\,\,\,\,\,\,\,\,\,\,\,\,\,\,\,\,\,\,\,\,\,\,\,\,\,\,\,\,\,\,\,\,\,\,\,\,\,\,\,\,\,\,\,\,\,\,\,\,\,\,\,\,\,\,\,\,\,\,\, - M_{\{1, \ldots, i - 1\}}(A)(I \times P_i^T) \|_F
\end{split}
\end{equation}
where on the left-hand side, $I$ now denotes the $n \times n$ identity matrix. Thus, it suffices for $\calL_{i - 1}$ to be a $\Big(1 \pm O\Big(\frac{\eps}{q}\Big)\Big)$ affine embedding for the column span of $M_{\{1, \ldots, i - 1\}}(\widetilde{U^1} \circ \cdots \circ \widetilde{U^{i - 1}})$ and $M_{\{1, \ldots, i - 1\}}(A)(I \times P_i^T)$. By Lemma \ref{lemma:jl_moment_implies_affine_embedding}, this holds as long as $\calL_{i - 1}$ has the $(O(\frac{\eps}{qk}), O(\frac{\delta}{q}), 2)$ JL moment property, and by Lemma \ref{lemma:jl_moment_of_tensor_train_sketch_based_on_S_i}, it is sufficient for $S_j$ to have the $(O(\frac{\eps}{q^{3/2}k}), O(\frac{\delta}{q}), 2)$ JL moment property for $1 \leq j \leq i - 1$ --- this holds as long as $S_j$ has $O(\frac{q^4 k^2}{\eps^2 \delta})$ rows. 

\item Now, for $t \in [T]$, let $\calE_{i, c, t}$ denote the event that, for all $U \in \R^{n \times k^2}$ satisfying
\begin{equation}
\begin{split}
\|U & (M_i(\widetilde{U^1} \circ \cdots \circ \widetilde{U^{i - 1}})\calL_{i - 1}^T \times M_1(U^{i + 1}_* \circ \cdots \circ U^q_*)P_i^T)(I \times Q_{i, t}) \\
& \,\,\,\,\,\,\,\,\,\,\,\,\,\,\,\,\,\,\,\,\,\,\,\,\,\,\,\,\,\,\,\,\,\,\,\,\,\,\,\,\,\,\,\,\,\,\,\,\,\,\,\,\,\,\,\,\,\,\,\,\,\,\,\,\,\,\,\,\,\,\,\,\,\,\,\,\,\,\,\,\,\,\,\,\,\,\,\,\,\,\,\,\,\,\,\,\,\,\,\,\,\,\,\,\,\,\,\,\,\,\,\,\,\,\,\,\,\,\,\,\, - M_i(A)(\calL_{i - 1}^T \times P_i^T)(I \times Q_{i, t}) \|_F \\
& \leq \Big(1 + O\Big(\frac{\eps}{q}\Big) \Big) \min_U \|U(M_i(\widetilde{U^1} \circ \cdots \circ \widetilde{U^{i - 1}})\calL_{i - 1}^T \times M_1(U^{i + 1}_* \circ \cdots \circ U^q_*)P_i^T)(I \times Q_{i, t}) \\
& \,\,\,\,\,\,\,\,\,\,\,\,\,\,\,\,\,\,\,\,\,\,\,\,\,\,\,\,\,\,\,\,\,\,\,\,\,\,\,\,\,\,\,\,\,\,\,\,\,\,\,\,\,\,\,\,\,\,\,\,\,\,\,\,\,\,\,\,\,\,\,\,\,\,\,\,\,\,\,\,\,\,\,\,\,\,\,\,\,\,\,\,\,\,\,\,\,\,\,\,\,\,\,\,\,\,\,\,\,\,\,\,\,\,\,\,\,\,\,\,\,  - M_i(A)(\calL_{i - 1}^T \times P_i^T)(I \times Q_{i, t}) \|_F
\end{split}
\end{equation}
it also holds that
\begin{equation}
\begin{split}
\|U & (M_i(\widetilde{U^1} \circ \cdots \circ \widetilde{U^{i - 1}})\calL_{i - 1}^T \times M_1(U^{i + 1}_* \circ \cdots \circ U^q_*)P_i^T) - M_i(A)(\calL_{i - 1}^T \times P_i^T) \|_F \\
& \leq \Big(1 + O\Big(\frac{\eps}{q}\Big) \Big) \min_U \|U(M_i(\widetilde{U^1} \circ \cdots \circ \widetilde{U^{i - 1}})\calL_{i - 1}^T \times M_1(U^{i + 1}_* \circ \cdots \circ U^q_*)P_i^T) \\
& \,\,\,\,\,\,\,\,\,\,\,\,\,\,\,\,\,\,\,\,\,\,\,\,\,\,\,\,\,\,\,\,\,\,\,\,\,\,\,\,\,\,\,\,\,\,\,\,\,\,\,\,\,\,\,\,\,\,\,\,\,\,\,\,\,\,\,\,\,\,\,\,\,\,\,\,\,\,\,\,\,\,\,\,\,\,\,\,\,\,\,\,\,\,\,\,\,\,\,\,\,\,\,\,\,\,\,\,\,\,\,\,\,\,\,\,\,\,\,\,\,\, - M_i(A)(\calL_{i - 1}^T \times P_i^T) \|_F
\end{split}
\end{equation}
For our purposes, it suffices for $\calE_{i, c, t}$ to occur with failure probability $\delta = \frac{1}{3}$. We will apply Theorem \ref{thm:approximate_multiple_l2_regression}, with $S$ being replaced by $I \times Q_{i, t}^T$:
\begin{itemize}
    \item First, we need $I \times Q_{i, t}^T$ to be a $(1 \pm \frac{1}{3})$ $\ell_2$ subspace embedding for the row span of $M_i(\widetilde{U^1} \circ \cdots \circ \widetilde{U^{i - 1}})\calL_{i - 1}^T \times M_i(U^{i + 1}_* \circ \cdots \circ U^q_*)P_i^T$ with probability at least $\frac{2}{3}$. By Lemma \ref{lemma:kronecker_product_subspace_embedding}, it suffices for $Q_{i, t}^T$ to be a $(1 \pm \frac{1}{3})$ $\ell_2$ subspace embedding for the row span of $M_i(U^{i + 1}_* \circ \cdots \circ U^q_*)P_i^T$. By Lemma \ref{lemma:sign_matrix_subspace_embedding}, this holds with probability at least $\frac{999}{1000}$ as long as $Q_{i, t}$ has $O(k)$ rows.
    \item In addition, we need $I \times Q_{i, t}^T$ to have the $(O(\sqrt{\frac{\eps}{qk}}), c)$ approximate matrix product property for a sufficiently small constant $c$. By Theorem \ref{thm:jl_property_implies_approx_matrix_product}, it suffices for $I \times Q_{i, t}^T$ to have the $(O(\sqrt{\frac{\eps}{qk}}), c, \ell)$ JL moment property for some $\ell \geq 2$. Furthermore, by Lemma \ref{lemma:jl_moment_with_identity_matrix}, this holds if $Q_{i, t}^T$ has the $(O(\sqrt{\frac{\eps}{qk}}), c, \ell)$ JL moment property, and by Lemma \ref{lemma:sign_matrix_jl_moment_property}, $Q_{i, t}$ has the $(O(\sqrt{\frac{\eps}{qk}}), c, \log(1/c))$ JL moment property as long as it has $O(\frac{qk}{\eps})$ rows.
    \item Finally, we need $E[(I \times Q_{i, t}^T)(I \times Q_{i, t})]$ to be the identity matrix.
    \begin{equation}
    \begin{split}
    E[(I \times Q_{i, t}^T)(I \times Q_{i, t})]
    & = I \times E[Q_{i, t}^T Q_{i, t}] \\
    & = I \times I \\
    & = I
    \end{split}
    \end{equation}
    where the second equality holds if we select the entries of $Q_{i, t}$ to be i.i.d., since the entries of $Q_{i, t}$ have mean $0$.
\end{itemize}
Thus, by Theorem \ref{thm:approximate_multiple_l2_regression}, as long as $Q_{i, t}$ has $O(\frac{qk}{\eps})$ rows, $\calE_{i, c, t}$ holds with probability at least $\frac{2}{3}$. 

\item Finally, for $t \in [T]$, we define the event $\calE_{i, d, t}$ as follows. First define
\begin{equation}
\begin{split}
& \widehat{U^i_t}
 = \argmin_U \|U(M_i(\widetilde{U^1} \circ \cdots \circ \widetilde{U^{i - 1}})\calL_{i - 1}^T \times M_1(U^{i + 1}_* \circ \cdots \circ U^q_*)P_i^T)(I \times Q_{i, t}^TR_{i, t}^T) \\
& \,\,\,\,\,\,\,\,\,\,\,\,\,\,\,\,\,\,\,\,\,\,\,\,\,\,\,\,\,\,\,\,\,\,\,\,\,\,\,\,\,\,\,\,\,\,\,\,\,\,\,\,\,\,\,\,\,\,\,\,\,\,\,\,\,\,\,\,\,\,\,\, - M_i(A)(\calL_{i - 1}^T \times P_i^T)(I \times Q_{i, t}^T R_{i, t}^T) \|_F \\
& = M_i(A) (\calL_{i - 1}^T \times P_i^T Q_{i, t}^T R_{i, t}^T) ((M_i(\widetilde{U^1} \circ \cdots \circ \widetilde{U^{i - 1}})\calL_{i - 1}^T)^\dagger \times (M_1(U^{i + 1}_* \circ \cdots \circ U^q_*)P_i^T Q_{i, t}^T R_{i, t}^T)^\dagger) \\
& = M_i(A) (\calL_{i - 1}^T \times P_i^T Q_{i, t}^T R_{i, t}^T) ((M_i(\widetilde{U^1} \circ \cdots \circ \widetilde{U^{i - 1}})\calL_{i - 1}^T)^\dagger \times X)
\end{split}
\end{equation}
where $X := (M_1(U^{i + 1}_* \circ \cdots \circ U^q_*)P_i^T Q_{i, t}^T R_{i, t}^T)^\dagger$ is unknown. We then let $\calE_{i, d, t}$ be the event that
\begin{equation}
\begin{split}
& \|\widehat{U^i_t}(M_i(\widetilde{U^1} \circ \cdots \circ \widetilde{U^{i - 1}})\calL_{i - 1}^T \times M_1(U^{i + 1}_* \circ \cdots \circ U^q_*)P_i^T)(I \times Q_{i, t}^T) \\
& \,\,\,\,\,\,\,\,\,\,\,\,\,\,\,\,\,\,\,\,\,\,\,\,\,\,\,\,\,\,\,\,\,\,\,\,\,\,\,\,\,\,\,\,\,\,\,\,\,\,\,\,\,\,\,\,\,\,\,\,\,\,\,\,\,\,\,\,\,\,\,\, - M_i(A)(\calL_{i - 1}^T \times P_i^T)(I \times Q_{i, t}^T) \|_F \\
& \,\,\,\,\,\,\,\, \leq \Big(1 + O\Big(\frac{\eps}{q}\Big)\Big) \min_U \|U(M_i(\widetilde{U^1} \circ \cdots \circ \widetilde{U^{i - 1}})\calL_{i - 1}^T \times M_1(U^{i + 1}_* \circ \cdots \circ U^q_*)P_i^T)(I \times Q_{i, t}^T) \\
& \,\,\,\,\,\,\,\,\,\,\,\,\,\,\,\,\,\,\,\,\,\,\,\,\,\,\,\,\,\,\,\,\,\,\,\,\,\,\,\,\,\,\,\,\,\,\,\,\,\,\,\,\,\,\,\,\,\,\,\,\,\,\,\,\,\,\,\,\,\,\,\, - M_i(A)(\calL_{i - 1}^T \times P_i^T)(I \times Q_{i, t}^T) \|_F \\
\end{split}
\end{equation}
\begin{lemma}[$e^{-\poly(k/\eps)}$ Success Probability for Multiple-Response $\ell_2$ Kronecker Regression] \label{lemma:gaussian_multiple_response_regression_kronecker}
 Let $m, n, k, s, t\in N, \eps > 0$, and $k, t < n$. Suppose $A\in \R^{mn\times k^2}$ has rank k and $B\in \R^{mn\times s}$ has rank $t$. In addition, suppose that $A$ can be written as $A_1\times A_2, A_1\in \R^{m\times k}, A_2\in \R^{n\times k}$. Let $R\in \R^{k\times n}$ have i.i.d. $N(0, 1/k)$ entries. If $X^* = \argmin\| AX - B\|_F^2$ and $\widehat{X} = \argmin \|(I_m\times R)AX - (I_m\times R)B\|_F^2$, then $\|A\widehat{X} - B\|_F^2\leq (1+\eps)\| AX^* - B\|_F^2$ with probability at least $e^{-\Theta(k^2\log k)}(\frac{\eps}{qk^2})^{O(kt)}$
\end{lemma}

\begin{proof}
As in the proof for the general version Lemma \ref{lemma:gaussian_with_exactly_k_rows_multiple_response_regression},
$$\|A\widehat{X} - B\|_F^2 = \|A(\widehat{X} - X^*)\|_F^2 + \|AX^* - B\|_F^2$$
Again, we aim to bound $\|A(\widehat{X} - X^*)\|_F$. Let $U\in \R^{mn\times k}$ be an orthonormal basis for the column span of $A$. As before, substituting $I\times R$ for $S$ in Lemma \ref{lemma:gaussian_with_exactly_k_rows_multiple_response_regression}, property $\calE$ holds  with probability at least $e^{-\Theta(k^2\log k)}e^{-k(k+t)}$. First, $I_m$ is always a subspace embedding for $A_1$, and if $R$ is a subspace embedding for $A_2$, $I\times R$ is a subspace embedding for $A$ by Lemma \ref{lemma:kronecker_product_subspace_embedding}. This happens with probability $e^{-\Theta(k^2\log k)}$ by Lemma \ref{lemma:gaussian_with_exactly_k_rows_multiple_response_regression} because $A_2$ has $k$ columns. Approximate matrix product result from Lemma \ref{lemma:gaussian_with_exactly_k_rows_approximate_matrix_product} also holds given subspace embedding, using the same proof as in Lemma \ref{lemma:gaussian_with_exactly_k_rows_multiple_response_regression}, conditioning on $\|(I\times R)B_j\|_2^2\leq \Theta(\frac \eps k)\|B_j\|_2^2$ instead. This probability is at least that of $\|SB_j\|_2^2\leq \Theta(\frac \eps k)\|B_j\|_2^2$ where $S$ is a regular Gaussian matrix with entries $N(0, 1/k)$, since $I\times R$ is a block matrix with nonzero entries drawn i.i.d. from $N(0, \frac 1 k)$. Using the subspace embedding,

\begin{equation}
    \begin{split}
    &\|A(\widehat{X} - X^*)\|_F = \|UU^TA(\widehat{X} - X^*)\|_F \\
    &= \|U^T A(\widehat{X} - X^*)\|_F \\
    &\leq \| U^T(I_m\times R)^T(I_m\times R)UU^TA(\widehat{X} - X^*)\|_F + \|(I - U^T(I_m\times R)^T(I_m\times R)U)U^TA(\widehat{X} - X^*)\|_F \\
    &= \|U^T (I_m\times R)^T(I_m\times R)A(\widehat{X} - X^*)\|_F + \|I - U^T(I_m\times R)^T(I_m\times R)U\|_2\|U^TA(\widehat{X} - X^*)\|_F \\
    &\leq \|U^T(I_m\times R)^T(I_m\times R)A(\widehat{X} - X^*)\|_F + \frac 1 3 \|U^T A (\widehat{X} - X^*)\|_F\\
    &= \|U^T(I_m\times R)^T(I_m\times R)A(\widehat{X} - X^*)\|_F + \frac 1 3 \|A (\widehat{X} - X^*)\|_F\\
    &\|A(\widehat{X} - X^*)\|_F\leq \frac 3 2 \| U(I_m\times R)^T(I_m\times R)A(\widehat{X} - X^*)\|_F
    \end{split}
\end{equation}
Using the same event $\calE$ as defined in Lemma \ref{lemma:gaussian_with_exactly_k_rows_multiple_response_regression}, noting that $A$ has $k^2$ columns, $\|U^T(I_m\times R)^T(I_m\times R)(A\widehat{X} - B)\|^2_F\leq k\eps\|A X^* - B\|^2|F$. Using normal equations for $\min_X\|(I_m\times R)AX - (I_m\times R)B\|_F$, $$U^T(I_m\times R)^T(I_m\times R)(A\widehat{X} - B) = A^T(I_m\times R)^T(I_m\times R)(A\widehat{X} - B) = 0$$
$$U^T(I_m\times R)^T(I_m\times R)A(\widehat{X} - X^*) = U^T(I_m\times R)^T(I_m\times R)(B - AX^*)$$
Conditioning on $\calE$,
$$\|A(\widehat{X} - X^*)\|_F\leq \frac 3 2 \| U^T(I_m\times R)^T(I_m\times R)(B - AX^*)\|_F\leq O(k\eps)\|AX^* - B\|_F^2$$
And therefore
$$\|A\widehat{X} - B\|_F = (1 + O(k\eps))\|AX^* - B\|_F^2$$
Taking the square root and using $\frac \eps k$ for $\eps$ in the approximate matrix product result gives
$$\|A \widehat{X} - B\|_F\leq (1+ O(\eps))\|AX^* - B\|_F$$ with probability at least $e^{-\Theta(k^2\log k)}(\frac{\eps}{k^2})^{O(kt)}$
\end{proof}

By our definition, $\widehat{U_t^i} = \argmin\| U A'(I\times R^T_{i, t}) - B'(I\times R^T_{i, t})\|_F$, where $A' = (M_i(\widetilde{U^1}\circ \cdots \circ \widetilde{U^{i-1}})\calL^T_{i-1}\times M_1(U_*^{i+1}\circ \cdots \circ U_*^q)P_i^T)(I\times Q_{i, t}^T)$, $B' = M_i(A)(\calL^T_{i-1}\times P^T_i)(I\times Q^T_{i, t})$. Therefore $\calE_{i, d, t}$ is equivalent to $\| \widehat{U_t^i}A' - B'\|_F \leq (1 + O(\eps / q))\min_U \|UA' - B'\|_F$. $A'\in \R^{k^2\times s_{i-1}m}, B'\in \R^{n\times s_{i-1}m}$. We can invoke the transpose of \ref{lemma:gaussian_multiple_response_regression_kronecker} and conclude that $\calE_{i, d, t}$ holds with probability $e^{-\Theta(k^2\log k)}(\frac{\eps}{qk^2})^{O(ks_{i-1}m)}$.

\end{itemize}
Therefore, for a fixed $t\in [T]$, $\calE_{i, c, t}$ and $\calE_{i, d, t}$ simultaneously hold with probability \\ $\frac 2 3 e^{-\Theta(k^2\log k)}(\frac{\eps}{qk^2})^{O(ks_{i-1}m)}$. The probability that for all $t\in [T]$, $\calE_{i, c, t}$ or $\calE_{i, d, t}$ fails is at most
\begin{equation}
    \begin{split}
        \left(1 - \frac 2 3 e^{-\Theta(k^2\log k)(\eps/qk^2)^{O(ks_{i-1}m)}}\right)^T = \left(1 - \frac{1}{\frac 3 2 e^{\Theta(k^2\log k) + \Theta(ks_{i-1}m\log \frac{qk^2}{\eps})}}\right)^T \\
        \leq e^{-\frac{T}{\frac 3 2 e^{\Theta(k^2\log k) + \Theta(ks_{i-1}m\cdot \log \frac{qk^2}{\eps})}}}
    \end{split}
\end{equation}
As long as $T\geq \log(\frac q \delta)\cdot \frac 3 2e^{\Theta(k^2\log k) + \Theta(s_{i-1}km\cdot \log \frac{qk^2}{\eps})}$, with probability $1 - O(\frac \delta q)$, there exists some $t\in [T]$ such that $\calE_{1, c, t}$ and $\calE_{1, d, t}$ both hold, $m\in O(k^2q/\eps)$.
.\\\\ Next we will show that, if $\exists t\in[T]$ such that $\calE_{i, c, t}$,$\calE_{i, d, t}$ simultaneously hold, in addition to $\calE_{i, a}$, $\calE_{i, b}$ both being satisfied, $\calE_i$ holds. Let $A_1 = M_i(\widetilde{U_1}\circ\cdots\circ\widetilde{U^{i-1}})$, $A_2 = M_1(U_*^{i+1}\circ\cdots U_*^q)$.

\begin{equation}
    \begin{split}
        &\min_{U^{i+1}, \cdots, U^q}\|\widetilde{U^1}\circ\cdots\widetilde{U^{i}}\circ U^{i+1}\circ\cdots\circ U^q - A\|_F \leq \|\widetilde{U^1}\circ\cdots\widetilde{U^{i-1}}\circ\widetilde{U^i}X'\circ U_*^{i+1}\circ\cdots U_*^q-A\|_F \\
        &= \| \widehat{U_t^i}\left(A_1 \times A_2 \right) - M_i(A)\|_F \\
        &\leq (1 + O(\frac \eps q))\| \widehat{U_t^i}(A_1\times A_2)(I\times P_i^T) - M_i(A)(I\times P_i^T))\|_F \\
        &\leq (1 + O(\frac \eps q)) \| \widehat{U_t^i}(A_1\calL^T_{i-1}\times A_2 P_i^T) - M_i(A)(\calL_{i-1}^T\times P_i^T)\|_F \\
        &\leq (1 + O(\frac \eps q))\min_U\| U(A_1\calL^T_{i-1}\times A_2P^T_i) - M_i(A)(\calL_{i-1}^T\times P_i^T) \|_F \\
        &\leq (1 + O(\frac \eps q))\min_U \| U(A_1\times A_2) - M_i(A)\|_F \\
        &= (1 + O(\frac \eps q)) \min_{U^i, \cdots, U_q}\|\widetilde{U^1}\circ\cdots\circ\widetilde{U^{i-1}}\circ U^i\circ\cdots U^q\|_F \\
        &\leq (1 + O(\frac \eps q))\left(1 + \frac{C\eps}{q}\right)^{i-1}\min_{U^1\cdots U^q}\|U^1\circ\cdots\circ U^q\|_F\\
        &\leq \left(1 + \frac{C\eps}{q}\right)^{i}\min_{U^1\cdots U^q}\|U^1\circ\cdots\circ U^q\|_F
    \end{split}
\end{equation}

The first inequality is true because $X' = (A_1\calL_{i-1}^T)^{\dagger}X$ where $X = (A_2P_i^TQ_{i, t}^QR_{i, t}^T)^{\dagger}$ is in $\R^{k\times k}$. The following two inequalities are due to $\calE_{i, a}$ and $\calE_{i, b}$. The fourth inequality is true when $\calE_{i, c, t}$ and $\calE_{i, d, t}$ for the $t$ picked (and that this $t$ must be picked by Subroutine \ref{algorithm:tensor_train_fpt_find_best_candidate} when it exists due to minimization). The fifth inequality is true when $\calL_{i-1}$ and $P_i$ are affine embeddings. By Lemma \ref{lemma:jl_moment_of_tensor_train_sketch_based_on_S_i} and Lemma \ref{lemma:jl_moment_implies_affine_embedding}, $\calL_{i-1}$ is an affine embedding with the $(\frac \eps q , \frac \delta q, \ell)$ JL moment property if $S_i$ has the $(\frac{\eps}{\sqrt{2q^3}}, \frac \delta q, \ell)$ JL moment property. $s_i = O(\frac{q^4k^2}{\eps^2\delta})$, $\calL_{i-1}$ is an $(1\pm \frac \eps q)$ affine embedding with probability $1 - O(\frac \delta q)$. $P_i$ is a Countsketch $O(\frac{q^3k^2}{\eps^2\delta})$ columns, therefore is an $1\pm \frac \eps q$ affine embedding with probability $1 - O(\frac \delta q)$ by \ref{thm:countsketch_affine_embedding}. By union bound on the failure probabilities of the aforementioned conditions, $\calE_i$ holds with probability $1 - O(\frac \delta q)$ as desired.

\paragraph{Final Case: Finding $U^q$ for each candidate $(\widetilde{U^1}, \cdots, \widetilde{U^{q-1}})$}.\\\\
By induction, we have proven that there is a candidate such that $$\min_{U^q}\|\widetilde{U_1}\circ\cdots\widetilde{U^{q-1}}\circ U^q\|_F \leq \left(1 + \frac{C\eps}{q}\right)^{q-1}\min_{U^1, \cdots, U^q}\|U^1\circ\cdots U^q - A\|_F.$$
By Algorithm \ref{algorithm:fpt_rank_exactly_k_tensor_train}, for each candidate tuple $(U^1, \cdots, U^q)$, using normal equations, $\widetilde{U^q}$ is the solution to $\min_U\|UM_q(U^1\circ \cdots\circ U^{q-1})\calL^T_{q-1} - M_q(A)\calL^T_{q-1}\|_F$. We need to show that $$\|U^1\circ\cdots\circ U^{q-1}\circ \widetilde{U^q} - M_q(A)\|F = (1\pm O(\frac \eps q)) \min_{U}\|UM_q(U^1\circ\cdots\circ U^{q-1}) - M_q(A)\|_F.$$
This holds when $\calL_{q-1}$ is an affine embedding for $M_q(U^1\circ\cdots\circ U^{q-1})$ and $M_q(A)$ with probability $1 - \frac{\delta}{q}$ by Lemma \ref{lemma:jl_moment_of_tensor_train_sketch_based_on_S_i} and Lemma \ref{lemma:jl_moment_implies_affine_embedding}. Therefore, if a good candidate $(\widetilde{U^1_*}, \cdots, \widetilde{U^{q-1}_*})$ exists in $L_{q-1}$ at the end, and the $\calL_{q-1}$ affine embedding holds, we generate a good candidate in $L_q$:
\begin{equation}
    \begin{split}
        \| \widetilde{U^1_*}\circ \cdots \circ \widetilde{U^q_*} - A\|_F &\leq \left (1 + \frac{C\eps}{q}\right)^q \min_{U^1, \cdots, U^q}\| U^1\circ\cdots\circ U^q - A\|_F \\
        &= \left(1 + O(\eps)\right)\min_{U^1, \cdots, U^q}\| U^1\circ\cdots\circ U^q - A\|_F
    \end{split}
\end{equation}
Each step of the inductive process fails with probability at most $O(\frac \delta q)$, as shown in each sub-event of $\calE_i$. By union bound, the overall failure probability is therefore at most $O(\delta)$.

The subroutine Algorithm \ref{algorithm:tensor_train_fpt_find_best_candidate} estimates the error of each candidate tuple in $L_q$ using the Tensor Train sketch $\calL_q$. We now determine the required sketch sizes of the $S_i$ in order to obtain the best candidate $(U^1, \ldots, U^q)$ up to a $(1 + \eps)$ factor, with a failure probability of at most $O(\delta)$. More precisely, we would like the following to hold with probability $1 - O(\delta)$: $\calL_q$ preserves $\|U^1 \circ \cdots \circ U^q - \widetilde{A}\|_F$ up to a factor of $(1 + O(\eps))$, for all candidates $(U^1, \ldots, U^q) \in L_q$. This means that $\calL_q(S_1, \ldots, S_q)$ should have the $(O(\eps), O(\frac{\delta}{|L_q|}), 2)$ JL moment property. By Lemma \ref{lemma:jl_moment_of_tensor_train_sketch_based_on_S_i}, it suffices for each of the $S_i$ to have the $(O(\frac{\eps}{\sqrt{q}}), O(\frac{\delta}{|L_q|}), 2)$ JL moment property. Since $|L_q| = T^{q - 1}$, this means that it suffices for the $S_i$ to have $O(\frac{q^2}{\eps} \cdot \frac{|L_q|}{\delta}) = O(\frac{q^2 T^{q - 1}}{\eps \delta})$ rows, in order for the candidate $(\widetilde{U^1}, \ldots, \widetilde{U^q})$ found by Algorithm \ref{algorithm:tensor_train_fpt_find_best_candidate} to be a $(1 + O(\eps))$-approximation with constant probability. In other words,

\begin{equation}
    \begin{split}
        \| \widetilde{U^1}\circ \cdots \circ \widetilde{U^q} - A\|_F &\leq (1+O(\eps))\min_{(U^1, \cdots, U^q)\in L_q}\| U^1\circ\cdots \circ U^q - A\|_F \\
        &\leq (1+O(\eps))\min_{U^1, \cdots, U^q}\| U^1\circ\cdots\circ U^q - A\|_F
    \end{split}
\end{equation}

\paragraph{Runtime Analysis}

Note that in each of the $q$ iterations, the CountSketch $P$ and linear sketch for tensor, $\calL$ is only applied once, rather than for each of the $T$ trials.

\begin{itemize}
    \item Generating $L_1$: For each $t$, $\widetilde{U_t^1} = M_1(A)P_1^TQ_{1,t}^TR_{1, t}^T$. $P_1$ is only applied to $A$ once, so the cost of this operation is $O(\nnz(A))$. Then, applying $Q_{1, t}^T R_{1, t}^T$ to the result takes $mr + km$ time for each trial, and the overall runtime of this step is $O(\nnz(A) + T(mr + km))$.
    \item Generating $L_2, \ldots, L_{q - 1}$: For a fixed $i$, applying $\calL_{i - 1}$ to $A$ takes an additional $O(\nnz(A))$ time (since we already have the result of applying $\calL_{i - 2}$ to $A$ from the $(i - 1)^{th}$ iteration). Applying $P_i$ to the result also takes $O(\nnz(A))$ time. For each of $|L_{i - 1}| T = T^i$ trials, applying $Q_{i, t}$, $R_{i, t}$, and $(M_i(U_1 \circ \cdots \circ U_{i - 1}) \calL_{i - 1}^T)^\dagger$ to the result takes $n \cdot \poly(\frac{qk}{\eps})$. Note that by Lemma \ref{lemma:ttsketch_contraction_train_time}, computing $(M_i(U_1 \circ \cdots \circ U_{i - 1}) \calL_{i - 1}^T)^\dagger$ also takes $n \cdot \poly(\frac{qk}{\eps})$ time. Thus, for iterations $2$ through $q - 1$, this step contributes $O(\nnz(A) + n \cdot \poly(\frac{qk}{\eps}) \cdot T^{q - 1})$ to the running time.
    \item Generating $L_q$: The cost of applying the last mode of $\calL_q$ is $O(\nnz(A))$ --- this only needs to be applied to $A$ once.

    \item Generating $L_q$: for each candidate, the additional cost of applying the last mode of $\calL_q$ is $O(\nnz(A))$. The cost of other operations is $O(k^2s_{q-1})$. The total cost is $O(T^{q-1}(k^2s_{q-1})+\nnz(A))$
    \item Selecting candidate: Computing $\widetilde{A}$ is $q\cdot \nnz(A)$ time since after each application of $S_i$, the number of nonzero entries left is still $\nnz(A)$. For each candidate, calculating $\calL_q(S_1, \cdots, S_1)(U^1\circ\cdots\circ U^q)$ is $O(qsnk(s+k))$ time by Lemma \ref{lemma:ttsketch_contraction_train_time}. The cost of computing the Frobenius norm of the error is $s_q$. The total cost is $O(T^q\cdot qsnk(s+k) + q\cdot\nnz(A))$, where $s = O(\frac{q^4 k^2}{\eps^2 \delta})$.
\end{itemize}
Taking $T = \log(\frac q \delta)\cdot e^{\Theta(k^2\log k + \Theta(qk^2/\eps \cdot \log (\frac{qk}{\eps})))}$, $s_i\in O(\frac{q^4k^2}{\eps^2\delta})$ in generation of $L_i$, $m\in O(k^2q/\eps)$, $r\in O(\frac{q^3k^2}{\eps^2\delta})$, we get the overall running time guarantee of $T^q \cdot n \cdot \poly(qk/(\eps\delta)) + q \cdot \nnz(A)$, which is $O(q \cdot \nnz(A)) + n \cdot \poly(\frac{qk}{\eps \delta}) \cdot e^{\Theta(\frac{q^2 k^2}{\eps} \log(\frac{qk}{\eps}))}$.

For space complexity,  note that an $m\times n$ CountSketch matrix can be represented in $O(\log m + \log n)$ using hash functions. Fixing $\delta$ as a constant, $\mathcal{L}$ take $O(q\log\poly(\frac{qk}{\eps}))$ space. $P_i$ take $O(q\log\poly(\frac{qk}{\eps}) + q^2\log n)$ space. $Q_i, R_i$ take $\poly(\frac{qk}{\eps})$ space. Computing each $\widetilde{U_i}$ takes $O(qn\nnz(A)\poly(\frac{qk}{\eps}))$. For finding the best candidate, each $S_i$ takes $O(\poly(\frac{qk}{\eps}) + \log\poly(\frac{qk}{\eps}))$. Computing the cost for each candidate takes $O(qn\nnz(A))$. Organizing the computation using backtrack yields polynomial space overall.

\end{proof}

\paragraph{Linear Sketch for Tensor Train}
This section contains the tools and proofs of the linear sketch $\calL$ used in the FPT Tensor Train algorithm \ref{algorithm:fpt_rank_exactly_k_tensor_train}.

\begin{definition}[Fusing First Two Modes]
\label{def:fusing_first_two_modes}
Let $A \in \R^{n_1 \times \ldots \times n_q}$ be a $q$-mode tensor, where $q \geq 2$. Then, we define $\calF(A) \in \R^{n_1n_2 \times \ldots \times n_q}$ by 
$$(\calF(A))_{(i_1, i_2), i_3, \ldots, i_q} = A_{i_1, \ldots, i_q}$$
Note that $\calF: \R^{n_1 \times \ldots \times n_q} \to \R^{n_1n_2 \times \ldots \times n_q}$ is a linear map.
\end{definition}

Next, we define $\calL_q$:

\begin{definition}[Linear Sketch for Tensor Train]
\label{def:tensor_train_sketch_with_placeholders}
Let $n, q \in \N$. In addition, suppose $S_1 \in \R^{s_1 \times n}$, and for $i \in [q] \setminus \{1\}$, suppose $S_i \in \R^{s_i \times s_{i - 1}n}$. Then, we define the linear map $\calL_q(S_1, \ldots, S_q): \R^{n \times \ldots \times n} \to \R^{s_q}$ as follows. Let $A \in \R^{n \times \ldots \times n}$ be a $q$-mode tensor. First, define $T_1 = S_1 \circ A \in \R^{s_1 \times n \times \ldots \times n}$ (which is a $q$-mode tensor), where $\circ$ denotes tensor contraction. Next, for $i \in [q] \setminus \{1\}$, define $T_i = S_i \circ \calF(T_{i - 1}) \in \R^{s_i \times n \times \ldots \times n}$ (which is a $(q - i + 1)$-mode tensor). Then, we define $\calL_q(S_1, \ldots, S_q)(A) = T_q \in \R^{s_q}$.
\end{definition}

\begin{remark}[Restating Definition \ref{def:tensor_train_sketch_with_placeholders} in terms of Kronecker Products]
\label{remark:tensor_train_sketch_in_terms_of_kronecker_products}
Note that Definition \ref{def:tensor_train_sketch_with_placeholders} can also be expressed in the following equivalent way. We can regard $\calL_q$ as a map defined on vectors in $\R^{n^q}$, i.e. $\calL_q(S_1, \ldots, S_q): \R^{n^q} \to \R^{s_q}$. Given $A \in \R^{n^q}$, we can define $T_1 := (S_1 \times I_{n^{q - 1}})A \in \R^{s_1 \cdot n^{q - 1}}$. Next, for each $i \in [q] \setminus \{2\}$, we can define $T_i = (S_i \times I_{n^{q - i}})T_{i - 1} \in \R^{s_i \times n^{q - i}}$. Thus, $T_q \in \R^{s_q}$, and we define $\calL_q(S_1, \ldots, S_q)(A) = T_q$. Note that if we identify $q$-mode tensors $A \in \R^{n \times \ldots \times n}$ with their vectorizations $\tensorvec(A) \in \R^{n^q}$, then this definition of $\calL_q$ is equivalent to Definition \ref{def:tensor_train_sketch_with_placeholders}.
\end{remark}

Definition \ref{def:tensor_train_sketch_with_placeholders} is useful for showing that $\calL_q$ can be applied to tensors of the form $U_1 \circ \cdots \circ U_q$ very efficiently:

\begin{lemma}
\label{lemma:ttsketch_contraction_train_time}
Let $k, q, n_1, \ldots, n_q \in \N$, and suppose $U_1 \in \R^{n_1 \times k}$, $U_j \in \R^{k \times n_j \times k}$ for $j \in [q] \setminus \{1, q\}$, and $U_q \in \R^{k \times n_q}$. In addition, suppose $S_1 \in \R^{s_1 \times n_1}$, and $S_j \in \R^{s_j \times s_{j - 1}n_j}$ for $j \in [q] \setminus \{1\}$. Then, $\calL_q(S_1, \ldots, S_q)(U_1 \circ \cdots \circ U_q)$ can be computed in $O(q s nk \cdot (s + k))$ time, where $n = \max_{j \in [q]} n_j$ and $s = \max_{i \in [q]} s_i$.
\end{lemma}
\begin{proof}
We use Definition \ref{def:tensor_train_sketch_with_placeholders}. First note that $T_1 = S_1 \circ (U_1 \circ \cdots \circ U_q)$ is simply $(S_1U_1) \circ U_2 \circ \cdots \circ U_q$, and $S_1 U_1$ can be computed in $O(s_1 n_1 k)$ time. Next, to compute $T_2$, we must first compute $\calF(T_1)$. However,
$$\calF(T_1) = \calF((S_1U_1) \circ U_2) \circ U_3 \circ \cdots \circ U_q$$
meaning that to (implicitly) compute $\calF(T_1)$, it suffices to compute $\calF((S_1 U_1) \circ U_2)$, and computing $\calF((S_1 U_1) \circ U_2)$ takes $O(s_1 n_2 k^2)$ time since $S_1 U_1 \in \R^{s_1 \times k}$ and $U_2 \in \R^{k \times n_2 \times k}$. In addition,
$$T_2 = S_2 \circ \calF(T_1) = S_2 \circ \calF((S_1 U_1) \circ U_2) \circ U_3 \circ \cdots \circ U_q$$
meaning that to compute $T_2$ (implicitly) it suffices to compute $S_2 \circ \calF((S_1 U_1) \circ U_2)$ --- since $S_2 \in \R^{s_2 \times s_1 n_2}$ and $(S_1 U_1) \circ U_2 \in \R^{s_1 \times n_2 \times k}$ (meaning $\calF((S_1 U_1) \circ U_2) \in \R^{s_1 n_2 \times k}$), computing $S_2 \circ \calF((S_1 U_1) \circ U_2) \in \R^{s_2 \times k}$ takes $O(s_2 s_1 n_2 k)$ time.

In general, we can proceed inductively as follows. Suppose we have computed $T_j \in \R^{s_j \times n \times \ldots \times n}$ (where $T_j$ has $(q - j + 1)$ modes) which can be written as
$$T_j = M_j \circ U_{j + 1} \circ \cdots \circ U_q$$
where $M_j \in \R^{s_j \times k}$. Then, let us determine the running time needed to compute $T_{j + 1}$ given $T_j$. Note that
$$T_{j + 1} = S_{j + 1} \circ \calF(T_j) = S_{j + 1} \circ \calF(M_j \circ U_{j + 1}) \circ U_{j + 2} \circ \cdots \circ U_q$$
Thus, if we define $M_{j + 1} = S_{j + 1} \circ \calF(M_j \circ U_{j + 1}) \in \R^{s_{j + 1} \times k}$, then
$$T_{j + 1} = M_{j + 1} \circ U_{j + 2} \circ \cdots \circ U_q$$
and therefore, to compute $T_{j + 1}$, it suffices to compute $M_{j + 1}$. Since $M_j \in \R^{s_j \times k}$ and $U_{j + 1} \in \R^{k \times n_{j + 1} \times k}$, computing $M_j \circ U_{j + 1} \in \R^{s_j \times n_{j + 1} \times k}$ takes $O(s_j n_{j + 1} k^2)$ time. Moreover, since $S_{j + 1} \in \R^{s_{j + 1} \times s_j n_{j + 1}}$ and $\calF(M_j \circ U_{j + 1}) \in \R^{s_j n_{j + 1} \times k}$, computing $M_{j + 1} = S_{j + 1} \circ \calF(M_j \circ U_{j + 1})$ takes $O(s_{j + 1} s_j n_{j + 1} k)$ time (once $\calF(M_j \circ U_{j + 1})$ has already been computed).

In summary, the overall running time needed to compute $\calL_q(S_1, \ldots, S_q)(U_1 \circ \cdots \circ U_q)$ is therefore
$$\sum_{j = 1}^{q - 1} O(s_j n_{j + 1} k^2) + \sum_{j = 1}^{q - 1} O(s_{j + 1} s_j n_{j + 1}k) + O(s_1 n_1 k) $$
Here, the first sum comes from the computation of $\calF(M_j \circ U_{j + 1})$ for $j \in [q - 1]$, and the second sum is the additional cost of computing $S_{j + 1} \circ \calF(M_j \circ U_{j + 1})$ for $j \in [q - 1]$ (once $\calF(M_j \circ U_{j + 1})$ has already been computed). If we let $s = \max_{j \in [q]} s_j$ and $n = \max_{j \in [q]} n_j$, then the overall running time is therefore $O(q \cdot (snk^2 + s^2 nk))$.
\end{proof}

The equivalent definition in Remark \ref{remark:tensor_train_sketch_in_terms_of_kronecker_products} in terms of Kronecker products is useful for showing that $\calL_q$ has the JL moment property if the $S_i$ are chosen properly --- in fact, it suffices for the $S_i$ to have the JL moment property, since we can then apply Lemma \ref{lemma:jl_moment_with_identity_matrix} to conclude that $S_i \times I_{n^{q - i}}$ has the JL moment property, and since $\calL_q$ is the composition of the maps $S_i \times I_{n^{q - i}}$ for $i \in [q]$, we can apply Lemma \ref{lemma:jl_moment_composition_of_jl_maps} to conclude that $\calL_q$ itself has the JL moment property. We make this precise below:

\begin{lemma}
\label{lemma:jl_moment_of_tensor_train_sketch_based_on_S_i}
Suppose that $S_1 \in \R^{s_1 \times n}$, and for $i \in [q] \setminus \{1\}$, $S_i \in \R^{s_i \times s_{i - 1}n}$, and $S_1, \ldots, S_q$ are independent. Moreover, suppose $S_i$ has the $(\frac{\eps}{\sqrt{2q}}, \delta, \ell)$ JL moment property. Then, $\calL_q(S_1, \ldots, S_q)$ has the $(\eps, \delta, \ell)$ JL moment property.
\end{lemma}
\begin{proof}
By Remark \ref{remark:tensor_train_sketch_in_terms_of_kronecker_products}, we can write $\calL_q(S_1, \ldots, S_q): \R^{n^q} \to \R^{s_q}$ as
$$\calL_q(S_1, \ldots, S_q) = S_q (S_{q - 1} \times I_n) \ldots (S_2 \times I_{n^{q - 2}}) (S_1 \times I_{n^{q - 1}})$$
The factors in the product on the right hand side are independent, and since $S_i$ has the $(\frac{\eps}{\sqrt{2q}}, \delta, \ell)$ JL moment property, by Lemma \ref{lemma:jl_moment_with_identity_matrix}, $(S_i \times I_{n^{q - i}})$ also has the $(\frac{\eps}{\sqrt{2q}}, \delta, \ell)$ JL moment property. By Lemma \ref{lemma:jl_moment_composition_of_jl_maps}, their product $\calL_q(S_1, \ldots, S_q)$ has the $(\eps, \delta, \ell)$ JL moment property. This proves the lemma.
\end{proof}

\newpage
\subsection{FPT Bicriteria CP Decomposition}

In this section, we give an algorithm for q-mode CP decomposition which outputs $U_1, \cdots, U_q\in \R^{n\times k^{q-1}}$ that obtain a $(1+\eps)$-approximation factor under the Frobenius norm. Formally, the algorithm solves the following problem:

\begin{problem}[$q$-mode CP rank-$k$ Decomposition]\label{problem:cp_decomposition_q_mode}
Let $A\in \R^{n\times \cdots \times n}$ be a $q$-mode tensor, and $k\in \N$. We wish to find a $q$-mode tensor $B\in \R^{n\times \cdots \times n}$ for which $\lVert A - B\rVert_F^2$ is as small as possible, such that $B$ has tensor rank at most $k$. In other words, we would like to find matrices $U_1, \cdots, U_q\in \R^{n\times k}$ that minimizes $$\lVert\sum_{i = 1}^{k} {U_1}_i\otimes \cdots \otimes {U_q}_i - A\rVert_F^2$$
\end{problem}

\begin{definition}[Flattening of $q$-mode tensors]\label{def:flattening_q_mode}
Let $A\in \R^{n\times \cdots \times n}$ be a $q$-mode tensor, flattening $A$ along the each mode respectively yields matrices $A_1, \cdots, A_q\in \R^{n\times n^{q-1}}$ such that $\forall j_1, \cdots, j_q\in [n], A_{j_1, \cdots, j_q} = {A_i}_{j_i, M_i}$, where $M_i = \sum_{l\in [q], l\neq i, r = q-2}^{r = 0}n^{r}j_{l}$. We refer to the inverse of a flatten operation as a re-tensorization.
\end{definition}

\begin{algorithm}
\caption{$(1+\eps)$-approximation bicriteria rank $k^{q-1}$ decomposition with $(q(\nnz(A)) + n\poly(k, q/\eps))^{\poly(k, q/\eps)}$}
\label{algorithm:q_mode_bicriteria_CP_decomposition}
\begin{algorithmic}

\Require A $q-$mode tensor $A\in\R^{n\times \cdots\ n}, k\in \N$
\Ensure $\widehat{U_1}, \cdots, \widehat{U_q}\in \R^{n\times k^{q - 1}}$
\State $\eps'\leftarrow \eps/(q^{3/2})$
\State Generate sketching matrices $S_1, \cdots S_{q-1}\in \R^{n^{q-1}\times k}$ by composing CountSketch matrices of size $n^{(q-1)}\times O(k^2/\eps'^2)$ and Gaussian matrices of size $O(k^2/\eps'^2)\times k$
\State Generate CountSketch matrices $T_1, \cdots, T_{q-1}$ of rank $O(k^2/\eps^2)$ to implement TensorSketch
\State $\widehat{U_i}\in \R^{n\times k^{q-1}}\gets$ matrices with $(A_iS_i)_{j_i}$ as columns $(j_1, \cdots, j_{q-1})$ for $i\in[q-1]$.
\State $A_q\in\R^{n\times n^{q-1}}\gets$ flattening of $A$ along the $q^{th}$ mode as in \ref{def:flattening_q_mode}
\State $Z_q\in \R^{k^{q-1}\times n^{q-1}} \gets$ $\vec((A_jS_j)_{i_j})\otimes \cdots\otimes (A_{q-1}S_{q-1})_{i_{q-1}})$ as the $M_i^{th}$ row, where $M_i = \sum_{j = 1,r = q-2}^{j = q-1, r=0}(i_j-1)^r + 1$.
\State Solve for minimizer $\widehat{U_q}$ for $\min_{U_q\in \R^{n\times k^{q-1}}}\lVert YZ_q - A_q\rVert_F^2$ with $(q-1)$-mode TensorSketch, sketching each face and constructing $Z_q$ recursively. \\
\Return $\widehat{U_1}, \cdots \widehat{U_q}$
\end{algorithmic}

\end{algorithm}

\begin{algorithm}
\caption{$(1+\eps)$-approximation rank-$k$ decomposition with $(q\cdot \poly(k/\eps)\cdot \nnz(A))^{\poly(k, q)}$ runtime and constant success probability.}
\label{algorithm:q_mode_rank_k_CP_decomposition}
\begin{algorithmic}
\Require A $q$-mode tensor $A \in \R^{n \times \cdots \times n}$, $k \in \N$
\Ensure $U_1,\cdots, U_q \in \R^{n \times k}$
\State $\eps' \gets \eps / (2q)$
\State Generate sketching matrices $S_1, \cdots, S_q\in \R^{n^{q-1}\times k}$ by composing CountSketch matrices of size $n^{(q-1)}\times O(k^2/\eps'^2)$ and Gaussian matrices of size $O(k^2/\eps'^2)\times k$.
\State $A_1, \cdots, A_q\gets$ flattenings of $A$ along each mode as in Definition \ref{def:flattening_q_mode}.
\State Compute $A_iS_i$ for each $i\in [q]$

\State $A_0' \gets A$
\For{$i = 1\rightarrow q$}
    \State ${A_f}_i\in \R^{n\times k^{i-1}n^{q-i}}\gets$ flattening of $A_{i-1}'$ along mode $i$.
    \State ${A_f'}_i\in \R^{k\times k^{i-1}n^{q-i}}\gets (A_iS_i)^T{A_f}_i$
    \State $A'_i\in \R^{k\times k\times \cdots n\times \cdots\times n}\gets$ retensorize$({A_f'}_i)$ along the $i^{th}$ mode. There are $i$ dimensions of size $k$ and $q-i$ dimensions of size $n$.
\EndFor
\State $A'\in \R^{k\times \cdots\times k} \gets A_q'$, a $q-$mode tensor.
\State $\widehat{X_1}, \cdots \widehat{X_q\gets}$ $\argmin \lVert X_1\otimes \cdots \otimes X_q - A'\rVert_F^2$, where $X_i\in \R^{k\times k}, A'$ is a $q-$mode tensor in $\R^{k\times \cdots\times k}$.\\
\Return $A_1S_1\widehat{X_1}, \cdots, A_qS_q\widehat{X_q}$
\end{algorithmic}
\end{algorithm}

\begin{lemma}\label{lemma:good_cp_sol_in_span}
Let $A\in \R^{n\times \cdots\times n}$, $S_i$ be the sketching matrices defined in Algorithm \ref{algorithm:q_mode_bicriteria_CP_decomposition}. Problem \ref{problem:cp_decomposition_q_mode} could be approximated by finding $X_1, \cdots, X_q\in \R^{k\times k}$ that minimize
$$\lVert \sum_{i = 1}^k (A_1S_1X_1)_i \otimes \cdots \otimes (A_qS_qX_q)_i - A\rVert_F^2$$
Additionally, there exists $X_1, \cdots, X_q$ such that the optimum is at most $(1+\eps) = (1 + \eps')^q$ the optimal solution for Problem \ref{problem:cp_decomposition_q_mode}.
\end{lemma}

\begin{proof}
The proof is based on the proof for Theorem C.1. in \cite{swz19_tensor_low_rank}.

Let $A \in \R^{n\times \cdots \times n}$ be a $q-$mode tensor, $U_1^*, \cdots, U_q^* = \argmin_{U_i}\lVert \sum_{i = 1}^k \bigotimes_{j\in [q]} {U_j}_i - A\rVert_F^2$.
Let matrices $Z_i \in \R^{k\times n^{(q-1)}}$ have rows $\text{vec}\left(\bigotimes_{j\in [q], j\neq i}{U_j^*}_i\right)$.

Consider the optimization problem $$\min_{U_1\in \R^{n\times k}}\lVert \sum_{i = 1}^k {U_1}_i\otimes {U_2^*}_i\otimes \cdots \otimes {U_q^*}_i\rVert_F^2$$ By definition of squared Frobenius norm, an equivalent objective function with the same optimum as the original problem is
$$\min_{U_1\in \R^{n\times k}}\lVert U_1Z_1 - A_1\rVert_F^2$$
Applying the matrix $S_1^T\in \R^{k\times n^{q - 1}}$ yields the sketched optimization problem
$$\min_{U_1\in \R^{n\times k}} \lVert U_1Z_1S_1 - A_1S_1\rVert_F^2$$
By Lemma \ref{lemma:gaussian_with_exactly_k_rows_multiple_response_regression}, letting $\widehat{U}_1 = A_1S_1(Z_1S_1)^{\dagger}$ gives
$\lVert \widehat{U_1}Z_1 - A_1\rVert_F^2\leq (1 + \eps')\mathcal{OPT}$ with probability $e^{\Theta(k^2\log k)}(\frac {\eps}{kq})^{O(k^2)}$.
Retensorizing along the first mode gives
$$\lVert \widehat{U_1}\otimes U_2^*\otimes \cdots \otimes U_q^*\rVert_F^2\leq (1 + \eps')\mathcal{OPT}$$
To repeat the process above, in the $i^{th}$ round, we fix $\widehat{U}_1, \cdots, \widehat{U}_{i-1}, U^*_{i+1}, \cdots, U^*_q\in \R^{n\times k}$ and use the sketching matrix $S_i\in \R^{n^{q-1}\times k}$. Let $\widehat{U_i} = A_iS_i(Z_iS_i)^{\dagger}$. We get
$$\lVert \widehat{U}_1\otimes \cdots \otimes \widehat{U}_i\otimes U^*_{i+1}\otimes \cdots U^*_q - A\rVert_F^2\leq (1+\eps')^i\mathcal{OPT}$$
At the end of the $q^{th}$ round, we have
$$\lVert\bigotimes_{j\in [q]}\widehat{U}_j - A\rVert_F^2\leq (1 + \eps')^q\mathcal{OPT}$$
with probability $e^{\Theta(qk^2\log k)}(\frac{\eps}{kq})^{O(qk^2)}$.\\
Replacing $\eps = \eps'\cdot (2q)$, we get $$\lVert\bigotimes_{j\in [q]}\widehat{U}_j - A\rVert_F^2\leq (1 + \eps)\mathcal{OPT}$$
Note that all $\widehat{U}_j$ are in the column spans of 
$A_jS_j$. i.e. $\widehat{U}_j = A_jS_jX_j$ for some unknown $X_j\in \R^{k\times k}$. This completes the proof.
\end{proof}

\begin{lemma}[Input Sparsity Reduction]\label{lemma:cp_decomposition_reduce_to_poly_k}
Let $A\in \R^{n\times \cdots \times n}$ be a $q-$mode tensor. Let $V_j$ be matrices of size $n\times k$ for $j\in [q]$. For every optimization problem
$$\min\lVert \sum_{i = 1}^k (V_1X_1)_i\otimes \cdots \otimes (V_qX_q)_i - A\rVert_F^2$$
with minimizers $X_1^*, \cdots, X_q^*$, there exists $B\in \R^{k\times \cdots \times k}$ a $q-$mode tensor such that
$$X_1^*, \cdots, X_q^* = \argmin_{X_j, j\in [q]}\lVert X_1\otimes \cdots \otimes X_q - B\rVert_F^2$$
\end{lemma}

\begin{proof}
Note that one could write $V_j = V_j'Y_j$ where the columns of $V_j'$ are orthonormal bases of the column span of $V_j$, and that $Y_j$ could be folded into the unknown $X_j$. Without losing generality, consider the case where $V_j$'s have orthonormal columns.

Let matrices $Z_i \in \R^{k\times n^{(q-1)}}$ have rows $\text{vec}\left(\bigotimes_{j\in [q], j < i}X_j\otimes\bigotimes_{j\in [q], j> i}V_jX_j\right)$. Let $A_j$ be the flattened $A$ along the $j^{th}$ mode as in definition \ref{def:flattening_q_mode}.
\begin{align*}
    \lVert \sum_{i = 1}^k(V_1X_1)_i \otimes \cdots \otimes (V_qX_q)_i - A\rVert_F^2 &= \lVert V_1X_1Z_1 - A_1\rVert_F^2 \\
    &= \lVert (V_1V_1^T - I_n)A_1 + V_1(X_1Z_1 - V_1^TA_1)\rVert_F^2 \\
    &= \lVert V_1(X_1Z_1 - V_1^TA_1)\rVert_F^2 + \lVert (V_1V_1^T - I_n)A_1\rVert_F^2 & V_1V_1^T\perp V_1 \\
    &= \lVert X_1Z_1 - V_1^TA_1\rVert_F^2 + \lVert (V_1V_1^T - I_n)A_1\rVert_F^2
\end{align*}
The last line is true since $V_1$ has orthonormal columns. Note that $(V_1V_1^T - I_n)A_1$ is a constant term. Therefore, minimizing $\lVert \sum_{i = 1}^k (V_1X_1)_i\otimes \cdots \otimes (V_qX_q)_i - A\rVert_F^2$ is equivalent to minimizing $\lVert X_1Z_1 - V_1^TA_1\rVert_F^2 = \lVert \sum_{i = 1}^k {X_1}_i\otimes (V_2X_2)_i\otimes \cdots \otimes (V_qX_q)_i - V_1^TA_1\rVert_F^2$.

Repeat the process for $j = 2$ by flattening $\lVert \sum_{i = 1}^k {X_1}_i\otimes (V_2X_2)_i\otimes \cdots \otimes (V_qX_q)_i - V_1^TA_1\rVert_F^2$ along the second mode, folding over $V_2^T$, and retensorizing again. Iterate up to $j = q$. With the same argument, one could inductively prove that after the $j^{th}$ iteration, we arrive at an equivalent form of the problem
$$\min_{X_j}\lVert \sum_{i = 1}^k {X_1}_i\otimes \cdots \otimes {X_j}_i\otimes (V_jX_j)_i \otimes \cdots \otimes (V_qX_q)_i - A^{(j)}\rVert_F^2$$
where $A^{(0)} = A, A^{(j)} = \text{Retensorize}(V_j^T(A^{(j-1)})_j)$, where $A^{(j-1)})_j$ is the flattening along the $j^{th}$ mode of $A^{(j-1)}$.

After the $q^{th}$ iteration, we prove the desired equivalence.
\end{proof}

\begin{theorem}[Bicriteria $k^{q-1}$-Rank Approximation for CP Decomposition]
Let $A\in \R^{n\times \cdots \times n}$ be a $q-$mode tensor. Algorithm \ref{algorithm:q_mode_rank_k_CP_decomposition} finds a $(1+\eps)$-approximation in 
$(q(\nnz(A)) + n\poly(k, q/\eps))^{\poly(k, q/\eps)}$ with constant probability.
\end{theorem}
\begin{proof}
By Lemma \ref{lemma:good_cp_sol_in_span}, for $\eps' < \frac 1 2$
$$\min_{X_1, \cdots, X_q}\lVert \sum_{l=1}^k \bigotimes_{j=1}^q (A_jS_jX_j)_l\rVert_F^2\leq (1+\eps')^q\mathcal{OPT}\leq (1 + 2q\eps')\mathcal{OPT}$$
Expanding the tensor products, this is equivalent to 
$$\min_{X_1, \cdots, X_q} \lVert \sum_{i_1, \cdots i_q = 1}^k \left(\bigotimes_{j = 1}^{q-1} (A_jS_j)_{i_j}\right) \otimes (\sum_{l = 1}^k(A_qS_qX_q)_l(X_1)_{i, l}\cdots(X_{q-1})_{i, l}) - A\rVert_F^2\leq (1+2q\eps')\mathcal{OPT}$$
Let $Z_q$ be as defined in Algorithm \ref{algorithm:q_mode_bicriteria_CP_decomposition}. $Z_q\in \R^{k^{q-1}\times n^{q-1}}$. Note that $Z_q$ is the matricization of the tensor product of the first $q-1$ sketched faces. As all the $X_j$'s are folded into $$\sum_{l = 1}^k(A_qS_qX_q)_l(X_1)_{i, l}\cdots(X_{q-1})_{i, l}$$, we arrive at an equivalent standard tensor regression problem with $\widehat{U_1}, \cdots, \widehat{U_{q-1}}$ be known, as specified in Algorithm \ref{algorithm:q_mode_bicriteria_CP_decomposition}, and $U_q$ being unknown. Since flattening does not affect squared Frobenius norm due to symmatry, we have
$$\min_{U_q\in \R^{n\times k^{q-1}}}\lVert U_qZ_q - A_q\rVert_F^2\leq (1+2q\eps')\mathcal{OPT}$$
where $Z_q\in \R^{k^{q-1}\times n^{q-1}}, A_q\in \R^{n\times n^{q-1}}$.
However, explicitly solving for this problem introduced a runtime in $\Omega(n^{q-1})$. Therefore, we generate sketching matrices $T_1, \cdots, T_{q-1}$ , sketch each $A_jS_j$ in advance, and construct sketched $Z_q$ and $A_q(T_1, \cdots, T_{q-1}, I)$ iteratively as in \cite{akkpvwz20_tensor_sketch_outer_product}, achieving $O(\nnz(A) \poly(k/\eps))^q$ cost for computing the sketch and solving the sketched problem, with an additional $\sqrt{q}$ factor on the error term. Boosting the success probability of Lemma \ref{lemma:good_cp_sol_in_span} by repeating the sketch and taking the minimum error yields the desired runtime and success probability.
\end{proof}

\begin{remark}
[rank-$k$ Approximation for CP Decomposition]
Algorithm \ref{algorithm:q_mode_rank_k_CP_decomposition} outputs $\widehat{U_1}, \cdots, \widehat{U_q}\in \R^{n\times k}$ such that
$$\lVert \sum_{i = 1}^k (\widehat{U_1})_i\otimes \cdots (\widehat{U_q})_i - A\rVert _F^2\leq (1+\eps)\mathcal{OPT}$$
\end{remark}

\begin{proof}
The theorem immediate follows Lemma \ref{lemma:good_cp_sol_in_span} and \ref{lemma:cp_decomposition_reduce_to_poly_k}, and that $\min_{X_1, \cdots, X_q}\lVert X_1\otimes \cdots\otimes X_q - A\rVert_F^2$ could be solved as a polynomial system in $2^{\Omega(k^q)}$.
\end{proof}

\newpage
\section{Hardness of Rank-$1$ Decomposition with Frobenius Norm} \label{sec:rank_1_hardness}

\begin{definition}
Let $B \in \R^{n \times n}$. Then, we define $\|B\|_{2, 4} = \Big( \sum_{i = 1}^n \|B_{i, :}\|_2^4 \Big)^{1/4}$, where $B_{i, :}$ is the $i^{th}$ row of $B$, and $\|B\|_{2 \to 4} = \sup_{x \in \R^n \setminus \{0\}} \frac{\|Bx\|_4}{\|x\|_2}$.
\end{definition}

\begin{proof}[Proof of Theorem \ref{thm:hardness_rank1}]
Let $B \in \R^{n \times n}$, and define $A \in \R^{n \times n \times n}$ as follows: for all $k \in \N$, $A^k = B_k B_k^T$, where $B_k \in \R^n$ is the transpose of the $k^{th}$ row of $B$. Then, for any unit vectors $u, v \in \R^n$,
\begin{equation} \label{eq:reduce_to_2_to_4_norm_part_1}
\begin{split}
\sum_{k = 1}^n \|uu^T A^k vv^T - A^k \|_F^2
& = \sum_{k = 1}^n \|uu^T A^k vv^T \|_F^2 - 2 \Tr((A^k)^T uu^T A^k vv^T) + \|A^k\|_F^2 \\
& = \sum_{k = 1}^n \|uu^T A^k vv^T \|_F^2 - 2 \Tr(A^k uu^T A^k vv^T) + \|A^k\|_F^2 \\
& = \sum_{k = 1}^n (u^T A^k v)^2 - 2 \Tr(A^k uu^T A^k vv^T) + \|A^k \|_F^2 \\
& = \sum_{k = 1}^n (u^T A^k v)^2 - 2 \Tr(u^T A^k vv^T A^k u) + \|A^k \|_F^2 \\
& = \sum_{k = 1}^n \|A^k \|_F^2 - (u^T A^k v)^2
\end{split}
\end{equation}
The third equality above is because $u$ and $v$ are unit vectors. The fourth equality above is because $\Tr(XYZ) = \Tr(ZXY)$ for any matrices $X, Y, Z$. The fifth equality is because $A^k$ is PSD (meaning that $v^T A^k u = u^T A^k v$). To simplify the final expression in Equation \ref{eq:reduce_to_2_to_4_norm_part_1}, we use the following Lemma:

\begin{lemma} [From \cite{hlz10_quadratic_to_bilinear}, Proposition 6.10 of \cite{mostNPHard}]
\label{lemma:quadratic_to_bilinear}
Let $A_k \in \R^{n \times n}$ be a symmetric matrix for $k \in [m]$. Then,
$$\max_{\|u\|_2 = 1} \sum_{k = 1}^m \langle A_k u, u \rangle^2 = \max_{\|x\|_2 = \|y\|_2 = 1} \sum_{k = 1}^m \langle A_k x, y \rangle^2$$
\end{lemma}
\begin{proof}
This result is shown in the proof of Proposition 2 of \cite{hlz10_quadratic_to_bilinear}.
\end{proof}

Thus, by Equation \ref{eq:reduce_to_2_to_4_norm_part_1},
\begin{equation} \label{eq:reduce_to_2_to_4_norm_part_2}
\begin{split}
\min_{\|u\|_2 = \|v\|_2 = 1} \sum_{k = 1}^n \|uu^T A^k vv^T - A^k \|_F^2
& = \min_{\|u\|_2 = \|v\|_2 = 1} \sum_{k = 1}^n \|B_k\|_2^4 - (u^T A^k v)^2 \\
& = \sum_{k = 1}^n \|B_k\|_2^4 - \max_{\|u\|_2 = \|v\|_2 = 1} \sum_{k = 1}^n (u^T A^k v)^2 \\
& = \sum_{k = 1}^n \|B_k\|_2^4 - \max_{\|u\|_2 = 1} \sum_{k = 1}^n (u^T A^k u)^2 \\
& = \|B\|_{2, 4}^4 - \max_{\|u\|_2 = 1} \sum_{k = 1}^n (u^T B_k B_k^T u)^2 \\
& = \|B\|_{2, 4}^4 - \max_{\|u\|_2 = 1} \sum_{k = 1}^n \langle u, B_k \rangle^2 \langle B_k, u \rangle^2 \\
& = \|B\|_{2, 4}^4 - \max_{\|u\|_2 = 1} \sum_{k = 1}^n \langle B_k, u \rangle^4 \\
& = \|B\|_{2, 4}^4 - \max_{\|u\|_2 = 1} \|Bu\|_4^4 \\
& = \|B\|_{2, 4}^4 - \|B\|_{2 \to 4}^4
\end{split}
\end{equation}
where the third equality is by Lemma \ref{lemma:quadratic_to_bilinear}, since $A^k$ is a PSD matrix.

Finally, we will show that obtaining a $(1 + \Theta(\frac{1}{n}))$-approximate solution $u, v$ to Tucker-$(2, 3)$ decomposition would allow us to compute a $\Theta(1)$-approximation to $\|B\|_{2 \to 4}$. To show this, we compare $\|B\|_{2, 4}$ to $\|B\|_{2 \to 4}$:

\begin{lemma} \label{lemma:2_to_4_norm_relationship}
$\|B\|_{2, 4}^4 \geq \|B\|_{2 \to 4}^4 \geq \frac{\|B\|_{2, 4}^4}{n}$.
\end{lemma}
\begin{proof}
First, for any $k \in [n]$, let $B_k \in \R^n$ be the transpose of the $k^{th}$ row of $B$. Then,
\begin{equation}
\begin{split}
\|B\|_{2 \to 4}^4
& \geq \Big\|B \cdot \frac{B_k}{\|B_k\|_2} \Big\|_4^4 \\
& \geq \frac{\langle B_k, B_k \rangle^4}{\|B_k\|_2^4} \\
& = \|B_k\|_2^4
\end{split}
\end{equation}
Thus,
$$\|B\|_{2 \to 4}^4 \geq \max_{k \in [n]} \|B_k\|_2^4 \geq \frac{1}{n} \sum_{k = 1}^n \|B_k\|_2^4 = \frac{\|B\|_{2, 4}^4}{n}$$
This shows the second inequality in the lemma. The first equality follows from Equation \ref{eq:reduce_to_2_to_4_norm_part_2}.
\end{proof}

Now suppose there is an algorithm which computes unit vectors $u, v$ such that
$$\sum_{k = 1}^n \|uu^T A^k vv^T - A^k \|_F^2 \leq \Big(1 + \frac{1}{Cn}\Big) \min_{\|u_*\|_2 = \|v_*\|_2 = 1} \sum_{k = 1}^n \|u_* u_*^T A^k v_* v_*^T - A^k \|_F^2$$
for some absolute constant $C \geq 1$. Then, this algorithm computes $F$ such that
$$\|B\|_{2, 4}^4 - \|B\|_{2 \to 4}^4 \leq F \leq \Big(1 + \frac{1}{Cn}\Big) \cdot (\|B\|_{2, 4}^4 - \|B\|_{2 \to 4}^4)$$
It follows from the left inequality that
$$\|B\|_{2, 4}^4 - F \leq \|B\|_{2 \to 4}^4$$
and from the right inequality that
\begin{equation}
\begin{split}
\|B\|_{2, 4}^4 - F
& \geq \|B\|_{2, 4}^4 - \Big(1 + \frac{1}{Cn} \Big) \cdot (\|B\|_{2, 4}^4 - \|B\|_{2 \to 4}^4) \\
& = \|B\|_{2 \to 4}^4 - \frac{1}{Cn}\cdot (\|B\|_{2, 4}^4 - \|B\|_{2 \to 4}^4) \\
& = \Big(1 + \frac{1}{Cn}\Big) \|B\|_{2 \to 4}^4 - \frac{\|B\|_{2, 4}^4}{Cn} \\
& \geq \Big(1 + \frac{1}{Cn}\Big) \|B\|_{2 \to 4}^4 - \frac{\|B\|_{2 \to 4}^4}{C} \\
& \geq \Big(1 - \frac{1}{C} + \frac{1}{Cn} \Big) \|B\|_{2 \to 4}^4
\end{split}
\end{equation}
where the second inequality is by Lemma \ref{lemma:2_to_4_norm_relationship}. In summary,
$$\Big(1 - \frac{1}{C} + \frac{1}{Cn}\Big) \|B\|_{2 \to 4}^4  \leq  \|B\|_{2, 4}^4 - F \leq \|B\|_{2 \to 4}^4$$
Thus, $\|B\|_{2 \to 4}^4 - F$ is equal to $\|B\|_{2 \to 4}^4$ up to a constant factor, and assuming Conjecture \ref{conjecture:2_to_4_norm_running_time}, it thus requires at least $2^{\Omega(n)}$ time to compute unit vectors $u, v$ satisfying
$$\sum_{k = 1}^n \|uu^T A^k vv^T - A^k \|_F^2 \leq \Big(1 + \frac{1}{Cn}\Big) \min_{\|u_*\|_2 = \|v_*\|_2 = 1} \sum_{k = 1}^n \|u_* u_*^T A^k v_* v_*^T - A^k \|_F^2$$
In particular, finding a $(1 + \eps)$-approximation for Tucker-$(2, 3)$, for rank $k = 1$, requires at least $2^{\Omega(1/\eps)}$ time assuming Conjecture \ref{conjecture:2_to_4_norm_running_time}.
\end{proof}

\newpage

\section{Tucker-$(p, q)$ Decomposition with A Robust Loss Function} \label{sec:robust_loss_fn}

In this section, we give an algorithm for Tucker-$(p, q)$ decomposition, using the robust loss function $\|\cdot\|_R$, defined in Definition \ref{def:robust_loss_function}. Formally, we solve the following problem:
\begin{problem}[Robust Tucker-$(p, q)$ Decomposition]
\label{problem:robust_tucker_decomposition}
Let $A \in \R^{n \times \ldots \times n}$ be a $q$-mode tensor, with $p < q$, and let $k \in \N$. Then, we wish to find a $q$-mode tensor $X \in \R^{n \times \ldots \times n}$ of multilinear rank at most $(k, \ldots, k, n, \ldots, n)$ (where the first $p$ entries of the tuple are $k$ and the last $q - p$ entries are $n$), such that $\|X - A\|_R$ is minimized.
\end{problem}

For convenience, we will use the following notation in this section:

\begin{definition}
Let $T \in \R^{n \times \ldots \times n}$ be a $q$-mode tensor. Then, we define $T^i := T(:, \ldots, :, i)$, that is, $T^i$ is the slice of $T$ whose index in the last mode is $i$.
\end{definition}

\begin{algorithm}
\caption{Algorithm for robust Tucker-$(p, q)$ decomposition with bicriteria rank $O(k \log^2 k)$ on each mode. Here $\textsc{12NormApprox}(A, k)$ denotes an algorithm for finding the best (up to an $O(1)$ factor) rank-$k$ approximation to $A$ in the $\ell_{1, 2}$ norm, with bicriteria rank $O(k \log^2 k)$.}
\label{algorithm:robust_tucker_p_q_decomposition}
\begin{algorithmic}
\Require A $q$-mode tensor $A \in \R^{n \times \ldots \times n}$, $p < q$, $k \in \N$
\Ensure $U^1, \ldots, U^p \in \R^{n \times O(k \log^2 k)}$

\For{$m = 1 \to p$}
    \For{$i = 1 \to n$}
        \State $S_{m, i} \gets $ An $O(pnk^2 \log^4 k) \times n^{q - 1}$ Countsketch matrix
        \State $T_{m, i} \gets $ An $O(k \log^2 k + \log(np)) \times O(pnk^2 \log^4 k)$ i.i.d. sign matrix
        \State $B_{m, i} \gets M_m(A^i) S_{m, i}^T T_{m, i}^T \in \R^{n \times O(k \log^2 k + \log(np))}$
    \EndFor
    \State $B_m \gets [B_{m, 1}, \ldots, B_{m, n}] \in \R^{n \times n \cdot O(k \log^2 k + \log(np))}$
    \State $U, V \gets \textsc{12NormApprox}(B_m, k)$ (best result after $O(\log p)$ trials)
    \State $\widehat{U^m} \gets U$
\EndFor \\ 
\Return $\widehat{U^1}, \ldots, \widehat{U^p}$
\end{algorithmic}
\end{algorithm}

\begin{theorem}
There is an algorithm for Problem \ref{problem:robust_tucker_decomposition} (Algorithm \ref{algorithm:robust_tucker_p_q_decomposition}) which, given a $q$-mode tensor $A \in \R^{n \times \ldots \times n}$, and $k \in \N$, outputs $\widehat{U^1}, \ldots, \widehat{U^p} \in \R^{n \times O(k \log^2 k)}$ for which, with probability at least $\frac{4}{5}$,
\begin{equation}
\begin{split}
\min_G \Big\|A - \sum_{i_1, \ldots, i_q} G(i_1, \ldots, i_p, i_{p + 1}, & \ldots, i_q)  \widehat{U^1}(:, i_1) \otimes \ldots \otimes \widehat{U^p}(:, i_p) \otimes e_{i_{p + 1}} \otimes \ldots \otimes e_{i_q} \Big\|_R \\
& \leq O(p \cdot (\sqrt{k} \log k + \sqrt{\log (np)})) \cdot \min_{X \text{ rank } (k, \ldots, k, n, \ldots, n)} \|X - A\|_R 
\end{split}
\end{equation}
The running time of Algorithm \ref{algorithm:robust_tucker_p_q_decomposition} is $O(p \cdot \nnz(A)) + n^3 \log n \cdot \poly(k)$.
\end{theorem}
\begin{proof}
Define
$$\mathcal{OPT} = \min_{X \text{ rank }(k, \ldots, k, n \ldots, n)} \|X - A\|_R$$
Let $X$ be a tensor of multilinear rank at most $(k, \ldots, k, n, \ldots, n)$. Then, we can write
$$X = \sum_{i_1, i_2, \ldots, i_p \in [p], i_{p + 1}, \ldots, i_q \in [n]} G(i_1, \ldots, i_p, i_{p + 1}, \ldots, i_q) U^1(:, i_1) \otimes \ldots \otimes U^p(:, i_p) \otimes e_{i_{p + 1}} \otimes \ldots \otimes e_{i_q}$$
for some $U^1, \ldots, U^p \in \R^{n \times k}$, where for $j \in [n]$, $e_j$ is the $j^{th}$ standard basis vector. Thus, for any $i \in [q]$, we can write
\begin{equation}
\begin{split}
X^i
& = \sum_{i_1, \ldots, i_p \in [p], i_{p + 1}, \ldots, i_{q - 1} \in [n]} G(i_1, \ldots, i_p, i_{p + 1}, \ldots, i_{q - 1}, i) U^1(:, i_1) \otimes \ldots \otimes U^p(:, i_p) \otimes e_{i_{p + 1}} \otimes \ldots \otimes e_{i_{q - 1}} \\
& = \sum_{i_1, \ldots, i_p \in [p], i_{p + 1}, \ldots, i_{q - 1} \in [n]} G^i(i_1, \ldots, i_p, i_{p + 1}, \ldots, i_{q - 1}) U^1(:, i_1) \otimes \ldots \otimes U^p(:, i_p) \otimes e_{i_{p + 1}} \otimes \ldots \otimes e_{i_{q - 1}}
\end{split}
\end{equation}
In particular, reshaping gives
$$\|X^i - A^i\|_F = \|(U^1 \times \ldots \times U^p \times I_n \times \ldots \times I_n) \tensorvec(G^i) - \tensorvec(A^i)\|_2$$
for each $i \in [n]$, meaning that
$$\|X - A\|_R = \sum_{i \in [n]} \|(U^1 \times \ldots \times U^p \times I_n \times \ldots \times I_n) \tensorvec(G^i) - \tensorvec(A^i)\|_2$$
Now, observe that we can assume, without loss of generality, that $U^m$ has orthonormal columns, for $m \in [p]$ --- otherwise, we can simply modify $G^i$ appropriately for each $i \in [n]$, since
\begin{equation}
\begin{split}
\|X - A\|_R
& = \sum_{i \in [n]} \|(U^1 \times \ldots \times U^p \times I_n \times \ldots \times I_n) \tensorvec(G^i) - \tensorvec(A^i) \|_2 \\
& = \sum_{i \in [n]} \|U^m M_m(G^i) (U^1 \times \ldots \times U^{m - 1} \times U^{m + 1} \times \ldots \times U^p \times I_n \times \ldots \times I_n)^T - M_m(A^i) \|_F
\end{split}
\end{equation}
meaning we can multiply $M_m(G^i)$ on the left by an appropriate $k \times k$ matrix. Next observe the following:

\begin{lemma} \label{lemma:replace_U_with_projection}
Problem \ref{problem:robust_tucker_decomposition} is equivalent to
$$\min_{P^1, \ldots, P^p, G^1, \ldots, G^n} \sum_{i \in [n]} \|(P^1 \times \ldots \times P^p \times I_n \times \ldots \times I_n)\tensorvec(G^i) - \tensorvec(A^i)\|_2 $$
where for $m \in [p]$, $P^m \in \R^{n \times n}$ is a rank-$k$ orthogonal projection matrix, and for $i \in [n]$, $G^i \in \R^{n \times \ldots \times n}$ is a $(q - 1)$-mode tensor.
\end{lemma}
\begin{proof}
As shown above, Problem \ref{problem:robust_tucker_decomposition} is equivalent to
$$\min_{U^1, \ldots, U^p, G^1, \ldots, G^n} \sum_{i \in [n]} \|(U^1 \times \ldots \times U^p \times I_n \times \ldots \times I_n) \tensorvec(G^i) - \tensorvec(A^i) \|_2$$
Fix $U^1, \ldots, U^p \in \R^{n \times k}$, and $G^1, \ldots, G^n$ with the appropriate dimensions. Then, by reshaping, we find that
\begin{equation} \label{eq:from_basis_to_projection_mode_1}
\begin{split}
\sum_{i \in [n]} \|(U^1 \times \ldots \times U^p \times I_n & \times \ldots \times I_n ) \tensorvec(G^i) - \tensorvec(A^i) \|_2 \\
& = \sum_{i \in [n]} \|U^1 M_1(G^i) (U^2 \times \ldots \times U^p \times I_n \times \ldots \times I_n)^T - M_1(A^i) \|_F \\
& = \sum_{i \in [n]} \|U^1 (U^1)^T M_1(\widetilde{G^i}) (U^2 \times \ldots \times U^p \times I_n \times \ldots \times I_n)^T - M_1(A^i) \|_F \\
& = \sum_{i \in [n]} \|P^1 M_1(\widetilde{G^i}) (U^2 \times \ldots \times U^p \times I_n \times \ldots \times I_n)^T - M_1(A^i) \|_F \\
& = \sum_{i \in [n]} \|(P^1 \times U^2 \times \ldots \times U^p \times I_n \times \ldots \times I_n) \tensorvec(\widetilde{G^i}) - \tensorvec(A^i) \|_2
\end{split}
\end{equation}
Here the first equality is by reshaping. For the second equality, first recall our assumption (without loss of generality) that $U^1$ has orthonormal columns. We can also assume without loss of generality that the columns of $M_1(G^i)$ are in the column span of $(U^1)^T$ (otherwise, the portion of those columns that is in the null space of $U^1$ could simply be discarded with no effect). Thus, there exists $\widetilde{G^i}$ such that $M_1(G^i) = (U^1)^T M_1(\widetilde{G^i})$. Note that the dimensions of $\widetilde{G^i}$ are the same as those of $G^i$, except that the dimension of $\widetilde{G^i}$ in the first mode is $n$ instead of $k$. For the third equality, we let $P^1 = U^1(U^1)^T$, which is an orthogonal projection matrix since $U^1$ has orthonormal columns, and the fourth equality is by reshaping again.

In this way, we can repeat this argument for each mode $m \in [p]$ to find that
\begin{equation} \label{eq:from_basis_to_projection_all_modes}
\begin{split}
\sum_{i \in [n]} \|(U^1 \times \ldots \times U^p \times I_n & \times \ldots \times I_n ) \tensorvec(G^i) - \tensorvec(A^i) \|_2 \\
& = \sum_{i \in [n]} \|(P^1 \times P^2 \times \ldots \times P^p \times I_n \times \ldots \times I_n) \tensorvec(\widetilde{G^i}) - \tensorvec(A^i) \|_2
\end{split}
\end{equation}
where for $m \in [p]$, $P^m = U^m (U^m)^T$. Therefore,
\begin{equation}
\begin{split}
\min_{U^m, G^i} \sum_{i \in [n]} & \|(U^1 \times \ldots \times U^p \times I_n \times \ldots \times I_n ) \tensorvec(G^i) - \tensorvec(A^i) \|_2 \\
& \geq \min_{P^m, G^i} \sum_{i \in [n]} \|(P^1 \times P^2 \times \ldots \times P^p \times I_n \times \ldots \times I_n) \tensorvec(\widetilde{G^i}) - \tensorvec(A^i) \|_2
\end{split}
\end{equation}
The reverse inequality can also be obtained by reversing the argument in Equation \ref{eq:from_basis_to_projection_mode_1}, letting $U^1$ be an orthonormal basis for the column span of $P^1$. This proves the lemma.
\end{proof}

We now extend Lemma \ref{lemma:replace_U_with_projection} slightly:

\begin{lemma} \label{lemma:replace_U_with_projection_no_core_tensor}
Problem \ref{problem:robust_tucker_decomposition} is equivalent to
\begin{equation} \label{eq:projection_version_problem_robust}
\begin{split}
\min_{P^1, \ldots, P^p} \sum_{i \in [n]} \|(P^1 \times \ldots \times P^p \times I_n \times \ldots \times I_n)\tensorvec(A^i) - \tensorvec(A^i)\|_2
\end{split}
\end{equation}
where for $m \in [p]$, $P^m \in \R^{n \times n}$ is a rank-$k$ orthogonal projection matrix, and for $i \in [n]$, $G^i \in \R^{n \times \ldots \times n}$ is a $(q - 1)$-mode tensor.
\end{lemma}
\begin{proof}
Note that Lemma \ref{lemma:replace_U_with_projection} states that Problem \ref{problem:robust_tucker_decomposition} is equivalent to
$$\min_{P^1, \ldots, P^p, G^1, \ldots, G^n} \sum_{i \in [n]} \|(P^1 \times \ldots \times P^p \times I_n \times \ldots \times I_n)\tensorvec(G^i) - \tensorvec(A^i)\|_2 $$
where the dimension of $G^i$ is $n$ in every mode. We now show that $G^i$ can be replaced with $A^i$. For each $m \in [p]$, write $P^m = U^m(U^m)^T$, where $U^m \in \R^{n \times k}$ is an orthonormal basis for the column span of $P^m$. Then, by standard properties of the Kronecker product (see page 4 of \cite{sdfhpf17_tensor_decomposition_survey}),
\begin{equation} \label{eq:kronecker_product_projection_matrix}
\begin{split}
P^1 \times \ldots \times P^p \times I_n & \times \ldots \times I_n
 = (U^1 (U^1)^T) \times \ldots \times (U^p (U^p)^T) \times I_n \times \ldots \times I_n \\
& = (U^1 \times \ldots \times U^p \times I_n \times \ldots \times I_n)((U^1)^T \times \ldots \times (U^p)^T \times I_n \times \ldots \times I_n) \\
& = (U^1 \times \ldots \times U^p \times I_n \times \ldots \times I_n)(U^1 \times \ldots \times U_p \times I_n \times \ldots \times I_n)^T
\end{split}
\end{equation}
Here the second equation is by the fact that for matrices $A, B, C, D$ whose dimensions match appropriately, $(A \times B)(C \times D) = (AC \times BD)$, and the third equation is by the fact that $A^T \times B^T = (A \times B)^T$. Now, we recall a result from \cite{djssw19_kronecker_product_regression}:

\begin{lemma}[Proposition 3.2 of \cite{djssw19_kronecker_product_regression}]
Let $U_i \in \R^{n_i \times d_i}$ be an orthonormal basis for $A_i \in \R^{n_i \times d_i}$, for $i \in [q]$. Then, $U_1 \times \ldots \times U_q$ is an orthonormal basis for $A_1 \times \ldots \times A_q$.
\end{lemma}

In particular, it follows from Proposition 3.2 of \cite{djssw19_kronecker_product_regression} that $U^1 \times \ldots \times U^p \times I_n \times \ldots \times I_n$ has orthonormal columns, and therefore, by Equation \ref{eq:kronecker_product_projection_matrix}, $P^1 \times \ldots \times P^p \times I_n \times \ldots \times I_n$ is a projection matrix, meaning that the regression objective below:
$$\|(P^1 \times \ldots \times P^p \times I_n \times \ldots \times I_n) \tensorvec(G^i) - \tensorvec(A^i)\|_2$$
is minimized when $G^i = A^i$, from which the lemma follows.
\end{proof}

Observe that from the proof of the above lemma, it is also clear that $P^i$ corresponds to $U^i (U^i)^T$ in the first formulation of the problem, assuming that $U^i$ has orthonormal columns. Now, using the above lemma, we show a lower bound on $\mathcal{OPT}$ that is useful for bounding the approximation factor of Algorithm \ref{algorithm:robust_tucker_p_q_decomposition}:
\begin{lemma} \label{lemma:opt_lower_bound_robust}
For all $m \in [p]$,
$$\mathcal{OPT} \geq \min_{P^m} \sum_{i \in [n]} \|P^m M_m(A^i) - M_m(A^i)\|_F$$
where the minima are taken over rank-$k$ orthogonal projection matrices $P^m$.
\end{lemma}
\begin{proof}
By Lemma \ref{lemma:replace_U_with_projection_no_core_tensor}, 
$$\mathcal{OPT} = \min_{P^1, \ldots, P^p} \sum_{i \in [n]} \|(P^1 \times \ldots \times P^p \times I_n \times \ldots \times I_n)\tensorvec(A^i) - \tensorvec(A^i)\|_2$$
Fix rank-$k$ orthogonal projection matrices $P^1_*, \ldots, P^p_* \in \R^{n \times n}$ such that
$$\sum_{i \in [n]} \|(P^1_* \times \ldots \times P^p_* \times I_n \times \ldots \times I_n)\tensorvec(A^i) - \tensorvec(A^i)\|_2 = \mathcal{OPT}$$
Then, for any $m \in [p]$,
\begin{equation} \label{eq:first_statement}
\begin{split}
\mathcal{OPT}
& = \sum_{i \in [n]} \|(P^1_* \times \ldots \times P^p_* \times I_n \times \ldots \times I_n) \tensorvec(A^i) - \tensorvec(A^i) \|_2 \\
& = \sum_{i \in [n]} \|P^m_* M_m(A^i) (P^1_* \times \ldots \times P^{m - 1}_* \times P^{m + 1}_* \times \ldots \times P^p_* \times I_n \times \ldots \times I_n)^T - M_m(A^i) \|_F \\
& \geq \sum_{i \in [n]} \|P^m_* M_m(A^i) - M_m(A^i) \|_F
\end{split}
\end{equation}
Here the second equality is by reshaping, and the first inequality is by the Pythagorean theorem (since $P^m_* M_m(A^i)$ is the closest matrix to $M_m(A^i)$, in Frobenius norm, whose columns are in the column span of $P^m_*$). This proves the lemma.
\end{proof}

We next prove an auxiliary lemma that helps us upper bound the error of Algorithm \ref{algorithm:robust_tucker_p_q_decomposition}:

\begin{lemma} \label{lemma:upper_bound_term_decomposition}
Let $P^1, \ldots, P^p \in \R^{n \times n}$ be orthogonal projection matrices. Then,
$$\|(P^1 \times \ldots \times P^p \times I_n \times \ldots \times I_n)\tensorvec(A^i) - \tensorvec(A^i)\|_F \leq \sum_{m \in [p]} \|P^m M_m(A^i) - M_m(A^i)\|_F$$
for any $i \in [n]$.
\end{lemma}
\begin{proof}
Let us show by induction on $t \in [p]$ that
$$\|(P^1 \times \ldots \times P^t \times I_n \times \ldots \times I_n)\tensorvec(A^i) - \tensorvec(A^i) \|_F \leq \sum_{m \in [t]} \|P^m M_m(A^i) - M_m(A^i)\|_F$$
from which the lemma follows. For the base case, observe that
\begin{equation}
\begin{split}
\|(P^1 \times I_n \times \ldots \times I_n)\tensorvec(A^i) - \tensorvec(A^i)\|_F
& = \|P^1 M_1(A^i) - M_1(A^i)\|_F
\end{split}
\end{equation}
by reshaping, as desired. Now, suppose for some $t \in [p]$, with $t < p$, that
$$\|(P^1 \times \ldots \times P^t \times I_n \times \ldots \times I_n) \tensorvec(A^i) - \tensorvec(A^i)\|_F \leq \sum_{m \in [t]} \|P^m M_m(A^i) - M_m(A^i) \|_F$$
Then,
\begin{equation}
\begin{split}
\|(P^1 \times \ldots \times P^{t + 1} \times & I_n \times \ldots \times I_n) \tensorvec(A^i) - \tensorvec(A^i)\|_F \\
& = \|P^{t + 1} M_{t + 1}(A^i) (P^1 \times \ldots \times P^t \times I_n \times \ldots \times I_n)^T - M_{t + 1}(A^i) \|_F \\
& \leq \|P^{t + 1} M_{t + 1}(A^i) (P^1 \times \ldots \times P^t \times I_n \times \ldots \times I_n)^T - P^{t + 1} M_{t + 1}(A^i) \|_F \\
& \,\,\,\,\,\,\,\, + \|P^{t + 1} M_{t + 1}(A^i) - M_{t + 1}(A^i)\|_F \\
& \leq \sum_{m \in [t + 1]} \|P^m M_m(A^i) - M_m(A^i) \|_F
\end{split}
\end{equation}
Here the first equality is by reshaping along the $(t + 1)^{th}$ mode. The first inequality is by the triangle inequality. The second inequality is by the induction hypothesis. This completes the induction, and thus the proof of the lemma.
\end{proof}

Finally, we perform a union bound over $i \in [n]$ to show that $S_{m, i}^T$ and $T_{m, i}^T$ give $(\eps, 0, O(k \log^2 k))$-PCPs for their respective $M_m(A^i)$, with at least constant probability. (Here we let $\eps = \frac{1}{10}$ for convenience.)

\begin{lemma} \label{lemma:robust_pcp_union_bound}
For $m \in [p]$ and $i \in [n]$, let $\calE_{m, i}$ be the event that $S_{m, i}^T$ is an $(\eps, 0, O(k \log^2 k))$-PCP for $M_m(A^i)$, and $\calF_{m, i}$ the event that $T_{m, i}^T$ is an $(\eps, 0, O(k \log^2 k))$-PCP for $M_m(A^i) S_{m, i}^T$. Finally, define $\calE = \bigcap_{m \in [p], i \in [n]} (\calE_{m, i} \cap \calF_{m, i})$. Then, $\calE$ occurs with probability at least $\frac{9}{10}$, as long as $S_{m, i}$ has at least $O(\frac{pnk^2 \log^4 k}{\eps^2})$ rows and $T_{m, i}$ has at least $O(\frac{k \log^2 k + \log(np)}{\eps^2})$ rows.

In addition, let $\calG$ be the event that, for all $m \in [p]$, $$\min_V \|\widehat{U^m} V - B_m\|_{1, 2} \leq O(1) \min_{X \text{ rank }k} \|X - B_m\|_{1, 2}$$
where $\widehat{U^1}, \ldots, \widehat{U^p} \in \R^{n \times O(k \log^2 k)}$ are the factors output by Algorithm \ref{algorithm:robust_tucker_p_q_decomposition}. Then, $\calG$ occurs with probability at least $\frac{9}{10}$.
\end{lemma}
\begin{proof}
Suppose $S_{m, i}$ has at least $O(\frac{pnk^2 \log^4 k}{\eps^2})$ rows. Then, by Theorem \ref{thm:countsketch_pcps_zero_additive} with $\delta$ replaced by $\frac{1}{20 np}$, $\eps$ replaced by $\eps$, and $k$ replaced by $O(k \log^2 k)$, $M_m(A^i)S_{m, i}^T$ is an $(\eps, 0, O(k \log^2 k))$-PCP for $M_m(A^i)$ with probability at least $1 - \frac{1}{20 np}$. Thus, for each $m \in [p]$ and $i \in [n]$, $\calE_{m, i}$ holds with probability at least $1 - \frac{1}{20 np}$. In addition, suppose $T_{m, i}$ has at least $O(\frac{k \log^2 k + \log(np)}{\eps^2})$ rows. Then, with probability at least $1 - \frac{1}{20np}$, $M_m(A^i) S_{m, i}^T T_{m, i}^T$ is an $(\eps, 0, O(k \log^2 k))$-PCP for $M_m(A^i) S_{m, i}^T$. In summary, $\calE_{m, i}$ and $\calF_{m, i}$ hold with probability at least $1 - \frac{1}{20 np}$ for all $m \in [p]$ and $i \in [n]$, and by a union bound, $\calE$ holds with probability at least $\frac{9}{10}$.

In addition, note that for any $m \in [p]$, with probability $1 - \frac{1}{10p}$, $\min_V \|U^m V - B_m\|_{1, 2} \leq O(1) \min_{X \text{ rank }k} \|X - B_m\|_{1, 2}$, since to obtain $U^m$, we apply the subroutine $\textsc{12NormApprox}$ for $O(\log p)$ trials and let $U^m$ be result of the trial which gives the least $\ell_{1, 2}$ norm error. Thus, by a union bound, $\calG$ occurs with probability at least $\frac{9}{10}$.
\end{proof}

Now, we analyze the error of Algorithm \ref{algorithm:robust_tucker_p_q_decomposition}. Let $\widehat{U^1}, \ldots, \widehat{U^p} \in \R^{n \times O(k \log^2 k)}$ be the output of Algorithm \ref{algorithm:robust_tucker_p_q_decomposition}, and let $\widehat{P^1}, \ldots, \widehat{P^p} \in \R^{n \times n}$ be the respective orthogonal projection matrices (i.e. $\widehat{P^m}$ is the orthogonal projection onto the column span of $\widehat{U^m}$). In addition, condition on $\calE$ and $\calG$ (which occur simultaneously with probability at least $\frac{4}{5}$ by Lemma \ref{lemma:robust_pcp_union_bound}). Then, the error obtained by Algorithm \ref{algorithm:robust_tucker_p_q_decomposition} can be bounded as follows. In the following sequence of inequalities, we use $a \lesssim b$ to denote $a \leq Cb$ for some absolute constant $C$ that is independent of $a$ and $b$.
\begin{equation} \label{eq:robust_final_bound_part_1}
\begin{split}
\sum_{i \in [n]} \|(P^1 \times \ldots \times P^p & \times I_n \times \ldots \times I_n) \tensorvec(A^i) - \tensorvec(A^i) \|_F \\
& \leq \sum_{i \in [n]} \sum_{m \in [p]} \|P^m M_m(A^i) - M_m(A^i) \|_F \\
& \lesssim \sum_{i \in [n]} \sum_{m \in [p]} \|P^m M_m(A^i) S_{m, i}^T - M_m(A^i) S_{m, i}^T \|_F \\
& \lesssim \sum_{i \in [n]} \sum_{m \in [p]} \|P^m M_m(A^i) S_{m, i}^T T_{m, i}^T - M_m(A^i) S_{m, i}^T T_{m, i}^T \|_F \\
& \leq \sum_{i \in [n]} \sum_{m \in [p]} \|P^m M_m(A^i) S_{m, i}^T T_{m, i}^T - M_m(A^i) S_{m, i}^T T_{m, i}^T \|_{1, 2} \\
& = \sum_{i \in [n]} \sum_{m \in [p]} \|P^m B_{m, i} - B_{m, i} \|_{1, 2} \\
& = \sum_{m \in [p]} \|P^m B_m - B_m\|_{1, 2}
\end{split}
\end{equation}
Here the first inequality is by Lemma \ref{lemma:upper_bound_term_decomposition}. The second and third inequalities are by conditioning on $\calE$, since $P^m$ is an orthogonal projection matrix of rank at most $O(k \log^2 k)$. The fourth inequality is because, for any matrix $M$, $\|M\|_F \leq \|M\|_{1, 2}$. The first equality is by the definition $B_{m, i} := M_m(A^i) S_{m, i}^T T_{m, i}^T$, and the second equality is by the definition $B_m := [B_{m, 1}, \ldots, B_{m, n}]$ (and the fact that the $\ell_{1, 2}$ norm decomposes across columns).

If we condition on $\calG$, then since $\textsc{12NormApprox}$ is an $O(1)$-approximation algorithm, there exists an absolute constant such that for all $m \in [m]$ and for all rank-$k$ projection matrices $Q^m$,
$$\|P^m B_m - B_m\|_{1, 2} \leq C \|Q^m B_m - B_m\|_{1, 2}$$
Thus, if $r$ is the number of rows in $T_{m, i}$ for $m \in [p]$ and $i \in [n]$, then
\begin{equation} \label{eq:robust_final_bound_part_2}
\begin{split}
\sum_{m \in [p]} \|P^m B_m - B_m \|_{1, 2}
& \lesssim \sum_{m \in [p]} \|Q^m B_m - B_m \|_{1, 2}
 \\
& = \sum_{i \in [n]} \sum_{m \in [p]} \|Q^m B_{m, i} - B_{m, i} \|_{1, 2} \\
& = \sum_{i \in [n]} \sum_{m \in [p]} \|Q^m M_m(A^i) S_{m,i }^T T_{m, i}^T - M_m(A^i) S_{m, i}^T T_{m, i}^T \|_{1, 2} \\
& \lesssim \sqrt{r} \sum_{i \in [n]} \sum_{m \in [p]} \|Q^m M_m(A^i) S_{m,i }^T T_{m, i}^T - M_m(A^i) S_{m, i}^T T_{m, i}^T \|_F \\
& \lesssim \sqrt{r} \sum_{i \in [n]} \sum_{m \in [p]} \|Q^m M_m(A^i) - M_m(A^i) \|_F
\end{split}
\end{equation}
Here, the first inequality is because for any matrix $M$ with $r$ columns, $\|M\|_{1, 2} \leq \sqrt{r} \|M\|_F$. The second inequality is by our assumption that $\calE$ holds.

Thus, combining Equations \ref{eq:robust_final_bound_part_1} and \ref{eq:robust_final_bound_part_2}, there is an absolute constant $C$ such that for all rank-$k$ projection matrices $Q^1, \ldots, Q^p$,
$$\sum_{i \in [n]} \|(P^1 \times \ldots \times P^p \times I_n \times \ldots \times I_n) \tensorvec(A^i) - \tensorvec(A^i) \|_F \leq C\sqrt{r} \sum_{m \in [p]} \sum_{i \in [n]} \|Q^m M_m(A^i) - M_m(A^i) \|_F$$
Taking the infimum over all $Q^1, \ldots, Q^p$ on the right-hand side gives
\begin{equation}
\begin{split}
\sum_{i \in [n]} \|(P^1 \times \ldots \times P^p & \times I_n \times \ldots \times I_n) \tensorvec(A^i) - \tensorvec(A^i) \|_F \\
& \leq C\sqrt{r} \min_{Q^1, \ldots, Q^p} \sum_{m \in [p]} \sum_{i \in [n]} \|Q^m M_m(A^i) - M_m(A^i) \|_F \\
& = C\sqrt{r} \sum_{m \in [p]} \min_{Q^m} \sum_{i \in [n]} \|Q^m M_m(A^i) - M_m(A^i) \|_F \\
& \leq C\sqrt{r} \sum_{m \in [p]} \mathcal{OPT} \\
& \leq Cp\sqrt{r} \cdot \mathcal{OPT} \\
& \leq O(p \cdot (\sqrt{k} \log k + \sqrt{\log (np)})) \cdot \mathcal{OPT}
\end{split}
\end{equation}
where the second inequality is by Lemma \ref{lemma:opt_lower_bound_robust}. Thus, Algorithm \ref{algorithm:robust_tucker_p_q_decomposition} obtains an $O(p \cdot (\sqrt{k} \log k + \sqrt{\log(np)}))$ approximation factor, and obtains a bicriteria rank of $O(k \log^2 k)$ on each mode.

Let us now analyze the running time of Algorithm \ref{algorithm:robust_tucker_p_q_decomposition}. First, for a single $m \in [p]$, let us determine the running time needed to find $\widehat{U^m}$. For $i \in [n]$, $B_{m, i}$ can be computed in $\nnz(A^i) + n \cdot O(pnk^2 \log^4 k) \cdot O(k\log^2 k + \log(np)) = \nnz(A^i) + n^2 \log n \cdot \poly(pk)$. Thus, $B_m$ can be computed in $\nnz(A) + n^3 \log n \cdot \poly(pk)$. Next, observe that $B_m \in \R^{n \times n \cdot O(k \log^2 k + \log(np))}$. Thus, the running time needed for one single trial of $\textsc{12NormApprox}(B_m, k)$ is at most $n^2 \log^2 n \cdot \poly(pk)$, and the time needed for $O(\log p)$ trials is $n^2 \log^2n \cdot \poly(pk)$ as well. Thus, the overall running time needed to find $\widehat{U^m}$ is $\nnz(A) + n^3 \log n \cdot \poly(pk)$, and the overall running time of Algorithm \ref{algorithm:robust_tucker_p_q_decomposition} is $O(p \cdot \nnz(A)) + n^3 \log n \cdot \poly(pk)$.
\end{proof}

\begin{remark}
The assumption that $p < q$ seems to be required for our technique to work. In other words, the factor $U^q$ must be an identity matrix.
\end{remark}

\newpage
\section{Additional Related Work} \label{sec:additional_related_work}
In this section, we discuss more related works mentioned earlier and compare them to our work in more detail.

The work of \cite{sun_lowranktucker} presents a family of algorithms that performs sketching based on Tucker-decompositions by applying dimensionality reduction to each factor matrix independently. They give a $2$-pass algorithm, proving a $2q$-approximation ratio with constant probability for rank-$k$ approximation for large enough $k$. 

 \cite{pmlr20_tensorized_random_proj} on CP decomposition and Tensor Train decomposition using randomized projections. There are similar prior results, such as \cite{TensorSketch}, that compute a subspace embedding for tensors. However, in many cases, this is too computationally expensive. For example, in \cite{pmlr20_tensorized_random_proj}, the cost of computing the embedding alone is $O(kNd\max(R, \tilde{R})^3)$, where $\tilde{R}$ is the actual best CP or TT rank of the tensor. The results are exponential in the number of modes.

\cite{MinsterRachel2020RAfL} builds upon Higher-Order SVD to perform lower rank approximation under the Tucker decomposition. The approximation error achieved was characterized by the discarded singular values of the matricizations of the tensor along different modes, which gives an $O(\sqrt{q-1})$ approximation.

\cite{oseledets2011tensor} defined Tensor Train decomposition and analyzed the TT-SVD algorithm. It incrementally finds an orthonormal bases for each factor matrix in the Tensor Train by finding the space for the top $k$ singular values and flattening along each mode.
\cite{huber2017randomized} modifies the original TT-SVD algorithm for Tensor Train decomposition studied in \cite{oseledets2011tensor} by performing a Gaussian sketch at each step. This introduces a $\sqrt{q-1}$ dependence in the error term, but is significantly more efficient. \cite{huber2017randomized} only characterizes their result with singular values of matricizations and numerical experiments, and does not show their algorithm gives a $(1+\varepsilon)$-approximation. 

\cite{faster_tt_sparse_2019} proposes an algorithm that uses a similar structure with TT-SVD and unlike most works, does achieve a $(1\pm\eps)$-relative error Tensor Train decomposition. However, their algorithm is very computationally expensive, and does not have a good guarantee of the rank of the output. This is because for each flattening, the algorithm sets $\delta = \frac{\eps}{\sqrt{d-1}}$, and finds a $\delta-$truncated SVD for that matricization. This requires searching for the target rank, which is computationally expensive.

\cite{handschuh2012changing} proposes algorithms to convert a tensor with low Tensor Train rank to a tensor with low Tensor Ring rank, and vice versa. This work also proposes an algorithm to convert a PEPS tensor network to a Tensor Train network. One key difference between our work and \cite{handschuh2012changing} is that we give a technique to convert an arbitrary tensor network (which can be structured according to any graph) to a tree of low rank. Our algorithm is different from that of \cite{handschuh2012changing} --- while they contract various edges of the network and stop when the network has become a tree, we contract the original graph $\calG$ until it is a single vertex, but additionally form two new vertices each time we contract an edge. Finally, the algorithm of \cite{handschuh2012changing} can increase the rank by a factor as large as $n^{O(q)}$ compared to our result.

Other traditional sketching and regression techniques, including leverage score sampling, have also been used to achieve better results in tensor approximation. \cite{fahrbach2021fast_ridge_regression} applies ridge leverage score sampling regression to the Tucker decomposition problem. This technique is based on the earlier work on Kronecker product regression \cite{djssw19_kronecker_product_regression}. However, \cite{fahrbach2021fast_ridge_regression} takes all the components of the Tucker decomposition as input and only solves for the core tensor.
 \cite{leverage_score_tensor_ring} uses leverage score sampling to obtain an approximation in the Tensor Ring format (a decomposition similar to Tensor Train, except that the first and the last mode are also contracted).

In the concurrent and independent work \cite{ms2022_tensor_network_embeddings_concurrent_and_independent},
the goal is to preserve the $\ell_2$ norm of a vector which has tensor network structure. \cite{ms2022_tensor_network_embeddings_concurrent_and_independent} gives embeddings which are tensor networks where each core has i.i.d. Gaussian entries. This is in the same spirit as our embeddings. For Tensor Train decomposition, the embeddings we use can be written as Tensor Trains where each core is formed from a Countsketch matrix, and for tree networks, our embeddings are tree networks (with the same structure as the desired tree network) whose cores are formed from Countsketch matrices. Our analysis is similar, as in \cite{ms2022_tensor_network_embeddings_concurrent_and_independent}, the embeddings are also decomposed into a sequence of sketching matrices, which are Kronecker products of a Gaussian matrix and an identity matrix. Our embeddings themselves are different however, as their embedding consists of a Kronecker product of Gaussian matrices (in order to reduce the dimension of the uncontracted edges of the input tensor) followed by a binary tree-structured network of Gaussian tensors. The structure of the binary tree is determined by a given contraction order for the input tensor. On the other hand, our embedding for tree networks is a tree network but is not a binary tree in general. \cite{ms2022_tensor_network_embeddings_concurrent_and_independent} also gives embeddings for general tensor networks which are not trees, while we do not consider this. We note that \cite{ms2022_tensor_network_embeddings_concurrent_and_independent} has the goal of obtaining a tensor network embedding which satisfies the conclusion of the Johnson Lindenstrauss lemma, and does not consider low rank approximation of tensors.


\bibliographystyle{alpha}
\bibliography{references}

\end{document}